\documentclass[a4paper,USenglish,cleveref,numberwithinsect]{lipics-v2021}

\title{Preprocessing for Outerplanar Vertex Deletion: An Elementary Kernel of Quartic Size}

\author{Huib Donkers}{Eindhoven University of Technology, The Netherlands}{h.t.donkers@tue.nl}{https://orcid.org/0000-0002-2767-8140}{}
\author{Bart M.\,P. Jansen 
}{Eindhoven University of Technology, The Netherlands}{b.m.p.jansen@tue.nl}{https://orcid.org/0000-0001-8204-1268}{}

\author{Micha\l{} W\l{}odarczyk
}{Eindhoven University of Technology, The Netherlands}{m.wlodarczyk@tue.nl}{https://orcid.org/0000-0003-0968-8414}{}

\authorrunning{H. Donkers, B.\,M.\,P. Jansen, M. W\l{}odarczyk}

\Copyright{Huib Donkers, Bart M.\,P. Jansen, Micha\l{} W\l{}odarczyk}



\funding{This project has received funding from the European Research Council (ERC) under the European Union's Horizon 2020 research and innovation programme (grant agreement No 803421, ReduceSearch).} 

\nolinenumbers 

\hideLIPIcs  

\EventEditors{John Q. Open and Joan R. Access}
\EventNoEds{2}
\EventLongTitle{42nd Conference on Very Important Topics (CVIT 2016)}
\EventShortTitle{CVIT 2016}
\EventAcronym{CVIT}
\EventYear{2016}
\EventDate{December 24--27, 2016}
\EventLocation{Little Whinging, United Kingdom}
\EventLogo{}
\SeriesVolume{42}
\ArticleNo{23}

 \usepackage{xspace}

\newtheorem{reduction}{Reduction Rule}
\newtheorem*{thmstar}{Theorem}
\newtheorem*{corstar}{Corollary}
\newtheorem{prop}{Proposition}

\newcommand{\brb}[1]{\langle#1\rangle}
\newcommand{\eps}{\varepsilon}

\newcommand{\Oh}{\mathcal{O}}

\newcommand{\setN}{\ensuremath{\mathbb{N}}}

\newcommand{\F}{\ensuremath{\mathcal{F}}\xspace}

\newcommand{\opvdfull}{\textsc{Outerplanar Deletion}\xspace}
\newcommand{\op}{outerplanar\xspace}
\newcommand{\opd}{\mathsf{opd}}
\newcommand{\bcc}{6288}

\newcommand{\NP}{\ensuremath{\mathsf{NP}}\xspace}
\newcommand{\NPhard}{\NP-hard\xspace}

\newcommand{\coNPpoly}{\ensuremath{\mathsf{coNP/poly}}}

\newcommand{\notcontainment}{\ensuremath{\NP \not\subseteq \coNPpoly}\xspace}
\newcommand{\containment}{\ensuremath{\NP \subseteq \coNPpoly}\xspace}

\theoremstyle{claimstyle}
\newtheorem{countclaim}{Claim}
\newtheorem{countobs}[countclaim]{Observation}

\ifdefined\DEBUG{}
\def\rem#1{{\marginpar{\raggedright\scriptsize #1}}}

\newcommand{\mic}[1]{{\color{blue}{#1}}}
\newcommand{\micr}[1]{\rem{\textcolor{blue}{\(\bullet \) #1}}}

\newcommand{\bmp}[1]{{\color{purple}{#1}}}
\newcommand{\bmpr}[1]{\rem{\textcolor{purple}{\(\bullet \) #1}}}

\newcommand{\hui}[1]{{\color{orange}{#1}}}
\newcommand{\huir}[1]{\rem{\textcolor{orange}{\(\bullet \) #1}}}

\else
\newcommand{\mic}[1]{#1}
\newcommand{\bmp}[1]{#1}
\newcommand{\hui}[1]{#1}

\newcommand{\micr}[1]{}
\newcommand{\bmpr}[1]{}
\newcommand{\huir}[1]{}
\fi

\begin{document}



\maketitle{}

\begin{abstract}
In the {\sc \F-Minor-Free Deletion} problem 
one is given an undirected graph $G$, an integer $k$, and the task is to determine whether there exists a vertex set $S$ of size at most $k$, so that $G-S$ contains no graph from the finite family \F as a minor.
It is known that whenever \F contains at least one planar graph, then {\sc \F-Minor-Free Deletion} admits a polynomial kernel, that is, there is a polynomial-time algorithm that outputs an equivalent instance of size $k^{\Oh(1)}$ [Fomin, Lokshtanov, Misra, Saurabh; FOCS 2012].
However, this result relies on non-constructive arguments based on well-quasi-ordering and does not provide a concrete bound on the kernel size.

We study the \textsc{Outerplanar Deletion} problem,
in which we want to remove at most $k$ vertices from a graph to make it outerplanar.
This is a special case of {\sc \F-Minor-Free Deletion} for
 the family $\F = \{K_4, K_{2,3}\}$.
The class of outerplanar graphs is arguably the simplest class of graphs for which no explicit kernelization size bounds are known. 
By exploiting the combinatorial properties of outerplanar graphs we present elementary reduction rules decreasing the size of a graph.
This yields a constructive kernel with $\Oh(k^4)$ vertices and edges. As a corollary, we derive that any minor-minimal obstruction to having an outerplanar deletion set of size~$k$ has~$\Oh(k^4)$ vertices and~edges.

\includegraphics[scale=0.1]{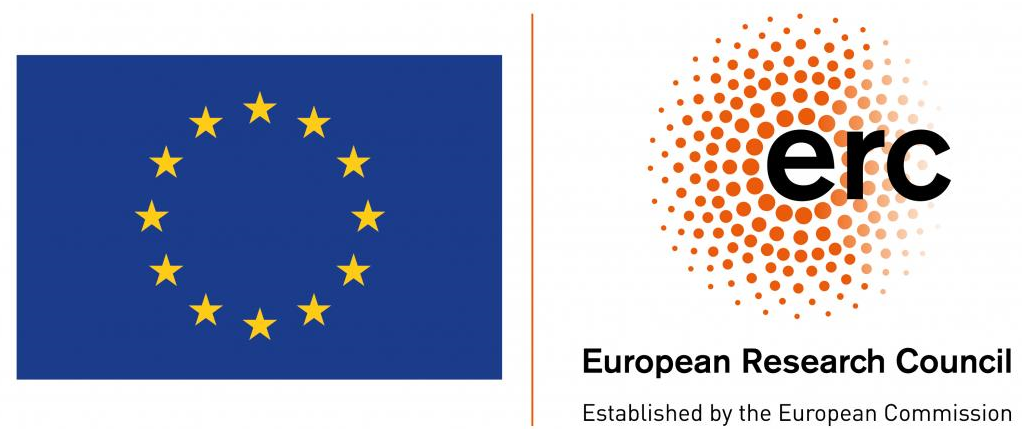}
\end{abstract}

\section{Introduction}
\subparagraph*{Background and Motivation}
Kernelization~\cite{FominLSZ19} is a subfield of parameterized complexity~\cite{cygan2015parameterized,DowneyF13} that investigates the complexity of preprocessing \NPhard problems. A parameterized problem includes in its input an integer~$k$ which we call the parameter. This parameter \hui{can be seen as} a measure of complexity of the problem input. A common choice is to treat the size of the desired solution as the parameter. A kernelization is a polynomial-time preprocessing algorithm that converts a problem instance with parameter~$k$ into an equivalent parameterized instance of the same problem such that both the size and the parameter value of the new instance are bounded by a function~$f$ of~$k$. The function~$f$ is called the size of the kernel. 
It is known that a decidable parameterized problem has a kernel if and only if it is fixed-parameter tractable~\cite[Lemma 2.2]{cygan2015parameterized}. A major challenge is to determine which parameterized problems admit a kernel of polynomial size.

One class of problems that received much attention~\cite{FominLMPS16,FominLMS12,Giannopoulou2017,JansenP20,JoretPSST14} is {\sc \F-Minor-Free Deletion}. For a~fixed finite family of graphs~\F, the {\sc \F-Minor-Free Deletion} problem asks, given a graph~$G$ and parameter~$k$, whether a vertex set~$S \subseteq V(G)$ of size~$k$ exists such that the graph~$G-S$, obtained from~$G$ by removing the vertices in~$S$, does not contain any graph~$F \in \F$ as a minor. This class of problems includes a large variety of well-studied problems such as {\sc Vertex Cover}, {\sc Feedback Vertex Set}, and {\sc Planarization}, which are obtained by taking~\F equal to (respectively)~$\{K_2\}$, $\{K_3\}$, and~$\{K_5, K_{3,3}\}$. All of the {\sc \F-Minor-Free Deletion} problems are fixed-parameter tractable~\cite{SauST20}, but it is unknown whether they all admit a polynomial kernel~\cite{FominLMS12}. If each graph in \F contains at least one edge, it follows from the general results of Lewis and Yannakakis~\cite{LewisY80} that {\sc \F-Minor-Free Deletion} is \NPhard.

If \F is restricted to only families containing a planar graph we speak of {\sc Planar-\F Deletion}. \bmp{Since the family of~$\mathcal{F}$-minor-free graphs has bounded treewidth if and only if~$\mathcal{F}$ includes a planar graph~\cite{RobertsonS86}, this restriction ensures that removing a solution to the problem yields a graph of constant treewidth. Hence any solution is a treewidth-$\eta$ modulator for some~$\eta \in \mathbb{N}$ depending on~$\mathcal{F}$.} For this more restricted class Fomin et al.~\cite{FominLMS12} have shown that polynomial kernels exist for each choice of~\F. However, the running time of this kernelization algorithm is described by the authors as ``horrendous'' and regarding the size the authors state the following:

\begin{quote}
The size of the kernel, however, is not explicit. Several of the constants that go into the proof of Lemma 29 depend on the size of the largest graph in certain antichains in a well-quasi-order and thus we don't know what the (constant) exponent bounding the size of the kernel is. We leave it to future work to make also the size of the kernel explicit.
\end{quote}

For some specific {\sc Planar-\F Deletion} problems kernels with explicit size are known. Most famous are {\sc Vertex Cover} and {\sc Feedback Vertex Set} which admit kernels with respectively a linear and quadratic number of vertices~\cite{ChenKJ01,Iwata17,Thomasse10}. Additionally, if~$\theta_c$ denotes the graph with two vertices and~$c \geq 1$ parallel edges, then {\sc $\{\theta_c\}$-Minor-Free Deletion} admits a kernel with~$\Oh(k^2 \log^{3/2} k)$ vertices and edges~\cite[Theorem 1.2]{FominLMPS16}; note that the cases~$c=1$ and~$c=2$ correspond to {\sc Vertex Cover} and {\sc Feedback Vertex Set}. 
Another \bmp{problem for which an explicit kernel size \bmp{bound} is known is \textsc{Pathwidth-one Deletion}}, where the goal is to obtain a graph of pathwidth one, i.e, each connected component is a caterpillar. First a kernel of quartic size was obtained~\cite{PhilipRV10} which was later improved to a quadratic kernel~\cite{CyganPPW10}.
If we want to remove at most $k$ vertices to obtain a graph of treedepth at most $\eta$, we obtain the {\sc Treedepth-$\eta$ Deletion} problem.
Since this property can be characterized by forbidden minors and bounded treedepth implies bounded treewidth,
this problem is also a special case of {\sc Planar-\F Deletion}.
Giannopoulou et al.~\cite{Giannopoulou2017} have shown
 that for every $\eta$, there is a kernel \bmp{with} $2^{\Oh(\eta^2)}\cdot k^6$ \bmp{vertices} for {\sc Treedepth-$\eta$ Deletion}.
They have also proven that in general there is no hope for a universal constant in the kernel exponent and the
 degree of the polynomial \bmp{which bounds the kernel size} must increase as a function of~\F unless \notcontainment.

In this paper we investigate \opvdfull, which asks for a graph~$G$ and parameter~$k$ whether a set~$S \subseteq V(G)$ of size~$k$ exists such that~$G-S$ is \op. A graph is \op if it admits a planar embedding for which all vertices lie on the outer face, or equivalently, if it contains neither~$K_4$ nor~$K_{2,3}$ as a minor. Outerplanar graphs form a rich superclass of forests and are frequently studied in graph theory~\cite{ChartrandH67,CoudertHS07,DingD16,FleischnerGH74,Syslo79}, graph drawing~\cite{Biedl11,Frati12,MchedlidzeS11}, and optimization~\cite{GiacomoLM16,Leydold1998MinimalCB,MorganF07,Poranen05}. 

Since outerplanarity can be characterized as being~$\{K_4,K_{2,3}\}$-minor-free~\cite{ChartrandH67}, the problem belongs to the class of {\sc Planar-\F Deletion} problems. It is \hui{arguably} the easiest problem in the class for which no \hui{explicit} polynomial kernel is known. This makes \opvdfull a well-suited starting point to deepen our understanding of {\sc Planar-\F Deletion} problems in the search for explicit kernelization bounds.

\subparagraph*{Results}
Let~$\opd(G)$ denote the minimum size of a vertex set~$S \subseteq V(G)$ such that~$G-S$ is outerplanar. Our main result is the following theorem:

\newcommand{\thmKernelCommand}{
The \opvdfull{} problem admits a polynomial-time kernelization algorithm that, given an instance $(G,k)$, outputs an equivalent instance $(G',k')$, such that $k' \le k$, graph~$G'$ is a minor of $G$, and $G'$ has $\Oh(k^4)$ vertices and edges. Furthermore, if~$\opd(G) \leq k$, then~$\opd(G') = \opd(G) - (k - k')$.
}

\begin{theorem} \label{thm:kernel}
\thmKernelCommand
\end{theorem}

The algorithm behind \cref{thm:kernel} is elementary, consisting of a subroutine to build a decomposition of the input graph~$G$ using marking procedures in a tree decomposition, together with a series of explicit reduction rules. In particular, we avoid the use of protrusion replacement (summarized below). Concrete bounds on the hidden constant in the~$\Oh$-notation follow from our arguments. 
The size bound depends on the approximation ratio of an approximation algorithm that bootstraps the decomposition phase, for which the current state-of-the-art is 40. We will therefore present a formula to obtain a concrete bound on the kernel size, rather than its value using the current-best approximation (which would exceed~$10^5$). 

\bmp{\Cref{thm:kernel} presents} the first concrete upper bound on the degree of the polynomial that bounds the size of kernels for \opvdfull{}. We hope that it will pave the way towards obtaining explicit size bounds for all {\sc Planar-\F Deletion} problems and give an impetus for research on the kernelization complexity of the \textsc{Planar Deletion} problem, which is one of the major open problems in kernelization today~\cite[4:28]{OpenProblems},\cite[Appendix A]{FominLSZ19}.

Via known connections~\cite{FominLMS12} between kernelizations that reduce to a minor of the input graph and bounds on the sizes of obstruction sets, we obtain the following corollary.

\newcommand{\corollaryBoundsCommand}{If~$G$ is a graph such that~$\opd(G) > k$ but each proper minor $G'$ of~$G$ satisfies~$\opd(G') \leq k$, then~$G$ has~$\Oh(k^4)$ vertices and edges.}

\begin{corollary} \label{cor:obstruction:bounds}
\corollaryBoundsCommand
\end{corollary}

The existence of a polynomial bound with unknown degree follows from the work of Fomin et al.~\cite{FominLMS12}; \cref{cor:obstruction:bounds} gives the first explicit size bounds and contributes to a large body of research on minor-order obstructions (e.g.~\cite{CattellDDFL00,Dinneen97,DinneenCF01,DinneenX02,Lagergren98,RueST12,SauST21}).

\subparagraph*{Techniques}
\mic{The known kernelization algorithms~\cite{FominLMPS16,FominLMS12}
for {\sc Planar-\F Deletion} make use of (near-)protrusions}. A protrusion is a vertex set that induces a subgraph of \bmp{constant} treewidth and boundary size. \hui{Protrusion replacement is a technique where sufficiently large protrusions \bmp{are} replaced by smaller ones without changing the answer.}
Protrusion techniques were first used to obtain kernels for problems on planar and other topologically-defined graph classes~\cite{BodlaenderFLPST16}. Later Fomin at al.~\cite{FominLMPS16} described how to use protrusion techniques for problems on general graphs. \bmp{They proved~\cite[Lemma 3.3]{FominLMPS16} that any graph~$G$, which contains a modulator~$X$ to constant treewidth such that~$|X|$ and the size of its neighborhood can be bounded by a polynomial in~$k$, contains a protrusion of size~$|V(G)|/k^{\Oh(1)}$ that can be found efficiently.} For any fixed~\F containing a planar graph, they present a method to obtain a small modulator to an \F-minor-free graph, which has constant treewidth. 
This leads to a polynomial kernel for {\sc Planar-\F Deletion} on graphs with bounded degree since the size of the neighborhood of the modulator can be bounded so protrusion replacement can be used to obtain a polynomial kernel.
Specifically for {\sc $\{\theta_c\}$-Minor-Free Deletion} they give reduction rules to reduce the maximum degree in a general graph, which leads to a polynomial kernel on general graphs.

The kernel for {\sc Planar-\F Deletion} given by Fomin et al.~\cite{FominLMS12} does not rely on bounding the size of the neighborhood of the modulator followed by protrusion replacement. Instead they present the notion of a near-protrusion: a vertex set that will become a protrusion after removing any size-$k$ solution from the graph. With an argument based on well-quasi-ordering they determine that if such near-protrusions are large enough one can, in polynomial time, \bmp{reduce to a proper minor of the graph} without changing the answer.

In this paper we present a method for \opvdfull to decrease the size of the neighborhood of a modulator to outerplanarity. This relies on a process that was called ``tidying the modulator'' in earlier work~\cite{BevernMN12} and also used in the kernelization for {\sc Chordal Vertex Deletion}~\cite{JansenP18}. The result is a larger modulator~$X \subseteq V(G)$ but with the additional feature that it retains its modulator properties when omitting any single vertex, that is,~$G - (X \setminus \{x\})$ is outerplanar for each~$x \in X$. We proceed by decomposing the graph into near-protrusions, \bmp{following along similar lines as the decomposition by Fomin et al.~\cite{FominLMPS16} but exploiting the structure of outerplanar graphs at several steps to obtain such a decomposition with respect to our larger tidied modulator, without leading to worse bounds.} With the additional properties of the modulator~$X$ obtained from tidying we no longer need to rely on well-quasi-ordering, but instead are able to reduce the size of the neighborhood of the modulator in two steps. \bmp{The first reduces the number of connected components of~$G - X$ which are adjacent to any particular modulator vertex~$x \in X$}. In the case of $\{\theta_c\}$-minor-free graphs, if~$G - (X \setminus \{x\})$ is $\{\theta_c\}$-minor-free then bounding the number of components of~$G-X$ adjacent to each~$x \in X$ this is sufficient to bound~$|N_G(X)|$, since any~$x \in X$ has less than~$c$ neighbors in any component of~$G - (X \setminus \{x\})$. \bmp{One of the major difficulties we face when working with $\{K_{2,3}\}$-minor-free graphs is that in such a graph there can be arbitrarily many edges between a vertex~$x$ and a \hui{connected component} of~$G - (X \setminus \{x\})$. Therefore we present an additional reduction rule that reduces, in a second step, the number of edges between a vertex and a connected component. After these two steps we obtain a bound on the size of the neighborhood of the modulator.}
%
\bmp{At this point, standard protrusion replacement could be applied to prove the \emph{existence} of a kernel for \opvdfull with~$\Oh(k^4)$ vertices.} In order to give an explicit kernelization algorithm we present a number of additional reduction rules to \bmp{avoid} the generic protrusion replacement technique. This eventually leads to a kernel with at most~\mic{$c \cdot k^4$ vertices and edges for \opvdfull. It is conceptually simple (yet tedious) to extract the explicit value of~$c$ from the algorithm description.}

\subparagraph*{Organization}
In the next section we give \mic{basic} definitions and notation we use throughout the rest of the paper, together with structural \mic{observations} for \op graphs.  \Cref{sec:modulators} describes how we obtain small modulators to outerplanarity with progressively stronger properties, and finally we obtain a
\mic{modulator of size $\Oh(k^4)$ such that each remaining component has only 4 neighbors in the modulator, \bmp{effectively forming a decomposition into protrusions}}.
The second stage of the kernelization reduces the size \bmp{of} \hui{the connected} components \bmp{outside the modulator}. These reduction rules  are described in \cref{sec:reduction-rules}. In \cref{sec:wrapup} we finally tie everything together to obtain a kernel with~$\Oh(k^4)$ vertices and edges.

\section{Preliminaries} \label{sec:prelims}

\subparagraph*{Approximation and kernelization}

Let $A$ be a minimization problem and let $\mathsf{OPT}(I)$ denote the \bmp{minimum} cost of a solution to an instance $I$. 
For a constant $\alpha > 1$,
an~$\alpha$-approximation algorithm for $A$ is an algorithm that, given an instance $I$, outputs a solution of cost at most $\alpha \cdot \mathsf{OPT}(I)$.

A~parameterized problem is a~decision problem in which every input has an~associated positive integer parameter that captures its complexity in some well-defined way. For a~parameterized problem~$A \subseteq \Sigma^* \times \mathbb{N}$ and a~function~$f \colon \mathbb{N} \to \mathbb{N}$, a~kernelization for~$A$ of~size~$f$ is an~algorithm that, on input~$(x,k) \in \Sigma^* \times \mathbb{N}$, takes time polynomial in~$|x| + k$ and outputs~$(x',k') \in \Sigma^* \times \mathbb{N}$ such that the following holds:
\begin{enumerate}
    \item $(x,k) \in A$ if and only if~$(x',k') \in A$, and
    \item both $|x'|$ and~$k'$ are bounded by~$f(k)$.    
\end{enumerate}

\subparagraph*{Graph theory}

The set $\{1,\ldots,p\}$ is denoted by $[p]$.
We consider simple undirected graphs without self-loops. A~graph $G$ has vertex set $V(G)$ and edge set $E(G)$. We use shorthand $n = |V(G)|$ and $m = |E(G)|$. For (\bmp{not necessarily disjoint}) $A,B \subseteq V(G)$, we define $E_G(A,B) = \{uv \mid u \in A, v \in B, uv \in E(G)\}$.
 The {open} neighborhood of $v \in V(G)$ is $N_G(v) := \{u \mid uv \in E(G)\}$, where we omit the subscript $G$ if it is clear from context. {For a vertex set~$S \subseteq V(G)$ the open neighborhood of~$S$, denoted~$N_G(S)$, is defined as~$\bigcup _{v \in S} N_G(v) \setminus S$. The closed neighborhood of a single vertex~$v$ is~$N_G[v] := N_G(v) \cup \{v\}$, and the closed neighborhood of a vertex set~$S$ is~$N_G[S] := N_G(S) \cup S$.} 
\mic{The boundary of a vertex set $S \subseteq V(G)$ is the set $\partial_G(S) = N_G(V(G) \setminus S)$.}
For $A \subseteq V(G)$, the graph induced by $A$ is denoted by $G[A]$ and {we say that the vertex set $A$ is connected if the graph $G[A]$ is connected.}
We use notation $G\brb A = G[N_G[A]]$ 
\mic{and, when $H$ is an induced subgraph of $G$, we write briefly $G\brb H = G \brb {V(H)}$ or $\partial_G(H) = \partial_G(V(H))$.}
We use shorthand $G-A$ for the graph $G[V(G) \setminus A]$. For $v \in V(G)$, we write $G-v$ instead of $G-\{v\}$.
\hui{For~$A \subseteq E(G)$ we denote by~$G \setminus A$ the graph with vertex set~$V(G)$ and edge set~$E(G) \setminus A$. For~$e \in E(G)$ we write~$G \setminus e$ instead of~$G \setminus \{e\}$.}
If $e = uv$, then $V(e) = \{u,v\}$.

{A tree is a connected graph that is acyclic. A~forest is a disjoint union of trees. In tree $T$ with root $r$, we say that $t \in V(T)$ is an ancestor of $t' \in V(T)$ (equivalently $t'$ is a descendant of $t$)} {if $t$ lies on the (unique) path from $r$ to $t'$.}
For two disjoint sets $X,Y \subseteq V(G)$, we say that $S \subseteq V(G) \setminus (X \cup Y)$ is an $(X,Y)$-separator if the graph $G-S$ does not contain any~path from any $u\in X$ to any $v\in Y$. 
By Menger's theorem, if $x,y \in V(G)$ are non-adjacent in $G$ then the size of \bmp{a} minimum $(x,y)$-separator is equal to the maximum number of internally vertex-disjoint paths from $x$ to $y$.
A vertex $v \in V(G)$ is an articulation point in a connected graph $G$ if $G-v$ is not connected.
A graph is called biconnected if it has no articulation points.
A biconnected component in $G$ is an inclusion-wise maximal \mic{subgraph which is biconnected}.
\bmp{A graph~$G$ is 2-connected if it is biconnected and has at least three vertices. In a 2-connected graph, for every pair of vertices $u,v \in V(G)$ there exists a cycle going through both $u$ and~$v$.} 
The structure of the biconnected components and articulation points in a connected graph $G$ is captured by a tree called the block-cut tree.
It has a vertex for each biconnected component and for each articulation point in $G$. 
A biconnected component $B$ and an articulation point $v$
are connected by an edge if \mic{$v \in V(B)$.}

A vertex set $A \subseteq V(G)$ is an independent set in $G$ if $E_G(A,A) = \emptyset$.
A graph \bmp{$G$} is bipartite if there is a partition of \bmp{$V(G)$} into two independent sets $A, B$.
We write shortly $G = (A \cup B, E)$ to specify \bmp{a bipartite graph on vertex set~$E = E(G)$ admitting this partition.}
\begin{definition}
For a vertex set $X \subseteq V(G)$ the \emph{component graph} $\mathcal{C}(G, X)$ is a bipartite graph $(X \cup Y, E)$, where $Y$ is the set of connected components of $G - X$, and $(v,C) \in E$ if there is at least one edge between $v \in X$ and the component $C \in Y$.
\end{definition}
For an integer $q$, the graph $K_q$ is the complete graph on $q$ vertices.
For integers $p,q$, the graph $K_{p,q}$ is the bipartite graph $(A \cup B, E)$, where $|A|=p$, $|B|=q$,
and $uv \in E$ whenever $u \in A, v \in B$.

\subparagraph*{Minors}
A contraction of $uv \in E(G)$ introduces a new vertex adjacent to all of {$N_G(\{u,v\})$}, after which $u$ and $v$ are deleted. The result of contracting $uv \in E(G)$ is denoted $G / uv$. For $A \subseteq V(G)$ such that $G[A]$ is connected, we say we contract $A$ if we simultaneously contract all edges in $G[A]$ and introduce a single new vertex.
We say that $H$ is a minor of $G$, if we can turn $G$ into $H$ by a (possibly empty) series of edge contractions, edge deletions, and vertex deletions.
\mic{If this series is non-empty, then $H$ is called a proper minor of $G$.}
We can represent the result of such a process with a mapping $\phi \colon V(H) \to 2^{V(G)}$, such that subgraphs $(G[\phi(h)])_{h\in V(H)}$ are connected and vertex-disjoint{, with an edge of~$G$ between a vertex in~$\phi(u)$ and a vertex in~$\phi(v)$ for all~$uv \in E(H)$.}
The sets $\phi(h)$ are called branch sets and the family $(\phi(h))_{h\in V(H)}$ is called a minor-model of $H$~in~$G$.

\subparagraph*{Planar and outerplanar graphs}
A plane embedding of graph $G$ is given by a mapping from $V(G)$ to $\mathbb{R}^2$ and a mapping that associates with each edge $uv \in E(G)$ a simple curve on the plane connecting the images of $u$ and $v$, such that the curves given by two distinct edges \bmp{can intersect only at the image of a vertex that is a common endpoint of both edges}.
A face in a plane embedding of a graph $G$ is a subset of the plane enclosed by images of some subset of the edges.
We say that a vertex $v$ lies on a face $f$ if the image of $v$ belongs to the closure of $f$.
In every plane embedding there is exactly one face of infinite area, referred to as the outer face.
Let $F$ denote the set of faces in a plane embedding of \bmp{$G$}.
Then Euler's formula states that $|V(G)| - |E(G)| + |F| = 2$. 
Given a plane embedding of $G$ we define the dual graph $\widehat{G}$ with $V(\widehat{G}) = F$ and edges given by pairs of distinct faces that are incident to an image of a common edge from $E(G)$.
A weak dual graph is obtained from the dual graph by removing the vertex created in place of the outer face.

A graph is called planar if it admits a plane embedding.
By Wagner's theorem, a graph $G$ is planar if and only if $G$ contains neither $K_5$ nor $K_{3,3}$ as a minor.
A graph is called \op if it admits a plane embedding with all vertices lying on the outer face.
A graph $G$ is \op if and only if $G$ contains neither $K_4$ nor $K_{2,3}$ as a minor~\cite{ChartrandH67}.
If a graph $G$ is planar (resp. \op) and $H$ is a minor of $G$, then $H$ is also planar (resp. \op).
\mic{The weak dual of an embedded biconnected \op graph $G$ is either an~empty graph, if $G$ is a single edge or vertex, or a tree otherwise~\cite{FleischnerGH74}.}
A graph $G$ is planar (resp. \op) if and only if every biconnected component in $G$ induces a planar (resp. \op) graph.

\begin{observation}
\label{criterion:articulation-point}
Let $v \in V(G)$. The graph $G$ is \op if and only if
for each connected component $C$ of $G - v$ the graph $G \brb C$ is \op.
\end{observation}

For a graph $G$ we call $S \subseteq V(G)$ an \op deletion set if $G-S$ is \op.
The \op deletion number of $G$, denoted $\opd(G)$, is the size of a smallest \op deletion set in $G$.

\subparagraph*{Structural properties of outerplanar graphs} \label{sec:op-struct}

We present a number of structural observations of \op graphs which will be useful in our later argumentation. The first is a characterization of \op graphs similar to \cref{criterion:articulation-point}. Rather than looking at the components of a graph with one vertex removed, it considers the components of a graph with both endpoints of an edge removed. This allows us for example to easily argue about outerplanarity of graphs obtained from ``gluing'' two outerplanar graphs on two adjacent vertices. \bmp{Recall that~$G\brb C =  G[N_G(V(C))]$.}

\begin{lemma}\label{criterion:edge-removal}
Let~$G$ be a graph and $e \in E(G)$. Then~$G$ is \op if and only if both of the following conditions hold:
\begin{enumerate}
    \item \label{criterion:edge-removal:connected-comps}
    for each connected component $C$ of $G - V(e)$ the graph $G \brb C$ is \op, and
    
    \item \label{criterion:edge-removal:induced-paths}
    the graph~$G \setminus e$ does not have three induced internally vertex-disjoint paths connecting the endpoints of~$e$.
\end{enumerate}
\end{lemma}
\begin{proof}
($\Rightarrow$) Suppose~$G$ is outerplanar. Then every subgraph of~$G$ is outerplanar, showing the first condition holds. If~$G \setminus e$ has three induced internally vertex-disjoint paths connecting the endpoints of~$e = xy$, then each path has at least one interior vertex which shows that~$G$ has a~$K_{2,3}$-minor, contradicting outerplanarity of~$G$.

($\Leftarrow$) Suppose the two conditions hold, and suppose for a contradiction that~$G$ is not outerplanar. Then~$G$ contains~$K_4$ or~$K_{2,3}$ as a minor. We consider the two cases separately.

\proofsubparagraph*{$G$ has a $K_{2,3}$-minor} Suppose that~$G$ contains~$K_{2,3}$ as a minor. It is easy to see that there exist two vertices~$u,v \in V(G)$ and three disjoint connected vertex sets~$A_1, A_2, A_3$ such that~$A_i$ contains a vertex of both~$N_G(u)$ and~$N_G(v)$ for all~$i \in [3]$. Let~$G_i$ be the graph obtained from~$G[\{u,v\} \cup A_i]$ by removing the edge~$uv$, if it exists. There exists an $(u,v)$-path in~$G_i$, so by taking a shortest path there exists an \emph{induced} $(u,v)$-path~$P_i$ in~$G_i$. Since the edge~$uv$ does not belong to~$G_i$, path~$P_i$ has at least one interior vertex. The three $(u,v)$-paths~$P_1, P_2, P_3$ in~$G$ obtained in this way are internally vertex-disjoint, have at least one interior vertex, and are induced after removing the edge~$uv$ if it exists. We use this to derive a contradiction.

If the edge~$uv$ exists and is equal to~$e$, then the existence of~$P_1,P_2,P_3$ shows that the second condition is violated and leads to a contradiction. So in the remainder, we may assume that~$e \neq uv$. Hence at least one vertex of~$\{u,v\}$ lies in a connected component~$C$ of~$G - V(e)$. Assume without loss of generality that~$u \notin V(e)$ and~$u$ lies in~$V(C)$. We show that~$v \in V(G \brb C)$ in this case. Suppose that~$v$ does not belong to~$G \brb C$. Then in particular~$v \notin V(e)$ and the vertices~$V(e)$ separate~$u$ from~$v$; but since~$P_1,P_2,P_3$ are three internally vertex-disjoint paths, vertices~$u$ and~$v$ cannot be separated by the set~$V(e)$ of two vertices. It follows that~$u,v \in V(G \brb C)$.

We claim that each path~$P_i$ is a subgraph of~$G \brb C$. To see this, note that the path starts and ends in~$G \brb C$. The two vertices~$V(e)$ are the only vertices of~$G \brb C$ which have neighbors in~$G$ outside~$G \brb C$. So a path starting and ending in~$G \brb C$ has to leave~$G \brb C$ at one vertex of~$V(e)$ and enter~$G \brb C$ at the other; but then~$e$ is a chord of this path other than~$uv$. Since the paths~$P_i$ do not have such chords, it follows that each path~$P_i$ is a subgraph of~$G \brb C$.

By the above, the graph~$G \brb C$ contains three internally vertex-disjoint paths~$P_1,P_2,P_3$ with at least one interior vertex each. But then~$G \brb C$ contains~$K_{2,3}$ as a minor and is not outerplanar; a contradiction to the first condition.

\proofsubparagraph*{$G$ has a $K_4$-minor} In the remainder, we may assume that~$G$ contains~$K_4$ as a minor but does not contain~$K_{2,3}$ as a minor, as otherwise the previous case applies. Observe that this means that~$G$ contains~$K_4$ as a subgraph: any subdivision of~$K_4$ leads to a~$K_{2,3}$ minor.

So let~$H$ be a~$K_4$ subgraph in~$G$. Observe that there cannot be two connected components of~$G - V(e)$ that both contain a vertex of~$H$: any two vertices of the clique~$H$ are connected by an edge, which merges the connected components. So there is one connected component~$C$ of~$G - V(e)$ that contains all vertices of~$V(H) \setminus V(e)$. But then~$H$ is a subgraph of~$G \brb C$, proving that~$G \brb C$ is not outerplanar and contradicting the first condition.
\end{proof}

In order to more easily apply \cref{criterion:edge-removal}, we show that no two induced paths as referred to in \cref{criterion:edge-removal}\eqref{criterion:edge-removal:induced-paths} can lie in the same connected component~$C$ as referred to in \cref{criterion:edge-removal}\eqref{criterion:edge-removal:connected-comps}.

\begin{lemma}\label{lem:paths-in-different-components}
 Suppose~$G$ is \op with an edge~$uv \in E(G)$. If~$P_1,P_2$ are internally vertex-disjoint $(u,v)$-paths in~$G \setminus uv$, then the interiors of~$P_1$ and~$P_2$ lie in different connected components of~$G-\{u,v\}$.
\end{lemma}
\begin{proof}
 Suppose for contradiction that the interiors of~$P_1$ and~$P_2$ are in the same connected component of~$G-\{u,v\}$, and let~$P$ be a path from~$V(P_1)$ to~$V(P_2)$ in~$G - \{u,v\}$. Let~$G'$ be the graph obtained from~$G$ by contracting the interiors of~$P_1$ and~$P_2$ into a single vertex~$p_1$ and~$p_2$ respectively and contracting~$P$ to realize the edge~$p_1p_2$. Clearly~$G'$ is a minor of~$G$ so then~$G'$ doesn't contain a $K_4$-minor. Observe however that~$\{u,v,p_1,p_2\}$ induce a~$K_4$ subgraph in~$G'$. Contradiction.
\end{proof}

We now give a condition under which an edge can be added to an outerplanar without violating outerplanarity. Intuitively, this corresponds to adding an edge between two vertices that lie on the same interior face.

\begin{lemma}\label{lem:induced-cycle}
Suppose $G$ is \op and vertices $x,y$ lie on an induced cycle~$D$ with~$xy \notin E(G)$.
Then adding the edge $xy$ to $G$ preserves outerplanarity.
\end{lemma}
\begin{proof}
Let~$D_1, D_2$ be the two parts of the cycle~$D - \{x,y\}$. We claim that~$D_1$ and~$D_2$ belong to different connected components of~$G - \{x,y\}$. Suppose not, and let~$P$ be a path from~$V(D_1)$ to~$V(D_2)$ in~$G - \{x,y\}$ that intersects~$V(D_1)$ and~$V(D_2)$ in exactly one vertex~$v_1$ and~$v_2$, respectively. The path~$P$ has at least one interior vertex since the cycle~$D$ is induced. But then~$P$ together with the two induced~$(v_1,v_2)$-paths along~$D$ give a~$K_{2,3}$-minor; a contradiction to the assumption that~$G$ is outerplanar.

Hence~$D_1$ and~$D_2$ belong to different connected components of~$G - \{x,y\}$. Let~$G'$ be the graph obtained from~$G$ by adding the edge~$xy$. We show that for each connected component~$C$ of~$G' - \{x,y\}$ the graph~$G' \brb C$ is a minor of~$G$ and therefore outerplanar. This follows from the fact that, by the argument above,~$C$ contains at most one segment of the cycle~$D$ and therefore we can contract the remaining segment to realize the edge~$xy$.

Using the above, we prove that~$G'$ is \op by applying \cref{criterion:edge-removal} to edge~$xy$. The preceding argument shows that the first condition is satisfied. To see that the second condition is satisfied as well, note~$G$ is \op and therefore~$G' \setminus xy = G$ does not contain three internally vertex-disjoint paths connecting the endpoints of~$e$.
\end{proof}

Finally, we observe that if an \op graph~$G$ has a cycle~$C$, then any component of~$G-V(C)$ is adjacent to at most two vertices of the cycle (else there would be a $K_4$ minor), and these must be consecutive on the cycle (else there would be a $K_{2,3}$ minor).

\begin{lemma} \label{lem:attaching:to:cycle}
If~$C$ is a cycle in an outerplanar graph~$G$, then each connected component of~$G - V(C)$ has at most two neighbors in~$C$, and they must be consecutive along the cycle.
\end{lemma}
\begin{proof}
Suppose for a contradiction that some component~$D$ of~$G - V(C)$ has two neighbors~$x,y$ which are not consecutive along~$C$. Then the cycle provides two vertex-disjoint $(x,y)$-paths with at least one interior vertex each, and component~$D$ provides a third $(x,y)$-path with an interior vertex. This yields a~$K_{2,3}$-minor where~$\{x\}$ and~$\{y\}$ are the branch sets of the degree-3 vertices, contradicting outerplanarity.

Now suppose that some component~$D$ of~$G - V(C)$ has three or more neighbors on~$C$. Let~$P_1,P_2,P_3$ be three paths that cover he entire cycle~$C$ such that each path contains a neighbor of~$D$ and observe that~$V(P_1), V(P_2), V(P_3), V(D)$ form the branch sets of a $K_4$-minor in~$G$, contradicting outerplanarity.
\end{proof}

\subparagraph*{Treewidth and the LCA closure}
A tree decomposition of graph $G$ is a pair $(T, \chi)$ where~$T$ is a rooted tree, and~$\chi \colon V(T) \to 2^{V(G)}$, such that:
\begin{enumerate}
    \item For each~$v \in V(G)$ the nodes~$\{t \mid v \in \chi(t)\}$ form a {non-empty} connected subtree of~$T$. 
    \item For each edge~$uv \in E(G)$ there is a node~$t \in V(G)$ with~$\{u,v\} \subseteq \chi(t)$.
\end{enumerate}
The width of a tree decomposition is defined as~$\max_{t \in V(T)} |\chi(t)| - 1$. The treewidth of a graph~$G$ is the minimum width of a tree decomposition of~$G$.
If \bmp{$w$} is a constant, then there is a \bmp{linear}-time algorithm that given a graph $G$ either outputs a tree decomposition of width at most $w$ or correctly concludes that treewidth of $G$ is larger than $w$~\cite{Bodlaender96}.
If a graph is \op, then its treewidth is at most 2~\cite[Lem. 78]{Bodlaender98}.
Since $n$-vertex graphs of treewidth~$w$ can have at most~$w\cdot n$ edges~\cite[Lem. 91]{Bodlaender98} we obtain the following.

\begin{observation}\label{lem:prelim:op-edges}
 If \bmp{$G$} is an \op graph, then \bmp{$|E(G)| \le 2 \cdot |V(G)|$}.
\end{observation}

Let $T$ be a rooted tree and $S \subseteq V(T)$ be a set of vertices in $T$. 
We define the least common ancestor of (not necessarily distinct) $u$ and $v$, denoted as \bmp{$\mathsf{LCA}(u, v)$}, to be the deepest node $x$ which is an ancestor of both $u$ and $v$.
The LCA closure of $S$ is the set
\[
\overline{\mathsf{LCA}}(S) = \{\mathsf{LCA}(u, v): u, v \in S\}.
\]

\begin{lemma}\cite[Lem. 9.26, 9.27, 9.28]{fomin2019kernelization}\label{lem:lca:basic}
Let $T$ be a rooted tree, $S \subseteq V(T)$, and $M = \overline{\mathsf{LCA}}(S)$. \bmp{All of the following hold.}
\begin{enumerate}
    \item Each connected component $C$ of $T - M$ satisfies $|N_T(C)| \le 2$.
    \item $|M| \le 2\cdot |S| - 1$.
    \item $\overline{\mathsf{LCA}}(M) = M$.
\end{enumerate}
\end{lemma}

%
%

\begin{lemma}\label{lem:lca:fewneighbors}
If~$(T, \chi)$ is a tree decomposition of width at most~$c$ of a graph~$G$, and~$B \subseteq V(T)$ is a set of nodes of~$T$ closed under taking lowest common ancestors (i.e., $\overline{\mathsf{LCA}}(B) = B$), then for~$M = \bigcup _{t \in B} \chi(t)$ and any connected component~$C$ of~$G - M$ we have~$|N_G(C) \cap M| \leq 2c$.
\end{lemma}
\begin{proof}
Let~$T_C$ denote the subgraph of~$T$ induced by the nodes whose bag contains a vertex of~$C$. Since~$C$ is a connected component of~$G - M$, we have~$V(T_C) \cap B = \emptyset$ and~$T_C$ is a connected tree rather than a forest. Hence there exists a tree~$T'$ in the forest~$T - B$ such that~$T_C$ is a subtree of~$T'$. Since~$B$ is closed under taking lowest common ancestors, it follows from \cref{lem:lca:basic} that for~$Z := N_T(V(T'))$ we have~$|Z| \leq 2$. For each~$z \in Z$, let~$f(z)$ denote the first node outside~$V(T_C)$ on the unique shortest path in~$T$ from~$V(T_C)$ to~$z$. Note that we may have~$z = f(z)$. Let~$g(z)$ denote the unique neighbor in~$T$ of node~$f(z)$ among~$V(T_C)$. Observe that both~$f(z)$ and~$g(z)$ lie on each path in~$T$ connecting a node of~$T_C$ to~$z$. 

By definition of~$T_C$ we have that each bag of~$T_C$ intersects~$V(C)$ while~$\chi(f(z)$ does not. Hence~$\chi(f(z)) \neq \chi(g(z)$. As each bag has size at most~$c+1$, it follows that~$|\chi(f(z)) \cap \chi(g(z))| \leq c$ for each~$z \in Z$. To prove the desired claim that~$|N_G(C) \cap M| \leq 2c$, it therefore suffices to argue that~$N_G(C) \cap M \subseteq \bigcup_{z \in Z} (\chi(f(z)) \cap \chi(g(z)))$.

Consider a vertex~$v \in N_G(C) \cap M$. We argue that~$v \in \chi(f(z)) \cap \chi(g(z))$ for some~$z \in Z$, as follows. Since~$v \in M$, there exists a node~$b^* \in B$ such that~$v \in \chi(b^*)$. Since~$v \in N_G(C)$ there exists~$u \in V(C)$ such that~$\{u,v\} \in E(G)$. Hence there is a bag in the tree decomposition containing both~$u$ and~$v$, and as vertices of~$V(C)$ only occur in bags of the subtree~$T_C$, we find that~$v$ occurs in at least one bag of~$T_C$. Since the occurrences of~$v$ form a connected subtree of~$T$, and~$v$ {appears} in at least one bag of~$T_C$ and at least one bag of~$B$, while the only neighbors in~$B$ of the supertree~$T'$ of~$T_C$ are the nodes in~$Z$, it follows that~$v$ occurs in at least one bag~$\chi(z)$ for some~$z \in Z$. But since all paths from~$T_C$ to~$z$ pass through~$f(z)$ and~$g(z)$ as observed above, this implies~$z \in \chi(f(z)) \cap \chi(g(z))$; this concludes the proof.
\end{proof}

\section{Splitting the graph into pieces} \label{sec:modulators}
\bmp{In this section we show how to reduce any input of \opvdfull{} to an equivalent instance which admits a decomposition into a modulator of bounded size along with a bounded number of outerplanar components containing at most four neighbors of the modulator.}

\subsection{The augmented modulator}

The starting point for both our kernelization algorithm and the one from \bmp{Fomin et al.}~\cite{FominLMS12} is \bmp{to employ} a constant-factor approximation algorithm.
We however begin with a different approximation algorithm, which has two advantages.
First, the algorithm is constructive: it relies only on separating properties of bounded-treewidth graphs and rounding a fractional solution from a~linear programming relaxation.
Second, the approximation factor can be pinned down to a~concrete value.

\begin{theorem}\label{thm:apx}\cite{gupta2019losing}
There is a polynomial-time deterministic 40-approximation algorithm for \opvdfull.
\end{theorem}
\begin{proof}
The article \cite{gupta2019losing} only states that the approximation factor is constant. However, it also provides a~recipe to retrieve its value. From \cite[Theorem 1.1]{gupta2019losing} we get that the approximation factor for \opvdfull is $2\cdot \alpha(3)$, for a function $\alpha$ satisfying the following: the problem $k$-\textsc{Subset Vertex Separator} admits a polynomial-time $(\alpha(k), \Oh(1))$-bicriteria approximation algorithm.
Without going into details, one can check that such an algorithm has been given by Lee~\cite{lee2019partitioning}: by examining the proof of Lemma 2 therein for $\eps = \frac 1 4$  we see that one can construct a polynomial-time $(8\cdot H_{2k}, 2)$-bicriteria approximation algorithm, where $H_k$ is the $k$-th harmonic number.
We check that $2 \cdot 8 \cdot H_6 < 40$.
Both algorithms in question are deterministic.
\end{proof}

In our setting, for a given graph $G$ and integer $k$, we want to determine whether $G$ admits an \op deletion set of size at most $k$.
Thanks to the theorem above, we can assume that we are given an \op deletion set $X$ (also called a modulator to outerplanarity) of size at most $40 \cdot k$.
As a next step, we would like to augment this set to satisfy a stronger property.
This step is inspired by the technique of tidying the modulator from \bmp{van Bevern, Moser, and Niedermeier}~\cite{BevernMN12}.
For each vertex $v \in X$ we would like to be able to ``put it back'' into $G-X$ while maintaining outerplanarity.
In order to do so, we look for a set of vertices from $V(G) \setminus X$ that needs to be removed if $v$ is put back.
Since $G-X$ is \op and hence has treewidth at most \bmp{two}, we can construct such a set of moderate size by a greedy approach.
We scan a tree decomposition in a bottom-up manner and look for maximal subgraphs that are \op when considered together with $v$.
When such a subgraph cannot be further extended we mark one bag of a decomposition, which gives 3 vertices to be removed. 
We show that this idea leads to a 3-approximation algorithm. \bmp{While this approach based on covering/packing duality is well-known, we present the proof for completeness.}

\begin{lemma}\label{lem:apx-undeletable}
There is a polynomial-time algorithm that, given a graph~$G$, an integer $k$, and a vertex~$v$ such that~$G - v$ is outerplanar, either finds an \op deletion set $S \subseteq V(G) \setminus \{v\}$ in $G$ of size of most $3k$ or correctly concludes that there is no \op deletion set $S \subseteq V(G) \setminus \{v\}$ in $G$ of size of most $k$.
\end{lemma}
\begin{proof}
Since $G-v$ is \op, its treewidth is at most \bmp{two}.
A tree decomposition  $(T, \chi)$ of $G-v$ of this width can be computed in \bmp{linear} time~{\cite{Bodlaender96}}.

Consider a process in which we scan the tree decomposition in a bottom-up manner and mark some nodes of $T$.
In the $i$-th step we will mark a node $t_i \in V(T)$ and maintain a family $Y_1, \dots, Y_i$ of disjoint subsets of $V(G) \setminus \{v\}$, so that for each $j \in [i]$ the graph $G[Y_i \cup \{v\}]$ is not \op.
We begin with no marked vertices and an empty family of vertex sets.
Let $U(t)$ be the set of vertices \bmp{appearing in a bag} in the subtree of $T$ rooted at $t \in V(T)$. 
In the $i$-th step we choose \bmp{a} lowest node $t_i \in V(T)$ (breaking ties arbitrarily), so that $U(t_i) \cup \{v\} \setminus \bigcup_{j=1}^{i-1}Y_j$ induces a non-\op subgraph of $G$.
If there is no such node, we terminate the process.
Otherwise we set $Y_i = U(t_i) \setminus \bigcup_{j=1}^{i-1}Y_j$ and continue the process.

By the definition, the sets $Y_1, \dots, Y_i$ are disjoint and each of them, when considered together with $v$, induces a subgraph which is not \op.
Suppose that the procedure has executed for at least $k+1$ steps.
Then for any set $S \subseteq V(G) \setminus \{v\}$ of size of most $k$, there is some $i \in [k+1]$ such that $Y_i \cap S = \emptyset$.
Since $G[Y_i \cup \{v\}]$ is not \op, we can conclude that $S$ is not an \op deletion set. \bmp{Hence we can conclude that no set as desired exists and terminate.}

Suppose now that the procedure has terminated at \bmp{the} $k'$-th step, where $k' \le k$.
Since $t_i$ is chosen as \bmp{a} lowest node among \bmp{those} satisfying the given condition, we get that $U(t_i) \cup \{v\} \setminus (\chi(t_i) \cup \bigcup_{j=1}^{i-1}Y_j)$ induces an \op subgraph of $G$.
Observe that $S = \bigcup_{j=1}^{k'} \chi(t_j)$ separates $Y_i$ from $Y_j$ in $G-v$ for each pair $1 \le i < j \le k'$, because in particular $\chi(t_i) \subseteq S$.
Let $Y_0 = V(G) \setminus \bigcup_{j=1}^{k'}Y_j$.
Then also $G[Y_0 \cup \{v\}]$ is \op and $S$ separates $Y_0$ from any $Y_i$ in $G-v$.
We apply \cref{criterion:articulation-point} to $G-S$ with articulation point $v$
and check that any connected component $C$ of $G-S-v$ is contained in some set $Y_i$, so $G\brb C$ is \op, 
and thus $G-S$ is \op.
The size of each bag in $(T,\chi)$ is at most 3, hence $|S| \le 3k' \le 3k$.
The claim follows.
\end{proof}

Observe that if it is impossible to remove $k$ vertices from $G-(\hui{X \setminus \{v\}})$ to make it \op, then any \op deletion set in $G$ of size at most $k$ must contain $v$.
In this situation it suffices to solve the problem on $G-v$.
Otherwise, we identify a set $R(v)$ of at most $3k$ vertices whose removal allows $v$ to be put back in $G-X$ without spoiling outerplanarity.
After inserting $R(v)$ into the set $X$,
\bmp{we} could put $v$ back ``for free''.
Let us formalize this idea of augmenting the modulator.

\begin{definition} \label{def:augmentedmod}
A $(k,c)$-\emph{augmented modulator} in graph $G$ is a pair of disjoint sets $X_0, X_1 \subseteq V(G)$ such that:
\begin{enumerate}
    \item \label{def:augmentedmod:mod} $G - X_0$ is \op,
    \item \label{def:augmentedmod:aug} for each $v \in X_0$, there is a set $R(v) \subseteq X_1$, such that $|R(v)| \le 3k$ and $G-((X_0 \setminus \{v\}) \cup R(v))$ is \op, and 
    \item \label{def:augmentedmod:size} $|X_0| \le c\cdot k$, $X_1 = \bigcup_{v\in X_0} R(v)$, which implies $|X_1| \le 3c\cdot k^2$.
\end{enumerate}
We classify the pairs of vertices within $X_0 \cup X_1$.
A pair $(u,v) \colon u, v \in X_0 \cup X_1$ is of type:
\begin{enumerate}
    \item [A:] if $u, v \in X_0$ or $(u \in X_0, v \in R(u))$ or $(v \in X_0, u \in R(v))$,
    \item [B:] if $(u,v)$ is not of type $A$ and $\{u,v\} \cap X_0 \ne \emptyset$,
    \item [C:] if $u, v \in X_1$.
\end{enumerate}
We note that the number of type-A pairs is at most $c(3+c) \cdot k^2$, the number of type-B pairs is at most $3c^2 \cdot k^3$,
and the number of type-C pairs is at most $9c^2 \cdot k^4$.
\end{definition}

The downside of the augmented modulator is that its size can be as large as $\Oh(k^2)$.
However, in return we obtain an even stronger property than previously sketched.
For most of the pairs of vertices $u, v$ from the augmented modulator $(X_0,X_1)$, putting them back into $G-(X_0 \cup X_1)$ at the same time still does not break outerplanarity.
This property will come in useful for bounding the size of the kernel.

\begin{observation}\label{lem:augmentedmod:properties}
Let $(X_0, X_1)$ be a $(k,c)$-augmented modulator in a graph $G$. Then for each $v \in X_0 \cup X_1$, the graph $G - (X_0 \cup X_1 \setminus \{v\})$ is \op.
Furthermore, if  $u,v \in X_0 \cup X_1$
and the pair $(u,v)$ is of type B or C, then 
the graph $G - (X_0 \cup X_1 \setminus \{u,v\})$ is \op.
\end{observation}

Let us summarize what we can compute so far.
We say that instances $(G,k)$ and $(G',k')$ are equivalent if $\opd(G) \le k \Leftrightarrow \opd(G') \le k'$.

\begin{lemma}\label{lem:modulator:compute-augmented}
There is a polynomial-time algorithm that, given an instance $(G,k)$, either correctly concludes that $\opd(G) > k$ or outputs an equivalent instance $(G',k')$, where $k' \le k$ and $G'$ is a subgraph of $G$, along with a $(k',40)$-augmented modulator in $G'$.
\mic{If~$\opd(G) \leq k$ then it holds that $\opd(G') = \opd(G) - (k - k')$.}
\bmp{Moreover, if for every vertex~$v \in V(G)$ there is an outerplanar deletion set~$S \subseteq V(G) \setminus \{v\}$ in~$G$ of size at most~$k$, then~$k' = k$.}
\end{lemma}
\begin{proof}
We run the 40-approximation algorithm from \cref{thm:apx} to obtain an \op deletion set $X_0$.
If $|X_0| > 40\cdot k$, we conclude that $\opd(G) > k$.
Otherwise, we iterate over $v \in X_0$ and execute the subroutine from \cref{lem:apx-undeletable} with respect to the graph $G_v = G - (X_0 \setminus \{v\})$.
If for any vertex $v$ we have concluded that
$G_v$ does not admit any \op deletion set $S \subseteq V(G_v) \setminus \{v\}$ of size at most $k$, then the same holds for $G$.
This implies that any \op deletion set in $G$ of size at most $k$ (if there is any) must include the vertex $v$
and the instance $(G-v, k-1)$ is equivalent to $(G,k)$.
\mic{Furthermore, in this case $\opd(G') = \opd(G) - 1$ as long as $\opd(G) \le k$.}
We can thus remove the vertex $v$ from $G$, decrease the value of parameter $k$ by 1, and start the process from scratch.
If during this process we reach an~instance $(G',0)$, then $(G,k)$ is satisfiable if and only if $G'$ is \op.
\mic{Observe that if for every vertex~$v \in V(G)$ there is an outerplanar deletion set~$S \subseteq V(G) \setminus \{v\}$ in~$G$ of size at most~$k$, then 
this holds also for the graph $G_v$ and thus we will not apply the reduction rule decreasing the value of $k$.}

Suppose now that for each $v \in X_0$
we have obtained a set $S_v \subseteq V(G_v) \setminus \{v\}$
of size at most $3k$ such that $G - ((X_0 \setminus \{v\}) \cup S_v)$ is \op.
Then setting $R(v) = S_v$ and $X_1 = \bigcup_{v\in X_0} R(v)$ satisfies the requirements of \cref{def:augmentedmod}.
\end{proof}

\mic{The reduction step above is the only one in our algorithm that may decrease the value of~$k$.
Moreover, no further reduction will modify the \op deletion number as long as $\opd(G) \le k$.
This observation will come in useful for bounding the size of minimal minor obstructions to having an \op deletion set of size $k$.
}

As the next step, we would like to bound the number of connected components in~$G-(X_0 \cup X_1)$ and the number of connections between the components and the modulator vertices.
We show that if vertices  $u, v \in X_0 \cup X_1$ are adjacent to sufficiently many components, then at least one of $u,v$ must be removed in any solution of size at most~$k$.
Together with the ``putting back'' property of the augmented modulator,
this allows us to forget some of the edges without modifying the space of solutions of size at most~$k$.
We formalize this idea with the following marking scheme.

\begin{figure}[bt]
  \centering
  \includegraphics{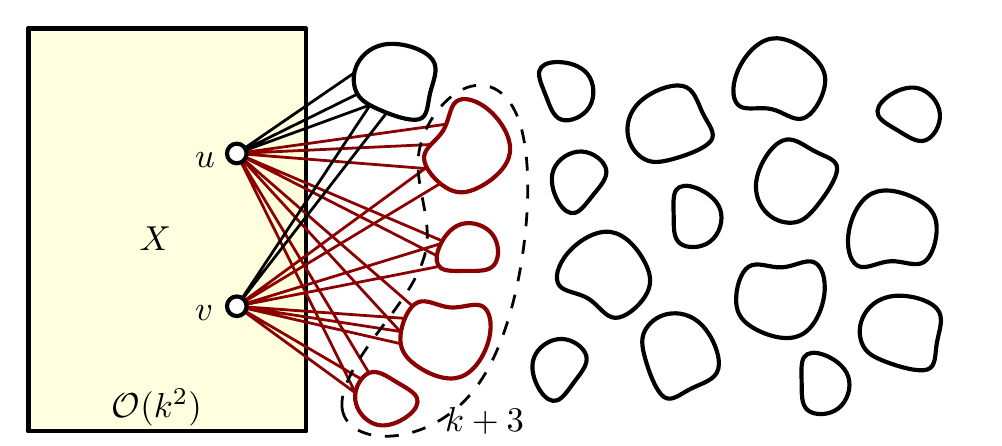}
  \caption{Illustration of \cref{reduction:reduce-degree-new}.
  For each pair $u, v \in X = X_0 \in X_1$ we choose up to $k+3$ components of $G-X$ with edges to both $u$ and $v$ and mark the corresponding edges in the component graph $\mathcal{C}(G, X)$. If a pair $(v,C)$ is not marked in the end, all the edges between $v$ and $C$ are removed.
  }
  \label{fig:reduce-degree-new}
\end{figure}

\begin{reduction}\label{reduction:reduce-degree-new}
Let $G$ be a graph, $k \in \mathbb{N}$,
and $(X_0, X_1)$ be a $(k,c)$-augmented modulator in $G$.
Consider the component graph $\mathcal{C}(G, X_0 \cup X_1)$.
For each pair $u,v \in X_0 \cup X_1$ choose up to $k+3$ components $C_i$ with edges to both $u$ and $v$, and mark the edges $(u,C_i), (v,C_i)$ in $\mathcal{C}(G, X_0 \cup X_1)$.
If an edge $(v,C)$ is unmarked in the end, remove all the edges between $v$ and $C$ in $G$.
If some component $C$ of $G - (X_0 \cup X_1)$ or a vertex $v \in X_0 \cup X_1$ becomes isolated, remove it from $G$.
\end{reduction}

\begin{lemma}[Safeness]
Let $G$ be a graph, $k \in \mathbb{N}$,
and $(X_0, X_1)$ be a $(k,c)$-augmented modulator in $G$.
Let $G'$ be obtained from $G$ by applying 
\cref{reduction:reduce-degree-new} with respect to $(X_0, X_1, k)$.
\mic{If $\opd(G) > k$ then $\opd(G') > k$ and
if $\opd(G) \le k$ then $\opd(G') = \opd(G)$.}
\end{lemma}
\begin{proof}
\mic{It suffices to show that any solution in $G'$ of size at most $k$ is also valid in $G$.}
Removing an \op connected component is always safe so it suffices to argue for the correctness of the edge removal rule.
Consider a single step of the reduction in a graph $G$, in which we have removed the edges between vertex $v \in (X_0 \cup X_1)$ and a connected component $C$ of $G-(X_0 \cup X_1)$.
Let $G'$ be the graph after this modification and $S$ be an \op deletion set of size at most $k$ in $G'$.
If $v \in S$, then $G' - S = G - S$ so let us assume that $v \not\in S$.

Suppose there is another $u \in X_0 \cup X_1$ with an edge to $C$ in $G$.
Since the pair $(v,C)$ was not marked, there are $k+3$ components $C_i$, different from $C$, of $G-(X_0 \cup X_1)$ with edges to both $u$ and $v$.
These pairs were marked, so they cannot be removed in any previous reduction step.
By {a} counting argument, at least 3 of these components have empty intersections with $S$.
If $u \not\in S$, then these components together with $\{u,v\}$ form a minor model of $K_{2,3}$ in $G'-S$, which is not possible.
Therefore, $u \in S$.

It follows that $v$ is the only neighbor of $C$ in $G - S$.
By \cref{lem:augmentedmod:properties} we can ``put back'' $v$ into $G-(X_0 \cup X_1)$ without spoiling the outerplanarity and so the graph $(G-S)\brb C$ being the subgraph of $G[C \cup \{v\}]$ is \op.
The graph $G - S - C$ is a subgraph of $G' - S$, so it is also \op.
The intersection of their vertex sets is exactly $\{v\}$ so
from \cref{criterion:articulation-point} we obtain that $G - S$ is \op.
\end{proof}

Now we show that after application of \cref{reduction:reduce-degree-new} the component graph $\mathcal{C}(G, X_0 \cup X_1)$ cannot be too large.
This will come in useful for proving further upper bounds.
We could trivially bound the number of its edges by $|X_0 \cup X_1|^2\cdot(k+3) = \Oh(k^5)$ but, thanks to the properties of the augmented modulator, we can be more economical.
First, we need a~simple observation about bipartite \op graphs.

\begin{prop}\label{lem:outerplanar-bipartite}
Consider an \op bipartite graph $(X \cup Y, E)$ such that all the vertices in $Y$ have {degree} at least two. Then $|Y| \le 4\cdot|X|$ and $|E| \le 10\cdot|X|$.
\end{prop}
\begin{proof}
Remove part of the edges so that each vertex in $Y$ has degree exactly two.
Now contract each vertex from $Y$ to one of its neighbors.
The constructed graph is a minor of $(X \cup Y, E)$ with a vertex set $X$, so it is \op and the number of edges is at most $2\cdot|X|$ by \cref{lem:prelim:op-edges}.
Each edge could have been obtained by at most 2 different contractions, as otherwise  $(X \cup Y, E)$ would contain $K_{2,3}$ as a minor.
Therefore $|Y| \le 4\cdot|X|$.
Again by \cref{lem:prelim:op-edges}, the number of edges in  $(X \cup Y, E)$ is at most $2\cdot(|X| + |Y|) \le 10\cdot|X|$.
\end{proof}

Recall the types of pairs from \cref{def:augmentedmod} and their properties from
\cref{lem:augmentedmod:properties}.
We know that the number of type-A pairs is at most $c(3+c) \cdot k^2$ and the number of type-B pairs is at most $3c^2 \cdot k^3$.
Moreover, pairs of type $B$ can be inserted back \bmp{into} $G-(X_0 \cup X_1)$ without affecting its outerplanarity.

\begin{lemma}\label{lem:component-graph-x0-x1-edge-bound}
After the application of Rule~\ref{reduction:reduce-degree-new} with respect to a $(k,c)$-augmented modulator ($X_0, X_1)$, the number of vertices and edges in $\mathcal{C}(G, X_0 \cup X_1)$ is at most $f_1(c) \cdot (k+3)^3$, where $f_1(c) = 14c^2 + 60c$.
\end{lemma}
\begin{proof}
For pairs of type A we have marked at most $(k+3)\cdot c(3+c) \cdot k^2$ edges.
If $(u,v)$ is of type B, then by \cref{lem:augmentedmod:properties} the graph $G - (X_0 \cup X_1 \setminus \{u,v\})$ is \op and there can be at most 2 components adjacent to both $u,v$ as otherwise we would obtain a $K_{2,3}$-minor.
Hence, for pairs of type B we have marked at most $2\cdot 3c^2 \cdot k^3$ edges.

Next, we argue that the total number of edges marked due to pairs of type C is $30c \cdot k^2$.
Let $E_C \subseteq E(\mathcal{C}(G, X_0 \cup X_1))$ denote the set of these edges.
Let $Y_C \subseteq V(\mathcal{C}(G, X_0 \cup X_1))$ be the set of these connected components of  $G - (X_0 \cup X_1)$ which are incident to at least one edge from $E_C$ in $\mathcal{C}(G, X_0 \cup X_1)$.
By the definition of the marking scheme, if $C \in Y_C$ then $C$ is in fact incident to at least 2 edges from $E_C$, and their other endpoints belong to $X_1$. Consider the subgraph $(X_1 \cup Y_C, E_C)$ of $\mathcal{C}(G, X_0 \cup X_1)$.
It is a minor of $G-X_0$, therefore it is \op.
By \cref{lem:outerplanar-bipartite}, we get that $|E_C| \le 10 \cdot |X_1| = 10 \cdot 3c \cdot k^2.$

We can thus estimate the number of edges in $\mathcal{C}(G, X_0 \cup X_1)$ by $(7c^2 + 3c) \cdot (k+3)^3 + 30c \cdot k^2 \le (7c^2 + 30c) \cdot (k+3)^2$.
Finally, since $\mathcal{C}(G, X_0 \cup X_1)$ contains no isolated vertices, the number of vertices is at most twice the number of edges.
\end{proof}

\subsection{The outerplanar decomposition}

We proceed by enriching the augmented modulator further.
We would like to provide additional properties \bmp{at the expense} of growing the modulator size to $\Oh(k^3)$.
For two vertices $u,v$ in an augmented modulator $(X_0, X_1)$ ideally we would like to ensure that no two components of $G-(X_0 \cup X_1 \cup Z)$ are adjacent to both $u$ and~$v$, where $Z$ is some vertex set of size $\Oh(k^3)$.
This is not always possible, but we will guarantee that in such a case any \op deletion set of size at most $k$ must contain either $u$ or $v$.

\begin{definition}
Let $Y \subseteq V(G)$ be a vertex subset in a graph $G$.
We say that $u,v \in Y$ are $Y$-separated if no connected component of $G - Y$ is adjacent to both $u$ and~$v$.
\end{definition}

In \cref{lem:modulator:outerplanar-separator} we are going to show that when $G$ is \op and $X \subseteq V(G)$, then there always exists a small set $Y \subseteq V(G)$ so that every pair from $X$ is $(X \cup Y)$-separated. \bmp{Towards that goal, the need the following proposition.}

\begin{prop}\label{prop:modulator:single-separation}
Let $X \subseteq V(G)$ be an independent set in an \op graph $G$.
Then there exists $v \in X$ and $S \subseteq V(G) \setminus X$ of size at most \bmp{four}, so that $S$ is a $(v, X \setminus v)$-separator in $G$.
\end{prop}
\begin{proof}
Consider a tree decomposition~\bmp{$(T,\chi)$ of $G$ of width \bmp{two} where~$T$ is rooted at an arbitrary node~$r$.}
For a vertex \bmp{$v \in V(G)$} let $t_v \in V(T)$ be the node which is closest to the root~$r$, among those whose \bmp{bag} contain $v$.
Consider $v \in X$ for which $t_v$ is furthest from the root (if there are many, pick any of them) and let \bmp{$B_v := \chi(t_v)$}.
\bmp{By standard properties of tree decompositions,} any path from $v$ to $X \setminus v$ either goes through $B_v \setminus v$ or ends at $B_v \setminus v$.

If $(B_v \setminus v) \cap X = \emptyset$, set $S = B_v \setminus v$.
If $B_v \setminus v = \{u_1, u_2\}$, where $u_1 \in X$, $u_2 \not\in X$, consider a minimal $(v,u_1)$-separator $S'$ and set $S = S' \cup \{u_2\}$.
There cannot be \bmp{three} vertex-disjoint paths connecting $v,u_1$ as $vu_1 \not\in E(G)$ and this would give a minor model of $K_{2,3}$ in $G$.
Therefore \bmp{by Menger's theorem} we have $|S'| \le 2$ and $|S| \le 3$.
Finally, suppose $(B_v \setminus v) \subseteq X$.
In this case, let $S$ be a minimal $(v,B_v \setminus v)$-separator.
If there were \bmp{five} vertex-disjoint paths connecting $v$ and $B_v \setminus v$ then in particular there would be \bmp{three} vertex-disjoint paths connecting $v$ and some $u \in B_v \setminus v$, which would again give a minor model of $K_{2,3}$ in $G$.
Therefore $|S| \le 4$.

Suppose there is a path in $G \setminus S$ connecting $v$ with some $x \in X \setminus v$.
It contains a subpath connecting $v$ with some $u \in B_v \setminus v$.
If $u \not\in X$, then $u \in S$, so suppose that $u \in X$.
But $S$ contains a $(v,u)$-separator, so such path cannot exist in $G - S$.
\end{proof}

\begin{lemma}\label{lem:modulator:outerplanar-separator}
There is a polynomial-time algorithm that, given
a vertex set $X \subseteq V(G)$ in an \op graph $G$, finds a vertex set $Y \subseteq V(G) \setminus X$ of size {at most} $4\cdot|X|$, so that every pair $u,v \in X$ \bmp{with~$\hui{u} \neq \hui{v}$} is $(X \cup Y)$-separated.
\end{lemma}
\begin{proof}
We can assume that $X$ is an independent set in $G$ because
\bmp{removing edges} between vertices in $X$ does not affect the neighborhood of a connected component in~$G - (X \cup Y)$.
Initialize $Y = \emptyset$.
By \cref{prop:modulator:single-separation} we can find a vertex~$v \in X$ that can be separated from~$X \setminus v$ by at most 4 vertices.
Add these vertices to $Y$ and repeat this operation recursively on~$X \setminus v$.
\end{proof}

Given an augmented modulator $(X_0, X_1)$, we would like to find a set $Z$ of moderate size so that for each pair $(u,v)$ from $X_0 \cup X_1$ either $u,v$ are $(X_0 \cup X_1 \cup Z)$-separated or there exist $k+4$ \bmp{internally vertex-disjoint} paths, with non-empty interior, connecting $u$ and~$v$ in $G$.
If the latter case occurs, then any \op deletion set of size bounded by $k$, can intersect at most $k$ of these paths' interiors.
Therefore, this solution must remove either $u$ or $v$ in order to get rid of \hui{all} $K_{2,3}$-minor\hui{s}.
We remark that this property already holds if we request $k+3$ disjoint $(u,v)$-paths, but in this stronger form it also holds for a graph obtained from $G$ by an edge removal.
This fact will be crucial for the safeness proof for \cref{reduction:irrelevant-edge}.

In order to find the set $Z$, we could consider all pairs $(u,v)$ from $X_0 \cup X_1$ and, if there exists an $(u,v)$-\bmp{separator} of size at most $k+3$, add it to $Z$.
This however would make $Z$ as large as $\Oh(k^5)$.
We \hui{can} make this process more economical \hui{by} analyzing what happens for
different types of pairs from \cref{def:augmentedmod}.
Recall that the number of type-A pairs is at most $c(3+c) \cdot k^2$ and the number of type-B pairs is at most $3c^2 \cdot k^3$.

\begin{lemma}\label{lem:modulator:decomposition-xz}
There is a polynomial-time algorithm that, given an instance $(G,k)$ with $(k,c)$-augmented modulator $(X_0, X_1)$, returns a set $Z \subseteq V(G) \setminus (X_0 \cup X_1)$ of size at most  $f_2(c) \cdot (k+3)^3$, where $f_2(c) = 4c^2 + 15c$, such that for each pair $u,v \in X_0 \cup X_1$ \bmp{of distinct vertices} one of the following holds:
\begin{enumerate}
    \item vertices $u,v$ are $(X_0 \cup X_1 \cup Z)$-separated, or
    \item there are $k+4$ vertex-disjoint paths, with non-empty interior, connecting $u$ and $v$ in $G$.\label{lem:modulator:decomposition-xz:item:many-paths}
\end{enumerate}
\end{lemma}
\begin{proof}
Initialize $Z_0 = \emptyset$.
Consider all the pairs $(u,v)$ from the augmented modulator.
If $(u,v)$ is of type A or B, compute \bmp{a} minimum $(u,v)$-separator $S_{u,v}$ \bmp{with~$u,v \notin S_{u,v}$} in $G - (X_0 \cup X_1 \setminus \{u,v\}) \setminus uv$, that is, we remove the edge $uv$ if it exists.
If $|S_{u,v}| \le k+3$, add $S_{u,v}$ to $Z_0$.
Recall {from \cref{lem:augmentedmod:properties}} that if $(u,v)$ is of type B, then the graph $G - (X_0 \cup X_1 \setminus \{u,v\})$ is \op, so $|S_{u,v}| \le 2$, as otherwise we could construct a $K_{2,3}$ minor.
For pairs of type A we add at most $(k+3)\cdot c(3+c) \cdot k^2$ elements, and for pairs of type B at most $2\cdot 3c^2 \cdot k^3$ elements.
\mic{If the pair $(u, v)$ does not satisfy condition (\ref{lem:modulator:decomposition-xz:item:many-paths}), then the set $Z_0$ contains a set $S_{u,v}$ which forms a $(u,v)$-separator in $G - (X_0 \cup X_1 \setminus \{u,v\}) \setminus uv$.
Therefore $u, v$ belong to different connected components of $G - (X_0 \cup X_1 \cup Z_0 \setminus \{u,v\}) \setminus uv$ and so they are $(X_0 \cup X_1 \cup Z_0)$-separated.} 

To cover pairs of type $C$, consider the \op graph $G-X_0$.
By \cref{lem:modulator:outerplanar-separator} we can find a vertex set $Z_1 \subseteq V(G) \setminus (X_0 \cup X_1)$ of size $4 \cdot |X_1| \le 12c \cdot k^2 $ so that all pairs $u,v\in X_1$ are $(X_1 \cup Z)$-separated in $G - X_0$.
We return the set $Z_0 \cup Z_1$, which has no more than
$(k+3)\cdot c(3+c) \cdot k^2 + 2\cdot 3c^2 \cdot k^3 + 12c \cdot k^2 \le (4c^2 + 15c) \cdot (k+3)^3$ elements.
\end{proof}

We would like to simplify the interface between a connected component $C$ of $G - (X_0 \cup X_1 \cup Z)$ and the rest of the graph.
Since~$G - X_0$ is \op it has treewidth at most two, which implies there is a tree decomposition in which each pair of distinct bags intersects in at most~$2$ vertices. When constructing a separator~$Z' \supseteq Z$ via the LCA closure, the neighborhood of each connected component~$C$ of~$G-Z'$ within the set~$Z'$ is contained in at most two bags of the decomposition. This allows us to guarantee that~$|N_G(C) \cap Z'| \leq 4$. 

\begin{lemma}\label{lem:modulator:lca}
There is a polynomial-time algorithm that, given an \op graph $G$ and $Z \subseteq V(G)$, returns a set $Z' \supseteq Z$ of size at most $6 \cdot |Z|$ such that each connected component of $G - Z'$ has at most \bmp{four} neighbors in $Z'$.
\end{lemma}
\begin{proof}
Consider a tree decomposition $(T,\chi)$ of width \bmp{two} of the graph $G$, rooted at a node~$r \in V(T)$.
It can be found in \bmp{linear} time~{\cite{Bodlaender96}}.
For a vertex $v$ let $t_v$ be the node which is closest to the root among those whose bags contain $v$.
Consider the set of nodes $T_Z = \{t_v \mid v \in Z\}$.
Let $T'_Z = \overline{\mathsf{LCA}}(T_Z)$ be the LCA closure of $T_Z$. 
Finally, let $Z'$ be union of all bags in $T'_Z$.
We have $|T'_Z| \le 2 \cdot |T_Z|$ and $|Z'| \le 6\cdot |Z|$.
{By \cref{lem:lca:fewneighbors} we obtain that
each connected component of $G - Z'$ has at most \bmp{four} neighbors.}
\end{proof}

In order to keep the kernel size in check, we need to analyze the number of connected components of $G - (X_0 \cup X_1 \cup Z)$.
We have managed to bound the size of $Z$ by $\Oh(k^3)$ and, in \cref{lem:component-graph-x0-x1-edge-bound}, we have also bounded by $\Oh(k^3)$ the number of edges in the component graph $\mathcal{C}(G,X_0 \cup X_1)$.
These two properties suffice to also bound the number of  connected components of $G - (X_0 \cup X_1 \cup Z)$ that have at least two neighbors in $X_0 \cup X_1 \cup Z$. It will be easier to deal with the remaining ones later.

\begin{lemma}\label{lem:modulator:decomposition-bounded}
Let $(X_0, X_1)$ be a $(k,c)$-augmented modulator in $G$,
so that the component graph $\mathcal{C}(G,X_0 \cup X_1)$ has at most $s$ vertices and $s$ edges, and let $Z \subseteq V(G) \setminus (X_0 \cup X_1)$.
Then \hui{there are at most~$3\cdot s + 4\cdot |Z|$ components of~$G - (X_0 \cup X_1 \cup Z)$ that have two or more neighbors in~$X_0 \cup X_1 \cup Z$.}
\end{lemma}
\begin{proof}
Let $X = X_0 \cup X_1$.
We analyze the number of connected components of~$G - (X \cup Z)$ by splitting them into three categories.
\begin{enumerate}
    \item Components with at least two neighbors in~$Z$.
    Consider a subgraph $(Z \cup Y, E)$ of $\mathcal{C}(G, X \cup Z)$
    given by restricting the vertex-side to $Z$ and the component-side to those components~$Y$ that have at least two neighbors in~$Z$.
    This graph is a minor of $G-X$, so it is \op.
    By \cref{lem:outerplanar-bipartite}, we get $|Y| \le 4 \cdot |Z|$. 
    
    \item Components with exactly one neighbor in~$Z$ \hui{and at least one in~$X$}.
    We call such a component a dangling component.
    For a connected component~$H$ of~$G - X$,
    consider the collection of dangling components $(C_i)_{i=1}^\ell$ within $H$.
    Since each dangling component has exactly one neighbor in $H$, removing it does not affect connectivity of $H$.
    Therefore the graph $H' = H - \bigcup_{i=1}^\ell C_i$ is connected.
    \hui{Note that~$H'$ cannot be empty since it must contain at least one vertex in~$Z$.}
    For each vertex~$x \in X$, there are at most two dangling components in~$H$ which are adjacent to~$x$:
    if there were three~$C_1, C_2, C_3$,  they would form a minor model of~$K_{2,3}$ together with~$x$ and~$V(H')$.
    By \cref{lem:augmentedmod:properties} this would contradict outerplanarity of~$G - (X \cup \{x\})$.
    
    Hence the number of dangling components within~$H$ is at most twice as large as $|N_G(H) \cap X|$, which is the degree of $H$ in $\mathcal{C}(G,X)$.
    The total number of dangling components is thus at most 2 times the
    sum of degrees of the component-nodes
    in $\mathcal{C}(G,X)$, which equals the number of edges in $\mathcal{C}(G,X)$.
    We obtain a bound $2s$ on the total number of dangling components.
    
    \item Components without any neighbors in~$Z$.
    These are also components of $G-X$, so 
    there are at most $s$ of them.
    \qedhere
\end{enumerate}
\end{proof}

\hui{The previous lemma gives us a bound on the number of components outside the modulator with at least two neighbors. To bound the total number of components outside the modulator, we employ the following reduction rule to remove the remaining components with at most one neighbor.}

\begin{reduction}\label{reduction:articulation-point}
If for some~$C \subseteq V(G)$ the graph~$G\brb C$ is \op and it holds that $|N_G(C)| \leq 1$, then remove the vertex set~$C$.
\end{reduction}

Safeness of this rule follows from \cref{criterion:articulation-point}, which implies~$\opd(G-C) = \opd(G)$.

With these properties at hand, we are able to construct the desired extension of the augmented modulator.
The decomposition below is inspired by the notion of a near-protrusion~\cite{FominLMS12}, combined
with the idea of the augmented modulator, and with an $\Oh(k^3)$ bound on the number of leftover connected components.

\begin{definition} \label{def:op-decomp}
For~$k,c,d\in \mathbb{N}$ a $(k,c,d)$-\op decomposition of a graph $G$ is a triple $(X_0,X_1,Z)$ of disjoint vertex sets in $G$, such that:
\begin{enumerate}
    \item \label{def:op-comp:augmod} $(X_0, X_1)$ is a $(k,c)$-augmented modulator for $(G,k)$,
    \item for each pair $u,v \in X_0 \cup X_1$ of distinct vertices one of the following holds:\label{item:decomp:pairs}
    \begin{enumerate}
        \item vertices $u,v$ are $(X_0 \cup X_1 \cup Z)$-separated, or
        \item there are $k+4$ vertex-disjoint $(u,v)$-paths in $G$, each with non-empty interior.
    \end{enumerate}
    \item for each connected component $C$ of $G - (X_0 \cup X_1 \cup  Z)$ it holds that  $|N_G(C) \cap Z| \le 4$,\label{item:decomp:neighbors}
    \item \label{def:op-decomp:size} $|Z| \le d\cdot (k+3)^3$ and there are at most $d\cdot (k+3)^3$ connected components in $G-(X_0 \cup X_1 \cup Z)$.
\end{enumerate}
\end{definition}

\begin{lemma}\label{lem:modulator:summary} 
There is a constant $c$, a function $f_3 \colon \setN \rightarrow \setN$, and a polynomial-time algorithm that, given an instance $(G,k)$, either returns an equivalent instance $(G',k')$, where $k' \le k$ and $G'$ is subgraph of $G$, along with a $(k',c,f_3(c))$-\op decomposition of~$G'$, or concludes that $\opd(G) > k$.
\mic{If~$\opd(G) \leq k$ then it holds that $\opd(G') = \opd(G) - (k - k')$.}
Furthermore, $c = 40$ and $f_3(c) = 3\cdot f_1(c) + 24\cdot f_2(c)$ (see \cref{lem:component-graph-x0-x1-edge-bound,lem:modulator:decomposition-xz}).
\end{lemma}
\begin{proof}
Begin with \cref{lem:modulator:compute-augmented} to either conclude $\opd(G) > k$ or find an equivalent instance $(G'',k')$, where $k' \le k$ and $G''$ is a subgraph of $G$, along with an $(k',c)$-augmented modulator.
Next, apply \cref{reduction:reduce-degree-new} to obtain an equivalent instance $(G',k')$ that satisfies the conditions in the statement, along with an $(k',c)$-augmented modulator $(X_0,X_1)$, so that the number of vertices and edges in $\mathcal{C}(G,X_0 \cup X_1)$ is at most $f_1(c) \cdot (k'+3)^3$ (see \cref{lem:component-graph-x0-x1-edge-bound}).
This reduction rule may remove edges and vertices from the graph, so $G'$ is a subgraph of $G$.

We find a set $Z_0$ of size at most $f_2(c) \cdot (k'+3)^3$ satisfying the Condition~\ref{item:decomp:pairs} with \cref{lem:modulator:decomposition-xz}.
Next, apply \cref{lem:modulator:lca} to graph $G-(X_0 \cup X_1)$ and set $Z_0$ to compute $Z \supseteq Z_0$, $|Z| \le 6\cdot f_2(c) \cdot (k'+3)^3$,
which satisfies the Condition~\ref{item:decomp:neighbors}.
Observe that that Condition~\ref{item:decomp:pairs} is preserved for any superset of $Z_0$, and hence for $Z$.

\hui{Now identify all connected components of~$G-(X_0 \cup X_1 \cup Z)$ with only one neighbor and apply \cref{reduction:articulation-point} to remove \bmp{them}. Note that this removed only vertices disjoint from~$X_0$, $X_1$, and~$Z$, so Conditions~\ref{def:op-comp:augmod}, \ref{item:decomp:pairs}(a), and~\ref{item:decomp:neighbors} remain satisfied. Since such a connected component only has one neighbor in~$X_0 \cup X_1$, the number of~$(u,v)$-paths in~$G$ cannot have been decreased for any distinct~$u,v \in X_0 \cup X_1$, hence Condition~\ref{item:decomp:pairs}(a) also remains satisfied. We complete the proof by showing that now Condition~\ref{def:op-decomp:size} holds. 

First note that~$|Z| \leq 6\cdot f_2(c) \cdot (k'+3)^3 \leq f_3(c) \cdot (k'+3)^3$. It remains to show that the number of components in~$G-(X_0 \cup X_1 \cup Z)$ is at most~$f_3(c) \cdot (k'+3)^3$. Any such connected component has at least two neighbors since otherwise we would have applied \cref{reduction:articulation-point} to remove it. Any other connected component has at least two neighbors in~$X_0 \cup X_1 \cup Z$, so by \cref{lem:modulator:decomposition-bounded} there are at most~$3 \cdot s + 4 \cdot |Z|$ of these components, where~$s$ denotes the number of edges in~$\mathcal{C}(G,X_0 \cup X_1)$ which is upper bounded by~$f_1(c) \cdot (k'+3)^3$ (see Lemma~\ref{lem:component-graph-x0-x1-edge-bound}). Hence in total there are at most~$3\cdot f_1(c) \cdot (k'+3)^3 + 4\cdot 6\cdot f_2(c) \cdot (k'+3)^3 = f_3(c) \cdot (k'+3)^3$ components.}
%
\end{proof}

As the last property of the $(k,c,d)$-outerplanar decomposition, we formulate the bound on the total number of connections between $X_0 \cup X_1 \cup Z$ and the leftover components, which will lead to the total kernel size $\Oh(k^4)$.

\begin{lemma}\label{lem:modulator:k4}
Let $(X_0,X_1,Z)$ be a $(k,c,d)$-\op decomposition of a graph $G$.
Then the number of edges in the component graph $\mathcal{C}(G, X_0 \cup X_1 \cup Z)$ is at most $f_4(c,d) \cdot (k+3)^4$,
where $f_4(c,d) = cd+6c+4d$.
\end{lemma}
\begin{proof}
By \cref{def:op-decomp}\eqref{def:op-decomp:size} there are at most~$d\cdot (k+3)^3$ components of~$G-(X_0 \cup X_1 \cup Z)$ and each can have at most
$|X_0| \le c\cdot k$ neighbors from $X_0$.
It remains to bound the total number of edges from $X_1 \cup Z$.
The graph given by restricting the vertex-side of $\mathcal{C}(G, X_0 \cup X_1 \cup Z)$ to $X_1 \cup Z$ is a minor of $G - X_0$, hence it is \op and, by \cref{lem:prelim:op-edges}, the number of edges is at most \bmp{twice} 
the number of vertices, that is,
$2 \cdot (|X_1 \cup Z| + d\cdot (k+3)^3) \le 6c \cdot k^2 + 4d \cdot (k+3)^3$.
\end{proof}

\subsection{Reducing the size of the neighborhood} \label{sec:fan-reduct}

Given a~$(k,c,d)$-\op decomposition~$(X_0,X_1,Z)$, we will now present the final reduction rule to reduce the size of the neighborhood~$N_G(X_0 \cup X_1)$ to~$\Oh(k^4)$. As the size of~$Z$ is already bounded by~$\Oh(k^3)$ we focus on reducing the size of~$N_G(X_0 \cup X_1) \setminus Z$. We have already shown the number of edges in the component graph~$\mathcal{C}(G,X_0 \cup X_1 \cup Z)$ is bounded by~$\Oh(k^4)$, so it suffices to reduce the number of edges between a single modulator vertex~$x \in X_0 \cup X_1$ and a connected component~$C$ of~$G-(X_0 \cup X_1 \cup Z)$ to a constant. For this, we first show in the following lemma where the neighbors of~$x$ occur in~$C$.

\begin{lemma}\label{lem:neighborhood}
Suppose~$G$ is \op, $x \in V(G)$, and~$G-x$ is connected. Then the vertices from~$N_G(x)$ lie on an induced path~$P$ in~$G-x$ such that for each biconnected component~$B$ of~$G-x$ and each pair of distinct vertices~$u,v \in V(P) \cap V(B)$ we have that~$uv \in E(G-x)$. We can find such a path in polynomial time.
\end{lemma}

\begin{proof}
 If~$|N_G(x)| = 1$ this is trivially true, so we assume~$|N_G(x)| \geq 2$ in the remainder of the proof.
 
 Consider a tree~$T$ obtained from a spanning tree of~$G-x$ by iteratively removing leaves that are not in~$N_G(x)$. We show~$T$ is a path. If~$T$ contains a vertex~$y$ of degree at least~$3$ then~$T-y$ contains three components containing a neighbor of~$x$ and, since~$T$ is connected, neighbors of~$y$. This forms a~$K_{2,3}$-minor in~$G$ contradicting outerplanarity of~$G$. Hence~$T$ is a path, and by construction both its leaves are a neighbor of~$x$. We now describe how to obtain the desired induced path~$P$ from~$T$. If~$T$ is not an induced path in~$G$, there are two nonconsecutive vertices $u,v$ in~$T$ with $uv \in E(G)$. If there is no vertex~$w \in N_G(x)$ between~$u$ and~$v$ on~$T$, then the path~$T'$ obtained from~$T$ by replacing the subpath between~$u$ and~$v$ with the edge~$uv$ is a shorter path containing all of~$N_G(x)$. Exhaustively repeat this shortcutting step and call the resulting path~$P$. Since this procedure does not affect the first an last vertices we know the first and last vertex of~$P$ are both neighbors of~$x$. All operations to obtain~$P$ can be performed in polynomial time.
 
 If the path~$P$ obtained after shortcutting is not an induced path in~$G$, there are two nonconsecutive vertices $u,v$ in~$P$ with $uv \in E(G)$. By construction of~$P$ we know that there is a vertex~$w \in N_G(x)$ between~$u$ and~$v$ on~$P$. Now~$G$ contains a~$K_4$-minor since~$\{x,u,v,w\}$ are pairwise connected by internally vertex-disjoint paths, contradicting outerplanarity of~$G$. So~$P$ is an induced path.
 
 Let~$B$ be an arbitrary biconnected component of $G-x$ and let $u,v \in B \cap P$ be distinct. If~$uv$ is a bridge in~$G-x$, it is trivial to see that $uv \in E(G-x)$, so we can assume that~$B$ contains at least 3 vertices (so $B$ is 2-connected). We first consider the case where~$u$ is the first vertex along~$P$ that is contained in~$B$ and~$v$ is the last. Since the first and the last vertex of~$P$ are neighbor to~$x$,
 \mic{$(u,x,v)$ forms a $(u,v)$-path in $G$.}
 Since~$u,v \in B$ and~$B$ is 2-connected, there is a cycle within~$B$ containing~$u$ and~$v$. This gives us two internally vertex-disjoint paths from~$u$ to~$v$ within~$B$. \mic{If~$uv \not\in E(G-x)$, these paths have non-empty interiors. Together with the path $(u,x,v)$, this leads to a $K_{2,3}$-minor in $G$ and contradicting its outerplanarity.} 
 Hence, we can assume that~$uv \in E(G-x)$.
 
 If~$u$ and~$v$ are not the first and last vertices along~$P$ contained in~$B$, then there are two vertices~$u'$ and~$v'$ that are. Then~$u'v'$ is an edge in $G-x$ and because~$P$ is an induced path, $u'$ and~$v'$ have to be consecutive in~$P$, contradicting existence of \mic{such a pair~$(u,v)$.}
\end{proof}

We now investigate what happens when a modulator vertex~$x \in X_0 \cup X_1$ is the only vertex in~$X_0 \cup X_1$ that is adjacent to a connected component~$C$ of~$G-(X_0 \cup X_1 \cup Z)$. If~$x$ has sufficiently many edges to a part of~$C$ that is not adjacent to~$Z$, then one of these edges can be removed without affecting the \op deletion number~$\opd(G)$.
We will \mic{also} exploit this property for a reduction rule later in this paper when we reduce the number of edges within a connected component of~$G-(X_0 \cup X_1 \cup Z)$.

\begin{lemma} \label{reduction:generalized-fan-rule}
Suppose we are given a graph~$G$, a vertex~$x \in V(G)$, and five vertices~$v_1, \ldots, v_5 \in N_G(x)$ that lie, in order of increasing index, on an induced path~$P$ in~$G-x$ from~$v_1$ to~$v_5$, such that~$N_G(x) \cap V(P) = \{v_1, \ldots, v_5\}$. Let~$C$ be the component of~$G-\{v_1, v_5, x\}$ containing~$P-\{v_1, v_5\}$. If~$G\brb C$ is \op, then~$\opd(G) = \opd(G \setminus xv_3)$.
\end{lemma}
\begin{proof}
 Clearly for any~$S \subseteq V(G)$ if~$G-S$ is \op, then~$G \setminus xv_3 - S$ is also \op, hence~$\opd(G) \geq \opd(G \setminus xv_3)$. To show~$\opd(G) \leq \opd(G \setminus xv_3)$, suppose~$G \setminus xv_3 - S$ is \op for some arbitrary~$S \subseteq V(G)$. If~$x \in S$ or~$v_3 \in S$ then clearly~$G-S$ is \op, so suppose~$x,v_3 \not\in S$. We show~$G-S'$ is \op for some~$S' \subseteq V(G)$ with~$|S'| \leq |S|$. Consider the following cases:
 
 \begin{enumerate}
  \item 
  If~$|S \cap V(P)| = 0$ then~$G\setminus xv_3 - S$ contains an induced cycle formed by~$x$ together with the subpath of~$P$ from~$v_2$ to~$v_4$. This cycle includes~$x$ and~$v_3$, so by \cref{lem:induced-cycle} the graph~$G \setminus xv_3 - S$ remains \op after adding the edge~$xv_3$, hence~$G-S$ is \op.
  
  \item 
  If~$|S \cap V(P)| \geq 2$ then let~$S' := \{v_1, v_5\} \cup (S \setminus V(C))$. Since~$|S'| \leq |S|$, showing that~$G-S'$ is \op proves the claim. Let~$\overline{C} := G - V(C)$ and note that~$\overline{C} - S'$ is \op since it is a subgraph of~$G \setminus xv_3 - S$. Also note that~$G[V(C) \cup \{x\}]$ is \op since it is a subgraph of~$G[V(C) \cup \{v_1, v_5, x\}] = G \brb C$. Since for any connected component~$H$ of~$G-S'-x$ the graph~$(G-S')\brb H$ is a subgraph of~$\overline{C} - S'$ or~$G[V(C) \cup \{x\}]$ we have that~$(G-S')\brb H$ is \op. Then by \cref{criterion:articulation-point} the graph~$G-S'$ is \op.
  
  \item
  If~$|S \cap V(P)| = 1$ then let~$u \in S \cap V(P)$ and assume without loss of generality that~$u$ lies on the subpath of~$P$ from~$v_3$ to~$v_5$, so the subpath of~$P$ from~$v_1$ to~$v_3$ does not contain vertices of~$S$ (recall that~$v_3 \not\in S$). Let~$S' := \{v_5\} \cup (S \setminus V(C))$ and note that~$|S'| \leq |S|$. We shall show that~$G-S'$ is \op. Since~$x, v_1 \not\in S$, we have that also~$x,v_1 \not\in S'$, so~$xv_1 \in E(G-S')$. In order to apply \cref{criterion:edge-removal} to~$G-S'$ and~$xv_1$ we have to show that
  \begin{itemize}
    \item for each connected component~$C'$ of~$G-S'-\{v_1,x\}$ the graph~$(G-S')\brb {C'}$ is \op, and
    \item there are at most two induced internally vertex-disjoint $(v_1,x)$-paths in~$(G-S') \setminus v_1x$.
  \end{itemize}
  
  Because~$v_5 \in S'$ we have~$G-S'-\{v_1,x\} = G  - \{v_1,v_5,x\} - S'$ and since~$C$ is a connected component of~$G - \{v_1,v_5,x\}$ we have that all connected components of $G-S' - \{v_1,x\}$ are either a connected component of~$C-S' = C$ or of~$G-S'-\{v_1,x\}-V(C)$. It is given that~$C$ is connected and~$G[V(C) \cup \{v_1,v_5,x\}]$ is \op so then $G[V(C) \cup \{v_1,x\}]=(G-S')\brb C$ is also \op.
  Any other connected component~$C'$ is a connected component of~$G-S'-\{v_1,x\}-V(C)$, so we have that~$(G-S')\brb {C'}$ is a subgraph of~$G-S'-V(C)$. This is in turn, a subgraph of $G \setminus xv_3 - S$ which is \op. Hence~$(G-S')\brb {C'}$ is \op.
  
  It remains to show that there are at most two induced internally vertex-disjoint $(v_1,x)$-paths in~$(G-S') \setminus v_1x$. 
  Suppose for contradiction that~$(G-S') \setminus v_1x$ contains three induced vertex-disjoint $(v_1,x)$-paths. As shown before,~$C$ is a connected component of~$G-S'-\{v_1,x\}$ adjacent to~$v_1$ and~$x$, so there exists an induced $(v_1,x)$-path~$P_1$ in~$G-S' \setminus v_1x$ whose internal vertices all lie in~$C$. Since~$G\brb C$ is \op \bmp{ and~$C$ is connected}, by \cref{lem:paths-in-different-components} \bmp{the graph~$G \brb C$ does not contain two vertex-disjoint~$(v_1,x)$-paths with nonempty interiors.} Hence there are two induced internally vertex-disjoint $(v_1,x)$-paths~$P_2, P_3$ in~$(G-S' \setminus v_1x)-V(C)$. Observe that~$P_2$ and~$P_3$ are then disjoint from~$S \setminus S' \subseteq V(C)$ and do not contain~$xv_3$. It follows that~$P_1$, $P_2$ and~$P_3$ are three induced internally vertex-disjoint $(v_1,x)$-paths in~$G\setminus xv_3 - S$, contradicting its outerplanarity by \cref{criterion:edge-removal}. We conclude also the second condition of \cref{criterion:edge-removal} holds for~$G-S'$ and the edge~$v_1x$, hence~$G-S'$ is \op.
  %
  %
  \qedhere
 \end{enumerate}
\end{proof}

We now use the properties of the~$(k,c,d)$-\op decomposition to show that any solution of size at most~$k$ contains all but possibly one vertex from~$(X_0 \cup X_1) \cap N_G(C)$, where~$C$ is a connected component from~$G-(X_0 \cup X_1 \cup Z)$. We use this fact together with the result from \cref{reduction:generalized-fan-rule} to identify an irrelevant edge, which leads to the following reduction rule:

\begin{reduction}\label{reduction:irrelevant-edge}
Given a~$(k,c,d)$-\op decomposition~$(X_0, X_1, Z)$ of a graph~$G$, a vertex~$x \in X_0 \cup X_1$, and five vertices~$v_1, \ldots, v_5 \in N_G(x) \setminus (X_0 \cup X_1)$ that lie, in order of increasing index, on an induced path~$P$ in~$G-(X_0 \cup X_1)$ from~$v_1$ to~$v_5$, such that~$N_G(x) \cap V(P) = \{v_1, \ldots, v_5\}$. Let~$C$ be the component of~$G - (X_0 \cup X_1) - \{v_1, v_5\}$ containing~$P-\{v_1, v_5\}$. If~$V(C) \cap Z = \emptyset$ remove the edge~$xv_3$.
\end{reduction}
\begin{lemma}[Safeness]\label{lem:irrelevant-edge-safeness}
 \mic{Suppose that \cref{reduction:irrelevant-edge} removes the edge~$e = xv_3$ from a graph~$G$.
 If $\opd(G) > k$ then $\opd(G \setminus e) > k$ and
if $\opd(G) \le k$ then $\opd(G \setminus e) = \opd(G)$.}
\end{lemma}
\begin{proof}
 \mic{Clearly~$\opd(G \setminus e) \leq \opd(G)$ so it suffices to show that~$\opd(G \setminus e) \leq k$ implies~$\opd(G \setminus e) = \opd(G)$. Suppose} $G \setminus e - S$ is \op for some~$S \subseteq V(G)$ of size \bmp{at most}~$k$; \bmp{we prove~$\opd(G) \leq |S|$}.
 If~$S$ contains~$x$ or~$v_3$ then the claim is trivial. Otherwise let~$X := X_0 \cup X_1$ and~$X_S := X \cap S$ and note that~$x \notin X_S$. Since~$C$ is a connected component of~$G - X - \{v_1, v_5\}$ we have that~$N_G(C) \subseteq X \cup \{v_1, v_5\}$. We first show that~$N_G(C) \subseteq X_S \cup \{v_1, v_5, x\}$. Suppose for contradiction that some~$u \in N_G(C)$ is not contained in~$X_S \cup \{v_1, v_5, x\}$, \bmp{so that~$u \in X \setminus (S \cup \{v\})$}. Since~$x$ and~$u$ are both neighbor to~$C$, which is connected and does not contain vertices from~$X$ or~$Z$, we have that~$x$ and~$u$ are not~$(X \cup Z)$-separated. It follows from \cref{def:op-decomp}\eqref{item:decomp:pairs} that there are~$k+4$ vertex-disjoint paths, with non-empty interior, connecting~$x$ and~$u$ in~$G$. At most one of these paths contains the edge~$e$, so in $G \setminus e$ there are at least~$k+3$ internally vertex-disjoint paths, with non-empty interior, connecting~$x$ and~$u$. Since~$x, u \not\in S$, and~$|S| \leq k$ we have that~$G \setminus e - S$ has at least~$3$ internally vertex-disjoint paths, with non-empty interior, connecting~$x$ and~$v$. This contradicts outerplanarity of~$G \setminus e - S$. Hence~$N_G(C) \subseteq X_S \cup \{v_1, v_5, x\}$.
 
 Consider the graph~$G' := G - X_S$. To prove \bmp{that~$\opd(G) \leq |S|$}, it suffices to \bmp{prove~$\opd(G') \leq |S| - |X_S|$.} Observe that~$x \in V(G')$ and~$v_1, \ldots, v_5 \in N_{G'}(x)$ lie in order of increasing index on the induced path~$P$ in~$G'-x$ from~$v_1$ to~$v_5$, such that~$N_{G'}(x) \cap P = \{v_1, \ldots, v_5\}$. Let~$C'$ be the connected component of~$G'-\{v_1, v_5, x\}$ containing~$P - \{v_1, v_5\}$. In order to apply \cref{reduction:generalized-fan-rule} we show that~$G'[V(C') \cup \{v_1, v_5, x\}]$ is \op.
 
 Since~$C$ is connected and~$N_G(C) \subseteq X_S \cup \{v_1, v_5, x\}$ we have that~$C$ is a connected component of~$G - (X_S \cup \{v_1, v_5, x\}) = G' - \{v_1, v_5, x\}$. As~$C'$ is also a connected component of~$G' - \{v_1, v_5, x\}$ and both~$C$ and~$C'$ contain~$P-\{v_1,v_5\}$ they are the same connected component. It follows that~$G[V(C') \cup \{v_1, v_5, x\}] = G[V(C) \cup \{v_1, v_5, x\}]$, which is \op by \cref{def:augmentedmod} since it only intersects with~$X_0 \cup X_1$ on the single vertex~$x$.
 
 This shows \cref{reduction:generalized-fan-rule} can be applied to~$G'$ with vertex~$x$ and the path~$P$. Since~$S \setminus X_S$ forms a \bmp{size-$(|S| - |X_S|)$} \op deletion set for~$G' \setminus e$, it follows there is a \bmp{size-$(|S| - |X_S|)$} \op deletion set~$S'$ for~$G'$. Then~$S' \cup X_S$ forms a \bmp{size-$|S|$} \op deletion set for~$G$.
\end{proof}

We now show how this reduction rule can be applied to reduce the number of edges between a vertex~$x \in X_0 \cup X_1$ and a connected component in~$G-(X_0 \cup X_1 \cup Z)$ to a constant. This leads to an~$\Oh(k^4)$ bound on~$N_G(X_0 \cup X_1)$; \mic{see Figure \cref{fig:apply-irrelevant-edge}.}

\begin{figure}[bt]
  \centering
  \includegraphics{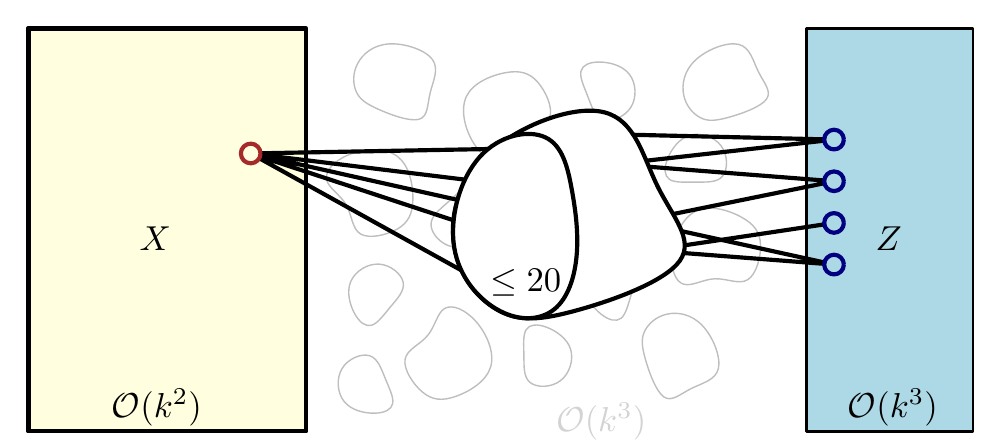}
  \caption{An illustration of \cref{lem:apply-irrelevant-edge}.
  Given a $(k,c,d)$-\op decomposition~$(X_0, X_1, Z)$ of a graph~$G$, a vertex~$x \in X = X_0 \cup X_1$ and a component~$C$ of~$G-(X_0 \cup X_1 \cup Z)$, we are guaranteed that $|N(C) \cap Z| \le 4$ and we can apply \cref{reduction:irrelevant-edge} until $|N(x) \cap V(C)| \leq 20$. \bmp{The expressions at the bottom bound the size of~$X$, the number of components of~$G - (X \cup Z)$, and size of~$Z$.}}
  \label{fig:apply-irrelevant-edge}
\end{figure}

\begin{lemma} \label{lem:apply-irrelevant-edge}
 There is a polynomial-time algorithm that, given a $(k,c,d)$-\op decomposition~$(X_0, X_1, Z)$ of a graph~$G$, a vertex~$x \in X_0 \cup X_1$ and a component~$C$ of~$G-(X_0 \cup X_1 \cup Z)$, applies \cref{reduction:irrelevant-edge} or concludes that~$|N_G(x) \cap V(C)| \leq 20$.
\end{lemma}
\begin{proof}
 We first describe the algorithm and then proceed to prove its correctness.
 \proofsubparagraph*{Algorithm}
 If~$x \not\in N_G(C)$ then conclude~$|N(x) \cap V(C)| \leq 20$. Otherwise let~$C^+ := G[V(C) \cup (N_G(C) \cap Z) \cup \{x\}]$. Apply \cref{lem:neighborhood} to find an induced path~$P$ in~$C^+$ containing all of~$N_{C^+}(x)$ (we will show~$C^+$ is \op). Let the vertices~$N_{C^+}(x) = \{v_1, \ldots, v_\ell\}$ be indexed by the order in which they occur on~$P$. For all~$1 \leq i \leq \ell-4$ let~$P_i$ be the subpath of~$P$ from~$v_i$ to~$v_{i+4}$ and let~$C_i$ denote the connected component of~$C^+ - \{v_i,v_{i+4},x\}$ containing~$P_i - \{v_i, v_{i+4}\}$. If for some~$1 \leq i \leq \ell-4$ we have~$V(C_i) \cap Z = \emptyset$ then apply \cref{reduction:irrelevant-edge} with~$x$ and~$P_i$ to remove the edge~$xv_{i+2}$. Otherwise conclude~$|N(x) \cap V(C)| \leq 20$.
 
 All operations can be performed in polynomial time.
 
 \proofsubparagraph*{Correctness}
  We first show \cref{lem:neighborhood} is applicable. Clearly~$C^+ -x = G[V(C) \cup (N_G(C) \cap Z)]$ is connected since~$C$ is connected, so it remains to show that~$C^+$ is \op. This follows from the fact that~$C^+$ is a subgraph of~$G - ((X_0 \cup X_1) \setminus \{x\})$, which is \op by \cref{lem:augmentedmod:properties}. For the remainder of the proof we first establish a number of properties of the graphs defined in the algorithm.

  \begin{countclaim} \label{claim:two-consecutive-neighbors}
   Any connected component of~$C^+ - (V(P) \cup \{x\})$ has at most two neighbors in~$V(P)$ and they must be consecutive along~$P$.
  \end{countclaim}
  \begin{claimproof}
   Consider the cycle in~$C^+$ formed by the vertices~$V(P) \cup \{x\}$ (recall the first and last vertices of~$P$ are neighbor to~$x$). Since~$C^+$ is \op the claim follows directly from \cref{lem:attaching:to:cycle}.
  \end{claimproof}
  
  \begin{countobs} \label{obs:components}
   For any~$1 \leq i \leq \ell-4$, since~$C_i$ is a connected component of~$C^+ - \{v_i, v_{i+4}, x\}$ we have that any connected component of~$C_i -V(P)$ is also a connected component of~$C^+ - \{v_i, v_{i+4}, x\} - V(P) = C^+ - (V(P) \cup \{x\})$.
  \end{countobs}
  
  \begin{countclaim} \label{claim:C_i-cap-P}
   For all~$1 \leq i \leq \ell-4$ we have~$V(C_i) \cap V(P) = V(P_i - \{v_i, v_{i+4}\})$.
  \end{countclaim}
  \begin{claimproof}
  Let~$u,v \in V(P) \cap C_i$ be distinct vertices. Since~$C_i$ is connected, there exists a path in~$C_i$ connecting~$u$ and~$v$. Let~$P'$ be a shortest $(u,v)$-path in~$C_i$. If~$P'$ contains a vertex not in~$P$ then this is a vertex in a connected component~$D$ in~$C_i - V(P)$. By \cref{obs:components} and \cref{claim:two-consecutive-neighbors} we have that~$D$ has at most two neighbors in~$P$ and they are consecutive along~$P$. Since the path~$P'$ must enter and leave~$D$, the path visits both these neighbors, however since these neighbors are adjacent we can obtain a shorter $(u,v)$-path by skipping vertices in~$D$. This contradicts that~$P'$ is a shortest $(u,v)$-path. It follows that the shortest path in~$C_i$ between any two vertices from~$P$ is a subpath of~$P$. Since~$C_i$ does not contain~$v_i$ and~$v_{i+4}$ by definition, we have~$V(C_i) \cap V(P) = V(P_i - \{v_i, v_{i+4}\})$.
  \end{claimproof}
  
  \begin{countclaim} \label{claim:N(C_i)}
    For all~$1 \leq i \leq \ell-4$, if~$V(C_i) \cap Z = \emptyset$ then $N_G(C_i) \subseteq X_0 \cup X_1 \cup \{v_i, v_{i+4}\}$.
  \end{countclaim}
  \begin{claimproof}
    Suppose for some~$1 \leq i \leq \ell-4$ that~$V(C_i) \cap Z = \emptyset$ and let~$v \in N_G(C_i)$. We show~$v \in X_0 \cup X_1 \cup \{v_i, v_{i+4}\}$. If~$v \in V(C^+)$ then since~$C_i$ is a connected component of~$C^+ - \{v_i, v_{i+4}, x\}$ we have~$v \in \{v_i, v_{i+4}, x\} \subseteq (X_0 \cup X_1 \cup \{v_i, v_{i+4}\}$, so suppose~$v \not\in V(C^+) \supseteq V(C)$. Since~$v \in N_G(C_i)$ there exists a vertex~$u \in N_G(v) \cap V(C_i)$ and note that~$u \not\in Z$ because~$N_G(C_i) \cap Z = \emptyset$. Since~$u \in V(C_i) \subseteq V(C^+ - \{v_i, v_{i+4}, x\}) \subseteq V(C) \cup (N_G(C) \cap Z)$ and $u \not\in Z$ we have~$u \in V(C)$. Because~$u$ and~$v$ are neighbors we have~$v \in N_G(C)$ so~$v \in X_0 \cup X_1 \cup Z$ since~$C$ is a connected component of~$G-(X_0 \cup X_1 \cup Z)$. Clearly if~$v \in X_0 \cup X_1$ the claim holds, so suppose~$v \in Z$. However since~$v \in N_G(C)$ and~$v \in Z$ we have~$v \in N_G(C) \cap Z$ so by definition of~$C^+$ we have~$v \in V(C^+)$, a contradiction since we assumed~$v \not\in V(C^+)$.
  \end{claimproof}
 
  Suppose that for some~$1 \leq i \leq \ell -4$ we have~$V(C_i) \cap Z = \emptyset$. In order to show that~\cref{reduction:irrelevant-edge} applies to~$x$ and~$P_i$,
  first note that~$(X_0, X_1, Z)$ is a $(k,c,d)$-outerplanar decomposition of~$G$ and~$x \in X_0 \cup X_1$. The vertices~$v_i, \ldots, v_{i+4}$ lie on~$P_i$, an induced path in~$G - (X_0 \cup X_1)$ from~$v_i$ to~$v_{i+4}$ such that~$N(x) \cap V(P_i) = \{v_i, \ldots, v_{i+4}\}$. We show that~$C_i$ is the connected component of~$G-(X_0 \cup X_1) - \{v_i, v_{i+4}\}$ containing~$P_i - \{v_i , v_{i+4}\}$.
  
  Note that~$C_i$ does not contain any vertices from~$X_0 \cup X_1 \cup \{v_i, v_{i+4}\}$ so~$C_i$ is a (connected) subgraph of~$G-(X_0 \cup X_1) - \{v_i, v_{i+4}\}$. By \cref{claim:N(C_i)} we have~$N_G(C_i) \subseteq X_0 \cup X_1 \cup \{v_1, v_{i+4}\}$. We can conclude that~$C_i$ is a connected component of~$G-(X_0 \cup X_1) - \{v_i, v_{i+4}\}$, and by definition it contains~$P_i - \{v_i, v_{i+4}\}$.
  
  Finally, observe that~$G[V(C_i) \cup \{v_i, v_{i+4}, x\}]$ is \op as it is a subgraph of~$C^+$. So since~$V(C_i) \cap Z = \emptyset$ we have that \cref{reduction:irrelevant-edge} applies.
 
  Now suppose that the algorithm was unable to apply \cref{reduction:irrelevant-edge}, i.e, for all~$1 \leq i \leq \ell-4$ we have~$V(C_i) \cap Z \neq \emptyset$. We show~$|N(x) \cap V(C)| \leq 20$.
  Suppose for contradiction that~$|N(x) \cap V(C)| > 20$. Then~$N_{C^+}(x) > 20$, so the path~$P$ contains more than 20 neighbors of~$x$, i.e.,~$\ell > 20$ so~$C_1, C_5, C_9, C_{13}, C_{17}$ are defined. Since~$V(C_i) \cap Z \neq \emptyset$ for all~$1 \leq i \leq \ell-4$ we know~$C_1, C_5, C_9, C_{13}, C_{17}$ all contain a vertex from~$Z$. We show~$C_1, C_5, C_9, C_{13}, C_{17}$ are disjoint. 
  
  If~$C_1, C_5, C_9, C_{13}, C_{17}$ are not disjoint, then there exist integers~$i,j \in \{1,5,9,13,17\}$ and a vertex~$v$ such that~$i < j$ and~$v \in V(C_i) \cap V(C_j)$. Using \cref{claim:C_i-cap-P} we find that~$V(P) \cap V(C_i) \cap V(C_j) \subseteq (V(C_i) \cap V(P)) \cap (V(C_j) \cap V(P)) = V(P_i - \{v_i, v_{i+4}\}) \cap V(P_j - \{v_j, v_{j+4}\}) = \emptyset$, so~$v \not\in V(P)$. Then~$v$ is a vertex in some connected component~$D$ of~$C_i - V(P)$ and a connected component~$D'$ of~$C_j - V(P)$. By \cref{obs:components}, both~$D$ and~$D'$ are connected components of~$C^+ - (V(P) \cup \{x\})$, and since both contain~$v$, they are the same connected component.
  Since~$C_i$ is connected,~$D$ must contain a neighbor~$u_1 \in V(C_i) \cap V(P) = V(P_i - \{v_i, v_{i+4}\})$. Similarly~$D'=D$ must contain a neighbor~$u_2 \in V(C_j) \cap V(P) = V(P_j - \{v_j, v_{j+4}\})$. Since these two sets are disjoint we have~$u_1 \neq u_2$. By \cref{claim:two-consecutive-neighbors} these neighbors must be the only neighbors of~$D$ and they must be consecutive along~$P$. However the vertex~$v_{i+4}$ lies on~$P$ between~$u_1$ and~$u_2$ since~$i+4 \leq j$ so~$u_1$ and~$u_2$ are not consecutive along~$P$. By contradiction,~$C_1, C_5, C_9, C_{13}, C_{17}$ are disjoint.
  
  Since~$C_1, C_5, C_9, C_{13}, C_{17}$ are disjoint subgraphs of~$C^+$ and each subgraph contains a vertex from~$Z$, we have that~$|Z \cap V(C^+)| \geq 5$. By definition of~$C^+$ we know~$V(C^+) = V(C) \cup \{x\} \cup (N(C) \cap Z)$. Recall that~$|N(C) \cap Z| \leq 4$ by \cref{def:op-decomp}\eqref{item:decomp:neighbors}, so then~$(V(C) \cup \{x\}) \cap Z \neq \emptyset$. This is a contradiction since~$x \in X_0 \cup X_1$ and~$C$ is a connected component of~$G-(X_0 \cup X_1 \cup Z)$, hence~$|N(x) \cap V(C)| \leq 20$.
\end{proof}

\mic{We are going to apply \cref{lem:apply-irrelevant-edge} to a computed \op decomposition in order to reduce the total neighborhood size of $X_0 \cup X_1$.
This allows us to construct a final modulator~$L$ of size~$\Oh(k^4)$ with a structure referred to in previous works as a protrusion decomposition.
}
\mic{We can now proceed to proving a lemma that encapsulates application of \cref{reduction:irrelevant-edge}.}

\begin{lemma}\label{lem:final-decomposition}
There exists a function~$f_5 \colon \setN^2 \rightarrow \setN$ and a polynomial-time algorithm that, given a $(k,c,d)$-\op decomposition $(X_0, X_1, Z)$ of a graph $G$, either applies \cref{reduction:irrelevant-edge} or \cref{reduction:articulation-point}, or outputs a set $L \subseteq V(G)$ such that
\begin{enumerate}
    \item \label{lem:decomp2:size} $|L| \le f_5(c,d) \cdot (k+3)^4$,
    \item \label{lem:decomp2:edges} $|E_G(L,L)| \le f_5(c,d) \cdot (k+3)^4$,
    \item \label{lem:decomp2:comps} there are at most $f_5(c,d) \cdot (k+3)^4$ connected components in $G- L$, and
    \item \label{lem:decomp2:nbr} for each connected component $C$ of $G-L$ the graph $G \brb C$ is \op and $|N_G(C)| \le 4$.
\end{enumerate}
Furthermore, $f_5(c,d) = 24 \cdot (20 \cdot f_4(c,d) + d + c + c^2)$ (see \cref{lem:modulator:k4}).
\end{lemma}
\begin{proof}
 We first describe the algorithm and then proceed to prove its correctness.
 
 \proofsubparagraph*{Algorithm}
 For all~$x \in X_0 \cup X_1$ and connected components~$C$ in~$G-(X_0 \cup X_1 \cup Z)$ we run the algorithm from \cref{lem:apply-irrelevant-edge} to apply \cref{reduction:irrelevant-edge} or conclude that~$|N_G(x) \cap V(C)| \leq 20$. If \cref{reduction:irrelevant-edge} could not be applied to any~$x$ and~$C$, we take~$X = X_0 \cup X_1$ and
 apply \cref{lem:modulator:lca} on the graph~$G-X$ with vertex set~$Z \cup N_G(X)$ to obtain a set~$Z' \subseteq V(G) \setminus X$.
 We set $L = X \cup Z'$.
 If some component~$C$ of $G -  L$ has at most one neighbor, we apply \cref{reduction:articulation-point} to remove~$C$. Otherwise we return~$L$.
 
 \proofsubparagraph*{Correctness}
 It can easily be seen that \cref{lem:apply-irrelevant-edge} applies on all~$x \in X_0 \cup X_1$ and connected components~$C$ in~$G-(X_0 \cup X_1 \cup Z)$. If by calling \cref{lem:apply-irrelevant-edge} we have applied \cref{reduction:irrelevant-edge} we can terminate the algorithm. Otherwise it holds that~$|N_G(x) \cap V(C)| \leq 20$ for each~$x \in X_0 \cup X_1$ and each  connected component~$C$ of~$G-(X_0 \cup X_1 \cup Z)$.
 Let us examine~$Z'$ and $L$ given by the execution of the algorithm.
 
 Clearly~$G-X$ is \op as it is a subgraph of~$G-X_0$, which  justifies that the algorithm correctly applies \cref{lem:modulator:lca}. To show Condition~\ref{lem:decomp2:size} and~\ref{lem:decomp2:edges}, we first prove a bound on~$|E_G(X,V(G) \setminus (X \cup Z)|$.
 
 Consider the component graph~$H = \mathcal{C}(G,X \cup Z)$. For any~$x \in X$ and connected component~$C$ of~$G-(X \cup Z)$ if~$H$ does not contain an edge between~$x$ and the vertex representing~$C$, then~$|N_G(x) \cap V(C)| = 0$. If~$H$ contains an edge between~$x$ and the vertex representing~$C$, then our earlier bound applies:~$|N_G(x) \cap V(C)| \leq 20$. By \cref{lem:modulator:k4} we have that~$H$ contains at most~$f_4(c,d) \cdot (k+3)^4$ edges, so $|E_G(X,V(G) \setminus (X \cup Z))| \le 20 \cdot f_4(c,d) \cdot (k+3)^4$ and
 $|N_G(X) \setminus Z| \leq 20 \cdot f_4(c,d) \cdot (k+3)^4$.

 By \cref{lem:modulator:lca} we have~$|Z'| \leq 6 \cdot |Z \cup N_G(X)|$ and by \cref{def:op-decomp}\eqref{def:op-decomp:size} we have~$|Z| \leq d \cdot (k+3)^3$. 
 To show Condition~\ref{lem:decomp2:size} we check that
\[
 |Z'| \leq 6 \cdot (20 \cdot f_4(c,d) + d)) \cdot (k+3)^4,\text{ and}
 \]
 \[
 |L| = |Z'| + |X| \le 6 \cdot (20 \cdot f_4(c,d) + d + c)) \cdot (k+3)^4.
 \]

 Let us now bound the number of edges in $E_G(L,L)$.
 We group these edges into four categories: (a)~edges within $X_0$, (b) edges between $X_0$ and $X_1 \cup Z$, (c) edges between $X_0$ and $L \setminus (X \cup Z)$, and (d)~edges within $L \setminus X_0$.
 The number of edges in (a) is clearly at most $|X_0|^2 = c^2 \cdot k^2$.
 Similarly, in case (b) we obtain the bound $|X_0| \cdot |X_1 \cup Z| = ck \cdot (3c \cdot k^2 + d \cdot (k+3)^3) \le (3c^2+d) \cdot (k+3)^4$.
 To handle case (c), observe that $E_G(X_0, L \setminus (X \cup Z)) \subseteq E_G(X,V(G) \setminus (X \cup Z))$ and the size of this set has already been bounded by $20 \cdot f_4(c,d) \cdot (k+3)^4$.
 Finally, the subgraph of $G$ induced by $L \setminus X_0$ is \op and by \cref{lem:prelim:op-edges} we bound the number of edges in case (d) by $2\cdot |L| \le 2 \cdot 6 \cdot (20 \cdot f_4(c,d) + d + c)) \cdot (k+3)^4$.
 By collecting all summands we obtain that $|E_G(L,L)| \le 13 \cdot (20 \cdot f_4(c,d) + d + c + c^2)) \cdot (k+3)^4$ and prove Condition~\ref{lem:decomp2:edges}.
 
 To show Condition~\ref{lem:decomp2:nbr} note that~$Z \cup N_G(X) \subseteq Z'$ by \cref{lem:modulator:lca} and so $Z \cup N_G[X] \subseteq L$. Consider a~connected component~$C$ of~$G-L$. Since~$C$ does not contain neighbors of~$X$ we have~$N_G[C] \cap X = \emptyset$. So then~$G\brb C$ is a subgraph of~$G-X$, hence it is \op.
 Furthermore, by \cref{lem:modulator:lca} we know that~$|N_{G-X}(C)| \leq 4$ and so $|N_{G}(C)| \leq 4$.

 If ~$|N_G(C)| \leq 1$ for some connected component~$C$ of~$G-L$, we have applied \cref{reduction:articulation-point} and terminated the algorithm. If the algorithm is unable to apply this reduction rule, we know that all components of~$G-L$ have at least two neighbors, which must belong to $L \setminus X$. The vertices representing these components in the component graph~$\mathcal{C}(G-X, L)$ all have degree at least~2. Note also that this graph is bipartite (by definition) and \op since it is a minor of~$G-X$, which is \op. It follows from \cref{lem:outerplanar-bipartite} that~$G-L$ has at most~$4 \cdot |L| \leq 4 \cdot 6 \cdot (20 \cdot f_4(c,d) + d + c)) \cdot (k+3)^4$ components. This shows that Condition~\ref{lem:decomp2:comps} holds.
\end{proof}

\section{Compressing the outerplanar subgraphs} \label{sec:reduction-rules}

\subsection{Reducing the number of biconnected components}

Once we arrive at the decomposition from \cref{lem:final-decomposition}, it remains to compress  outerplanar subgraphs with a small boundary.
First, we present a reduction to bound the number of biconnected components in such a subgraph.
It will also come in useful later, for reducing the maximum size of a face in a biconnected \op graph with a small boundary. 
Intuitively, this reduction checks whether an \op subgraph with exactly two non-adjacent neighbors can supply one or two vertex-disjoint paths to the rest of the graph and replaces this subgraph with a minimal gadget with the same property, see also \cref{fig:reductions-2-nbrs}.

\begin{reduction} \label{reduction:replace-component}
Consider a graph $G$ and vertex set $C \subseteq V(G)$ such that $N_G(C) = \{x,y\}$, $xy \not\in E(G)$, $G[C]$ is connected, and $G\langle C \rangle$ is \op.
Let $P = (u_1, u_2, \dots, u_m)$, $u_1 = x$, $u_m = y$ be any shortest path connecting $x$ and $y$ in $G\langle C \rangle$ and $D_1, D_2, \dots, D_\ell$ be the connected components of $G\langle C \rangle - V(P)$.
We consider 3 cases:
\begin{enumerate}
    \item if there is a component $D_i$, for which $N_G(D_i)$ includes two non-consecutive elements of $P$, replace $C$ with two vertices $c_1, c_2$, each adjacent to both $x$ and~$y$, \label{item:replace-component:1}
    \item if there are two distinct components $D_i, D_j$, for which $|N_G(D_i) \cap N_G(D_j)| \ge 2$, replace $C$ with two vertices $c_1, c_2$, each adjacent to both $x$ and~$y$, \label{item:replace-component:2}
    \item otherwise replace $C$ with one vertex $c_1$ adjacent to both $x$ and~$y$. \label{item:replace-component:3}
\end{enumerate}
\end{reduction}

\begin{lemma}\label{lem:replace-component:minor}
Let $x, y \in V(G)$ and $C \subseteq V(G)$ be such that \cref{reduction:replace-component} applies and let~$G'$ be the graph obtained after application of the rule. Then $G'$ is a minor of $G$. 
\end{lemma}
\begin{proof}
In case (\ref{item:replace-component:1}), there exists a component $D_i$ adjacent to non-consecutive vertices $u_j$, $u_h$ from $V(P)$, $j<h$.
Let $P_x,P_C,P_y$ denote the non-empty subpaths: $(u_1, \dots, u_j)$  $(u_{j+1}, \dots, u_{h-1})$, $(u_{h}, \dots, u_{m})$.
First, we remove all the connected components of $G[N[C]] - V(P)$ different from $D_i$.
Next, we contract $P_x$ into $x$, $P_y$ into $y$, $P_C$ into a vertex denoted $c_1$, and
 $D_i$ into a  vertex denoted $c_2$.
 By the choice of $D_i$ we see that each of $c_1, c_2$ is adjacent to both $x, y$, therefore we have obtained $G'$ through vertex deletions and edge contractions.
 
In case (\ref{item:replace-component:2}), there exist distinct components $D_{i_1}, D_{i_2}$ both adjacent to vertices $u_j$, $u_h$ from $V(P)$, $j<h$.
Let $P_x,P_y$ denote the subpaths $(u_1, \dots, u_j)$ and $(u_{h}, \dots, u_{m})$.
Again, we begin by removing all the connected components of $G[N[C]] - V(P)$ different from $D_{i_1}, D_{i_2}$.
Next, we contract $P_x$ into $x$, $P_y$ into $y$, $D_{i_1}$ into a vertex denoted $c_1$, and
 $D_{i_2}$ into a  vertex denoted $c_2$, thus obtaining $G'$.
 
In case (\ref{item:replace-component:3}), we simply contract $C$ into a vertex $c_1$.
\end{proof}

In order to show correctness of the reduction rule, we will prove that any \op deletion set in the new instance can be turned into an \op deletion set in the original instance without increasing its size.
If we replaced the vertex set $C$ with two vertices,  we show that any \op deletion set must break all the connections between the neighbors of $C$ which go outside $C$.
In the other case, when we replaced $C$ with just one vertex, we show that we can undo the graph modification from \cref{reduction:replace-component} while preserving the outerplanarity.

\begin{lemma}\label{lem:replace-component:correctness}
Let $x, y \in V(G)$ and $C \subseteq V(G)$ be such that \cref{reduction:replace-component} applies and let~$G'$ be the graph obtained after application of the rule. If~$S' \subseteq V(G')$ is an \op deletion set in $G'$, then there exists a set $S \subseteq V(G)$ such that $|S| \leq |S'|$ and which is an \op deletion set in $G$. 
\end{lemma}
\begin{proof}
Let $C' \subseteq V(G')$ \bmp{consist} of the vertices put in place of $C$, that is, $c_1$ and, if we replaced $C$ with two vertices, $c_2$.
We naturally identify the elements of $V(G') \setminus C'$ with $V(G) \setminus C$.
In particular, $N_{G'}(C') = N_G(C) = \{x,y\}$.
We consider \bmp{four} cases:
\begin{itemize}
    \item \bmp{$S' \cap \{x,y\} \neq \emptyset$. We show that~$G - S'$ is~\op. If~$S' \supseteq \{x,y\}$ then this is immediate since~$G[C]$ is~\op by assumption and forms a connected component of~$G-S'$, while~$(G - S') - V(C)$ is a subgraph of~$G' - S'$. 
    
    Otherwise, let~$z = S' \cap \{x,y\}$ and~$\overline{z} = \{x,y\} \setminus \{z\}$. As before,~$(G - S') - V(C)$ is \op since it is a subgraph of~$G' - S'$. The graph~$G - S'$ can be obtained from~$(G - S') - V(C)$ by attaching~$G[C]$ onto the articulation point~$\overline{z}$, and is therefore \op by \cref{criterion:articulation-point}.}
    \item $S' \cap V(C') \ne \emptyset$. We define the set $S$ as $(S' \setminus V(C')) \cup \{x\}$.
    It clearly holds that $|S| \le |S'|$.
    Furthermore, $y$ is an articulation point in $G-S$.
    The graph $G - (S \cup C \cup \{x\})$ is isomorphic with
    $G' - (S' \cup C' \cup \{x\})$, hence it is \op.
    On the other hand, $G[C \cup \{y\}]$ is \op by assumption.
    Therefore, all the components obtained by splitting $G-S$ at $y$ are \op and thus $G-S$ is \op
    by \cref{criterion:articulation-point}. 
    
     \item $S' \cap N_{G'}[C'] = \emptyset$ and $C' = \{c_1, c_2\}$.
     We can simply write $S = S'$ as we have identified elements of $V(G') \setminus C'$ and $V(G) \setminus C$. 
     Let $F_1,F_2\dots,F_\ell$ be the connected components of $G'-(S' \cup C' \cup \{x,y\})$.
     Observe that no $F_i$ can be adjacent to both $x,y$,
     as otherwise $x,y,c_1,c_2,F_i$ would form branch sets of a $K_{2,3}$-minor in $G'-S'$.
     The graph $G-S$ can be obtained from $G \brb C$ by appending the components $F_1,F_2\dots,F_\ell$ at $x$ or $y$.
     For each $i\in[\ell]$ it holds that $G \brb {F_i}$ is a subgraph of $G'-S'$, so it is \op.
     From \cref{criterion:articulation-point} we infer that $G-S$
     is \op.
  
    \item $S' \cap N_{G'}[C'] = \emptyset$ and $C' = \{c_1\}$.
    We again set $S=S'$ via vertex identification and we are going to transform $G'-S'$ into $G-S$ while preserving outerplanarity of the graph.
    Note that~the path $P$ contains at least one vertex from $C$ as $xy \not\in E(G)$.
    Subdividing a~subdivided edge multiple times preserves outerplanarity, and so does replacing $(x, c_1, y)$ with~$P$. 
    Let~$G''$ denote the resulting graph.
    
    Since $P$ is a~shortest $(x,y)$-path in $G\brb C$, there are no edges in $G\brb C$ connecting non-adjacent vertices of $P$.
    Recall that $D_1, D_2, \dots, D_\ell$ are the connected components of $G\brb C - V(P)$.
    Since~$C' = \{c_1\}$, the conditions from cases (\ref{item:replace-component:1}, \ref{item:replace-component:2}) in \cref{reduction:replace-component} are not satisfied.
    Therefore each component $D_i$ is either adjacent to one vertex from $V(P)$ or to two vertices which are consecutive.
    Furthermore, for any pair of consecutive vertices on $P$, there can be only one component $D_i$ adjacent to both of them. 
    
    For each $i \in [\ell]$ it holds that $N_G[D_i] \subseteq N_G[C]$, so $G \brb {D_i}$ is \op.
    \mic{
    If~$D_i$ has two neighbors~$u,v$ in~$P$ then any $(u,v)$-path in~$G'' \setminus uv$ includes~$x$ or~$y$ as an internal vertex, hence there cannot be two induced internally vertex-disjoint $(u,v)$-paths in~$G'' \setminus uv$.
    By \cref{lem:paths-in-different-components} appending~$D_i$ to the edge~$uv$ in~$G''$ supplies at most one more induced $(u,v)$-path  and no other~$D_j$, $j \in [\ell] \setminus \{i\}$, can supply a $(u,v)$-path in~$G'' \setminus uv$, so this preserves outerplanariy due to \cref{criterion:edge-removal} applied to the edge $uv$.}
    Next, by \cref{criterion:articulation-point} the graph obtained by appending each component adjacent to a single vertex is still \op. 
    We have replaced $c_1$ back with $C$, thus transforming $G'-S'$ into $G-S$, while preserving outerplanarity of the graph, hence $G-S$ is \op.
\end{itemize}
\bmp{As~$N_{G'}[C'] = V(C') \cup \{x,y\}$, the case distinction is exhaustive and completes the proof.}
\end{proof}

\begin{lemma}\label{lem:replace-component:safeness}
Let $G$ be a graph and $G'$ be obtained from $G$ by applying
\cref{reduction:replace-component}.
Then $\opd(G') = \opd(G)$.
\end{lemma}
\begin{proof}
By \cref{lem:replace-component:minor} we know that $G'$ is a~minor of $G$, so $\opd(G') \le \opd(G)$.
On the other hand, if $G'$ admits an \op deletion set of size at most $\ell$, then \cref{lem:replace-component:correctness} guarantees that the same holds for $G$.
\end{proof}

We are now going to make use of \cref{reduction:replace-component} to reduce the number of biconnected components in an \op graph with a small boundary.
Recall that the {block-cut tree} of a graph $H$
has a vertex for each biconnected component of $H$ and for each articulation point in $H$. 
A biconnected component $B$ and an articulation point $v$
are connected by an edge if $v \in B$.
We will show that when the block-cut tree of $H = G\brb A$ is large then we can always find either one or two articulation points that cut off an \op subgraph which can be either removed or compressed. 

\begin{lemma}\label{lem:replace-component:apply}
Consider a graph $G$ and a vertex set $A \subseteq {V(G)}$, such that $|N_G(A)| \le 4$, $G[A]$ is connected, and $G\brb A$ is \op.
There is a polynomial-time algorithm that, given $G$ and $A$ satisfying the conditions above, outputs either
\begin{enumerate}
    \item a block-cut tree of $G\brb A$ with at most 25 biconnected components, where each such biconnected component $B$ satisfies $|\partial_G(B)| \le 4$, or
    \item a vertex set $C \subseteq A$, to which either \cref{reduction:articulation-point} or \cref{reduction:replace-component} applies and decreases the number of vertices in the graph.
\end{enumerate}
\end{lemma}
\begin{proof}
We begin with computing
the block-cut tree $T$ of $G\brb A$ and rooting it at an~arbitrary node. \bmpr{Add citation for block-cut tree computation.}
For a node $t \in V(T)$ let $\chi(t)$ denote the vertex set represented by $t$, either a biconnected component, or a single vertex of $t$ that corresponds to an articulation point.
Note that each leaf in $T$ must represent a biconnected component.
Furthermore, observe that no vertex from $N_G(A)$ can be an articulation point in $G\brb A$, because $G[A]$ is connected.
Therefore for each $v \in N_G(A)$ there is a unique biconnected component containing $v$.
Let $t_v \in V(T)$ be the node in the block-cut tree representing this component.

First suppose that there exists a biconnected component $B$ of $G\brb A$ with $|\partial_G(B)| > 4$.
The set $\partial_G(B)$ is a disjoint union of $N_G(A) \cap B$ and the articulation points of $G\brb A$ lying in $B$.
Let $d = |N_G(A) \cap B|$.
Then there are at least $5-d$ articulation points of $G\brb A$ lying in $B$, but at most $4-d$ vertices of $N_G(A) \setminus B$.
By a counting argument, there is one articulation point $v \in B$ of $G\brb A$ which separates $B \setminus \{v\}$ from a set $C \subseteq N_G[A]$ which do not contain any vertex from $N_G(A)$.
Hence, $C \subseteq A$ and it induces a connected \op subgraph of $G$ having exactly one neighbor in $G$.
Therefore, \cref{reduction:articulation-point} applies for $C$ and removes it, decreasing the number of vertices in $G$.

Suppose for the rest of the proof that there at least 26 biconnected compoenents of $G\brb A$.
Let $L \subseteq V(T)$ denote the LCA closure of the set $\{t_v \mid v \in N_G(A)\}$.
By \cref{lem:lca:basic} we know that $|L| \le 7$ and each connected subtree of $T - L$ then is adjacent to at most two nodes from $L$.
It follows that if $t \in V(T) \setminus L$, then $\chi(t) \cap N_G(A) = \emptyset$.

Suppose that some component $T_C$ of $T-L$ is adjacent to just one node from $L$.
It is either an articulation point or its neighbor in $T_C$ is an articulation point.
The set $C = \bigcup_{t \in T_C} \chi(t)$
has an empty intersection with $N_G(A)$ and
contains a vertex $v$ which separates $C \setminus \{v\}$ from the rest of the graph $G$.
Therefore we can find a subset $C' \subseteq C \setminus \{v\}$ which induces a connected subgraph of $G$,
has exactly one neighbor $v$, and $C' \cup \{v\} \subseteq N_G[A]$ induces an \op graph.
Hence, \cref{reduction:articulation-point} applies for $C'$ and removes it, decreasing the size of the graph.

A similar situation occurs when $T_C$ has a vertex of degree at least 3.
Then there exists a leaf $t$ in $T_C$ which represents a biconnected component and is not adjacent to $L$, so it is also a leaf in $T$.
Again, $\chi(t) \cap N_G(A) = \emptyset$, so
$\chi(t)$ contains a single vertex $v$ which separates $\chi(t) \setminus \{v\}$ from the rest of the graph $G$.
Analogously as above, \cref{reduction:articulation-point} applies and decreases the size of the graph.

Suppose now that the previous cases do not hold. Then each connected component of $T-L$ is adjacent to exactly two nodes from $L$ and induces a path in $T$ with the endpoints adjacent to $L$.
If we contracted each such component to an edge connecting two nodes from $L$, we would obtain a tree with vertex set $L$ and $|L|-1$ edges.
Hence, the number of such components of $T-L$ is at most 6.
Since the total number of biconnected components in $G\brb A$ is at least 26, we infer that there exists a~connected component $T_C$ of $T-L$ containing
at least $\lceil\frac{26 - 7}{6}\rceil = 4$ nodes representing biconnected components of $G\brb A$.
The set $C = \bigcup_{t \in T_C} \chi(t)$ contains two vertices $u,v$ which together separate $C - \{u,v\}$ from the rest of $G\brb A$.
Note that $u,v$ belong to disjoint biconnected components of $G\brb A$, so $uv \not\in E(G)$ and the set $C - \{u,v\}$ induces a connected subgraph of $G$.
As $C \cap N_G(A) = \emptyset$,
this implies that $C - \{u,v\}$ is a connected component of $G - \{u,v\}$.
Furthermore, a union of 4 biconnected components has at least 5 vertices (the corner case occurs when they are all single edges),
so $C - \{u,v\}$ has at least 3 vertices.
Therefore, \cref{reduction:articulation-point} applies for $C- \{u,v\}$ and replaces it with at most two new vertices, therefore
the total number of vertices in the graph decreases.
All the described operations on the block-cut tree can be implemented to run in polynomial time.
\end{proof}

\begin{figure}[bt]
  \centering
  \includegraphics[width=\linewidth]{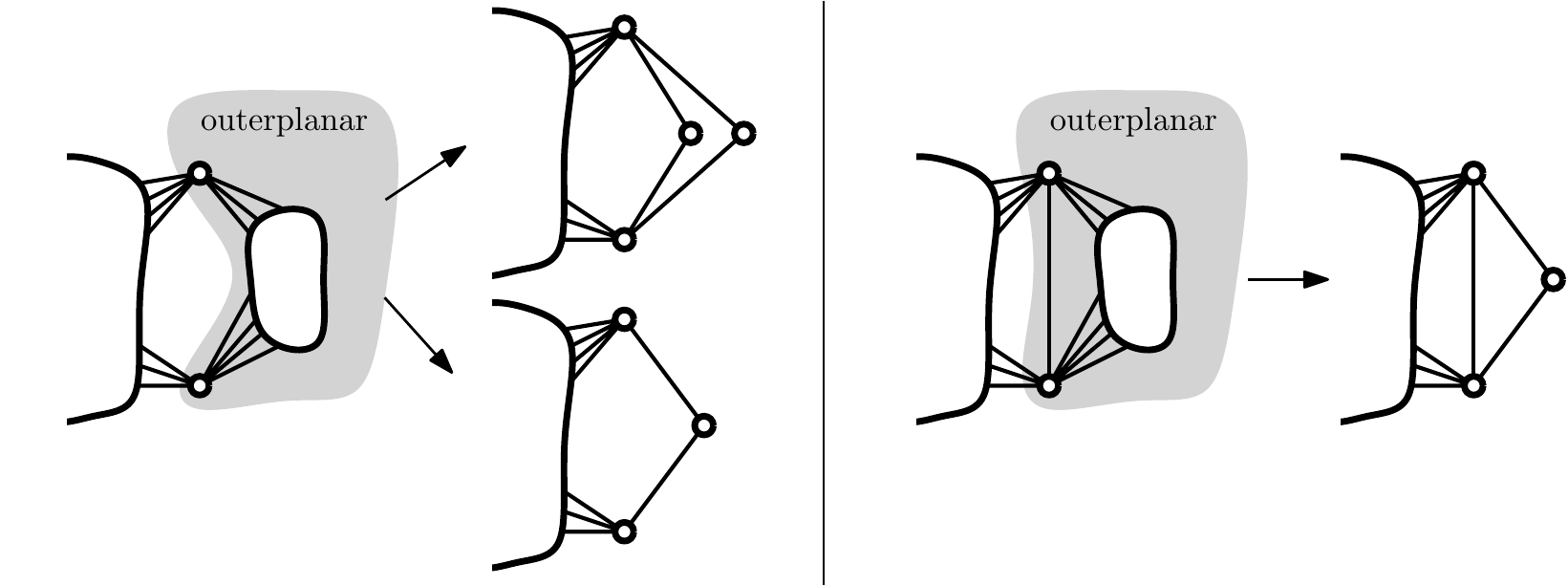}
  \caption{On the left a depiction of \cref{reduction:replace-component}, which reduces a connected subgraph to one or two vertices depending on its internal structure. On the right a depiction of \cref{reduction:contract-bump} which contracts a connected subgraph to a single vertex if it is \op together with the two adjacent vertices that form its neighborhood.}
  \label{fig:reductions-2-nbrs}
\end{figure}

\subsection{Reducing a large biconnected component}

We now give the remaining reduction rules to reduce the size a biconnected component of a protrusion. If a subgraph of a graph is \op and adjacent only to two connected vertices, we can use \cref{criterion:edge-removal} to argue that the entire subgraph can be replaced by any other \op graph that is adjacent to the same two vertices. The following reduction rule exploits this by replacing such a subgraph with a single vertex, see also \cref{fig:reductions-2-nbrs}.

\begin{reduction} \label{reduction:contract-bump}
Suppose that there is an edge $e = uv$ in a graph $G$ such that $G - V(e)$ has a connected component~$C$ such that~$G \brb C$ is \op. Then contract~$C$ into a single vertex.
\end{reduction}

\begin{lemma}[Safeness]
\label{lem:contract-bump}
Let $G$, $uv \in E(G)$, $C$ satisfy the requirements of \cref{reduction:contract-bump} and let $G'$ be obtained from $G$ by contracting $C$ to a new vertex $c$. Then $\opd(G) = \opd(G')$.
\end{lemma}
\begin{proof}
It suffices to prove inequality $\opd(G) \le \opd(G')$.
If $|N_G(C)| = 1$, then we obtain the case already considered in \cref{reduction:articulation-point}, which is safe due to \cref{criterion:articulation-point}.
Assume for the rest of the proof that $N_G(C) = \{u,v\}$.

Let $S' \subseteq V(G')$ be any \op deletion set in $G'$; we will prove~$\opd(G) \leq |S'|$. We first deal with two easy cases.

\bmp{If~$S' \cap \{u,v\} \neq \emptyset$, then we argue that~$S = (S' \setminus \{c\})$ is an outerplanar deletion set in~$G$. This follows from the fact that~$G \brb C$ is \op while~$C$ has at most one neighbor in~$G - S'$, so that \cref{criterion:articulation-point} shows~$G-S$ is outerplanar. Hence~$\opd(G) \leq |S| \leq |S'|$.

If~$c \in S'$ but~$S' \cap \{u,v\} = \emptyset$, then~$S = (S' \setminus \{c\}) \cup \{u\}$ is not larger than~$S'$. To see that~$G - S$ is \op, we apply \cref{criterion:articulation-point} to the articulation point~$v$. Since~$G \brb C$ is outerplanar by assumption, while~$(G - S) - V(C)$ is a subgraph of~$G' - S'$ and therefore outerplanar, \cref{criterion:articulation-point} ensures~$G - S$ is \op. Hence~$\opd(G) \leq |S'|$.}

Suppose now that $S' \cap \{u,v,c\} = \emptyset$.
We check the requirements of \cref{criterion:edge-removal} for the graph $G-S'$ and edge $uv$.
First, each connected component of $G-S'-\{u,v\}$
is \op when considered together with its neighborhood.
It remains to show that $G-S' \setminus uv$ does not have three induced internally vertex-disjoint paths connecting $u$ and $v$.
Since $(u,c,v)$ already gives such a path and $G'-S'$ is \op, the graph $G' - (S' \cup \{c\}) \setminus uv = G-(S' \cup C) \setminus uv$ can have at most one $(u,v)$-path.
By \cref{lem:paths-in-different-components}, there can be also only one such path in $G\brb C \setminus uv$.
Therefore replacing $c$ back with $C$ does not increase the number of $(u,v)$-paths and so \cref{criterion:edge-removal} applies.
We have thus shown that $S'$ is an \op deletion set in $G$, which concludes the proof. 
\end{proof}

The next reduction rule addresses high degree vertices within a biconnected component. It uses the same idea as used in \cref{reduction:irrelevant-edge}, in fact, its safeness follows directly from \cref{reduction:generalized-fan-rule}. See also \cref{fig:reductions-fan-and-ladder}.

\begin{reduction}\label{reduction:fan}
Suppose we are given a graph~$G$, a vertex~$x \in V(G)$, and five vertices~$v_1, \ldots, v_5 \in N_G(x)$ that lie, in order of increasing index, on an induced path~$P$ in~$G-x$ from~$v_1$ to~$v_5$, such that~$N_G(x) \cap V(P) = \{v_1, \ldots, v_5\}$. Let~$C$ be the component of~$G-\{v_1, v_5, x\}$ containing~$P-\{v_1, v_5\}$. If~$G\brb C$ is \op, then remove the edge $xv_3$.
\end{reduction}


The final reduction rule reduces the number of ``internal'' vertices of an \op biconnected graph. These are the edges whose endpoints form a separator in the graph. The previous rule addresses the case where these edges share an endpoint. The final reduction rule focuses on the case where the edges are disjoint: they form a matching.

\begin{definition} \label{def:order-respecting}
 For a graph~$G$, a sequence of edges~$e_1, \ldots, e_\ell \in E(G)$ is an \emph{order-respecting matching} if the set of edges is a matching and if for all~$1 \leq i<j<k \leq \ell$ we have that~$e_i$ and~$e_k$ are in different connected components of~$G - V(e_j)$.
\end{definition}

We can now formulate a property of biconnected graphs containing an order-respecting matching. This allows us to identify a number of cycles in the graph that are useful in the proof of the final reduction rule.

\begin{lemma} \label{lem:paths-in-ladder}
 For an integer~$\ell > 1$, if~$G$ is biconnected and~$e_1, \ldots, e_\ell$ is an order-respecting matching in~$G$ then 
 there exists vertex-disjoint paths $P_x, P_y$ in $G$, such that each of $P_x, P_y$ intersects every set $V(e_i)$, $i \in [\ell]$, and these intersections appear in order of increasing index.
 \end{lemma}
\begin{proof}
 Let~$G'$ be the graph obtained from~$G$ by subdividing~$e_1$ and~$e_\ell$ with new vertices~$a$ and~$b$. Since subdividing edges preserves biconnectivity, the graph~$G'$ is biconnected. So then there are two internally vertex-disjoint $(a,b)$-paths~$P_1$ and~$P_2$. Take~$P_x := P_1 - \{a,b\}$ and~$P_y := P_2 - \{a,b\}$. Let~$x_1,x_\ell,y_1$, and~$y_\ell$ be the unique vertices in (respectively)~$V(P_x) \cap V(e_1), V(P_x) \cap V(e_\ell), V(P_y) \cap V(e_1)$, and~$V(P_y) \cap V(e_\ell)$. Observe that~$P_x$ is an $(x_1, x_\ell)$-path in~$G$ and~$P_y$ is a $(y_1, y_\ell)$-path in~$G$ and both paths are vertex disjoint.
 
 For any~$e_i \in \{e_2, \ldots, e_{\ell-1}\}$ we have by definition of order-respecting matching that~$x_1$ and~$x_\ell$ are in different connected components of~$G-V(e_i)$, hence one of the endpoints of~$e_i$ must lie on~$P_x$. Similarly one of the endpoints of~$e_i$ must lie on~$P_y$. For all~$1 < i < \ell$ let~$x_i$ denote the endpoint of~$e_i$ that lies on~$P_x$ and let~$y_i$ denote the endpoint of~$e_i$ that lies on~$P_y$.
 
 By \cref{def:order-respecting} we have for all~$1 \leq i < j < k \leq \ell$ that~$x_i$ and~$x_k$ are in different connected components of~$G-V(e_j)$. So the subpath of~$P_x$ between~$x_i$ and~$x_k$ must contain a vertex of~$V(e_j)$. Since~$y_j$ lies on~$P_y$ which is disjoint from~$P_x$ we must have that~$x_j$ lies on~$P_x$ between~$x_i$ and~$x_k$. Since~$1 \leq i < j < k \leq \ell$ are arbitrary it follows that~$\{x_1, \ldots, x_\ell\}$ occur in order of increasing index on the path~$P_x$. Similarly,~$\{y_1, \ldots, y_\ell\}$ occur in order of increasing index on the path~$P_y$.
 \end{proof}

We are now ready to formulate the final reduction rule. It applies within a biconnected outerplanar part of the graph that has seven or more disjoint ``internal'' edges. We make use of the definition of order-respecting matching to define an order on the edges. If this biconnected \op part only connects to the remainder of the graph via the endpoints of the first and the last edge, then we show that the edges in the middle can be removed without affecting the outerplanar deletion number of the graph, see also \cref{fig:reductions-fan-and-ladder}.

\begin{figure}[bt]
  \centering
  \includegraphics[width=\linewidth]{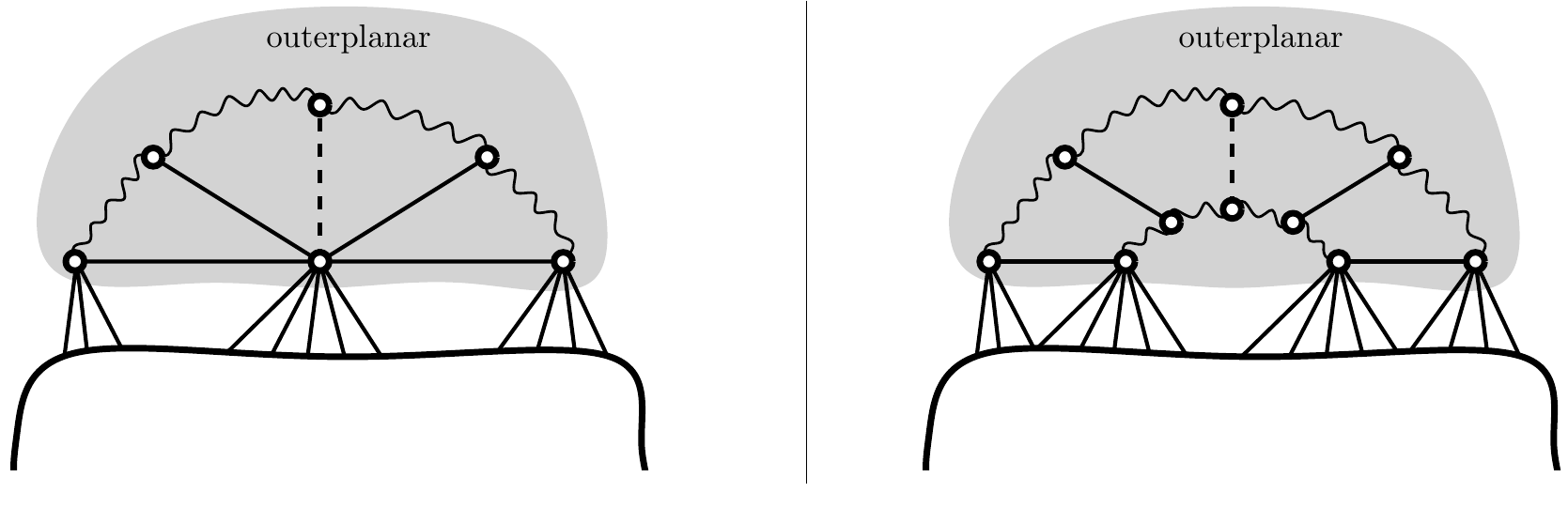}
  \caption{On the left a depiction of \cref{reduction:fan} which is able to remove the middle edge of a fan structure in an \op subgraph that is sufficiently isolated from the rest of the graph. On the right a depiction of \cref{reduction:ladder}, which removes the middle edge of an order-respecting matching in an \op subgraph that is sufficiently isolated from the rest of the graph.}
  \label{fig:reductions-fan-and-ladder}
\end{figure}

\begin{reduction}\label{reduction:ladder}
 Let~$G$ be a graph,~$e_1,\ldots,e_7$ be a matching in~$G$, and let~$C$ be a connected component of~$G - (V(e_1) \cup V(e_7))$. If~$\{e_2,\ldots,e_6\} \subseteq E(C,C)$, $N_G(C) = V(e_1) \cup V(e_7)$, $G\langle C \rangle$ is biconnected and \op, and~$e_1,\ldots,e_7$ is an {order-respecting matching} in~$G\langle C \rangle$, then remove~$e_4$.
\end{reduction}

\begin{lemma}[Safeness]\label{lem:reduction:ladder}
 Let~$e_1, \ldots, e_7$ be a matching in a graph~$G$ and let~$C$ be a connected component of~$G$.
 If \cref{reduction:ladder} applies to~$G$,~$e_1, \ldots, e_7$, and~$C$, then $\opd(G \setminus e_4) = \opd(G)$.
\end{lemma}
\begin{proof}
 Clearly any solution to~$G$ is also a solution to~$G \setminus e_4$ so~$\opd(G \setminus e_4) \leq \opd(G)$. To show~$\opd(G \setminus e_4) \geq \opd(G)$ suppose that~$G \setminus e_4 - S$ is \op. We will show that there is a set~$S' \subseteq V(G)$ of size at most~$|S|$ exists such that~$G-S'$ is \op. We first formulate the following structural property:
 
 \begin{countclaim} \label{ladder:claim:paths}
  If~$i\in \{2,\ldots,6\}$ then for any induced path~$P$ in~$G\setminus e_i$ between the endpoints of~$e_i$ either $P$ is an induced path in~$G\brb C \setminus e_i$, or $P$ has a subpath disjoint from~$C$ connecting an endpoint of~$e_1$ to an endpoint of~$e_7$.
 \end{countclaim}
 \begin{claimproof}
  Let~$i \in \{2,\ldots,6\}$ be arbitrary. Suppose for that~$P$ is an induced path in~$G\setminus e_i$ between the endpoints of~$e_i$ that contains a vertex~$v$ outside~$G\brb C$. We show~$P$ contains a subpath disjoint from~$C$ connecting an endpoint of~$e_1$ to an endpoint of~$e_7$. Consider the subpath~$P'$ of~$P$ from the last vertex~$x$ before~$v$ to the first vertex~$y$ after~$v$, and note that~$P'$ is disjoint from~$C$. By definition~$x,y$ have neighbors outside~$G\brb C$ so~$x,y \in N_G(C) = V(e_1) \cup V(e_7)$. Since~$P$ is induced, so is~$P'$, hence there is no edge between~$x$ and~$y$. It follows that one of~$x,y$ is an endpoint of~$e_1$ and the other an endpoint of~$e_7$, hence~$P'$ is a subpath of~$P$ disjoint from~$C$ that connected an endpoint of~$e_1$ to an endpoint of~$e_7$.
 \end{claimproof}

 We apply \cref{lem:paths-in-ladder} to $G\brb C$ and the matching $(e_1, \dots, e_7)$ to obtain vertex-disjoint
 paths $P_x, P_y$ in $G\brb C$ which intersect each $V(e_i)$, $i \in [7]$, in order of increasing index. 
 For all~$1 \leq p < q \leq 7$ let~$C_{p,q}$ denote the cycle in~$G\langle C \rangle$ as obtained
 by combining the edges $e_p, e_q$ with the subpaths in $P_x, P_y$ from $V(e_p)$ to $V(e_q)$.
 Observe that whenever $p_1 < q_1 < p_2 < q_2$, then $V(C_{p_1,q_1}) \cap V(C_{p_2,q_2}) = \emptyset$.
 For brevity let~$C_i$ denote~$C_{i,i+1}$ for all~$1 \leq i \leq 6$.
 We consider the following cases:
 \begin{itemize}
  \item \label{ladder:case:geq3}
  If~$|N_G[C] \cap S| \geq 3$ then take~$S' := \{x_1,y_1,x_7\} \cup (S \setminus V(C))$ and observe that~$|S'| \leq |S|$. We show~$G-S'$ is \op using \cref{criterion:articulation-point}. We show for all connected components~$C'$ of~$G-S' -y_7$ that~$(G-S')\brb {C'}$ is \op. Since~$S' \cap V(C) = \emptyset$ we have that~$C$ is a connected component of~$G-S'-V(e_1)-V(e_7)$. Since~$(V(e_1) \cup V(e_7)) \setminus S' = \{y_7\}$ this yields that~$C$ is a connected component of~$G-S'-y_7$. Clearly~$(G-S') \brb C$ is \op since it is a subgraph of~$G\brb C$. Any other connected component~$C'$ of~$G-S'-y_7$ clearly is a connected component of~$G-S'-y_7-V(C)$, so then~$(G-S')\brb {C'}$ is a subgraph of~$G-S'-V(C)$. Since both endpoints of~$e_4$ are in~$C$ we have that~$G-S'-V(C) = G\setminus e_4 -S'-V(C)$. This is a subgraph of~$G\setminus e_4-S$ since~$S \subseteq S' \cup V(C)$. Because~$G\setminus e_4 - S$ is \op we can conclude that~$(G-S') \brb {C'}$ is \op. By \cref{criterion:articulation-point} we have that~$G-S'$ is \op.
  
  \item If~$|N_G[C] \cap S| = 2$ and at least one of~$\{C_1,C_6\}$ does not intersect~$S$, we may assume without loss of generality that~$C_6$ does not intersect~$S$. Take~$S' := V(e_1) \cup (S \setminus V(C))$ and observe that~$|S'| \leq |S|$. We show~$G-S'$ is \op using \cref{criterion:edge-removal} on the edge~$e_7$.

  We first show for all connected components~$C'$ of~$(G-S') - V(e_7)$ that~$(G-S')\brb {C'}$ is \op. Since~$S' \cap V(C) = \emptyset$ we have that~$C$ is a connected component of~$G-S'-V(e_1)-V(e_7) = (G-S')-V(e_7)$. Clearly~$(G-S') \brb C$ is \op since it is a subgraph of~$G\brb C$. Any other connected component~$C'$ of~$G-S'-V(e_7)$ is a connected component of~$G-S'-V(e_7)-V(C)$, so then~$(G-S')\brb {C'}$ is a subgraph of~$G-S'-V(C)$. Since both endpoints of~$e_4$ are in~$C$ we have that~$G-S'-V(C)$ is a subgraph of~$G\setminus e_4-S$, which is \op. Hence~$(G-S') \brb {C'}$ is \op.
  
  It remains to show that there are at most two induced internally vertex-disjoint paths in~$(G-S')\setminus e_7$ connecting the endpoints of~$e_7$. Since~$C_6$ does not intersect~$S$ or~$S'$ we have that~$C_6 \setminus e_7$ is a path connecting the endpoints of~$e_7$ in~$(G-S)\setminus \{e_4,e_7\}$ and in~$(G-S')\setminus e_7$. By shortcutting we obtain an induced path~$P_1$ in~$(G-S)\setminus \{e_4,e_7\}$ and in~$(G-S')\setminus e_7$, so that~$P_1$ connects the endpoints of~$e_7$ and its internal vertices are all contained in~$C_6$. Suppose for contradiction that~$(G-S') \setminus e_7$ contains two more induced internally vertex-disjoint paths connecting the endpoints of~$e_7$. At least one of these paths is not present in~$(G-S)\setminus \{e_4,e_7\}$ by \cref{criterion:edge-removal} since~$G\setminus e_4-S$ is \op. Let~$P_2$ be such an induced path in~$(G-S')\setminus e_7$ that is not present in~$(G-S)\setminus \{e_4,e_7\}$. We know that~$P_2$ contains a vertex from~$V(C)$, otherwise~$P_2$ would be present in~$(G-S)\setminus \{e_4,e_7\}$ as~$V(e_4) \subseteq V(C)$ and~$S \setminus S' \subseteq V(C)$. Since~$C$ is a connected component of~$(G-S')-V(e_7)$ (as shown before), we have that all internal vertices of~$P_2$ are in~$C$. Hence~$P_1$ and~$P_2$ are internally vertex-disjoint paths in~$G\brb C$ connecting the endpoints of~$e_7$. However, since~$G\brb C$ is \op, this contradicts outerplanarity of~$G\brb C$ by \cref{lem:paths-in-different-components}. By contradiction we conclude~$(G-S') \setminus e_7$ contains at most two induced internally vertex-disjoint paths connecting the endpoint of~$e_7$.
  This shows also the second condition of \cref{criterion:edge-removal} is satisfied, hence~$G-S'$ is \op.
  
  \item Otherwise we have~$|N_G[C] \cap S| \leq 1$ or~$S$ intersects both~$C_1$ and~$C_6$. We show that~$G-S$ is \op.
  
  \begin{countclaim} \label{ladder:claim:separate}
  Suppose the preconditions of \cref{reduction:ladder} hold, $S \subseteq V(G)$, $G \setminus e_4 - S$ is \op, and either $|N_G[C] \cap S| \leq 1$ or~$|N_G[C] \cap S| = 2$ and $S$ intersects both~$C_1$ and~$C_6$.
   Then any path in~$G \setminus e_4 - S$ from an endpoint of~$e_1$ to an endpoint of~$e_7$ intersects~$V(C)$.
  \end{countclaim}
  \begin{claimproof}
    First we argue that at least one of~$C_{1,3}, C_{3,5}, C_{5,7}$ is disjoint from~$S$. If~$S$ intersects~$C_1$ and $C_6$, then $S$ cannot intersect~$C_{3,5}$ since~$C_1$,~$C_6$, and~$C_{3,5}$ are disjoint. If~$C_1$ or~$C_6$ does not intersect~$S$ then by assumption ~$|N_G[C] \cap S| \leq 1$ so~$S$ cannot intersect both~$C_{1,3}$ and~$C_{5,7}$ as they are disjoint. Hence~$S$ is disjoint from at least one of~$C_{1,3}$, $C_{3,5}$, or~$C_{5,7}$.
    
    For the sake of contradiction, suppose that there is a path $P_z$ in $G \setminus e_4 - (S \cup C)$ which connects $z_1 \in V(e_1)$ to $z_7 \in V(e_7)$.
    First consider the case $V(C_{3,5}) \cap S = \emptyset$.
    Let $P'_x, P'_y$ be the subpaths of respectively $P_x, P_y$ from $x_1 \in V(e_1)$ to $x_3 \in V(e_3)$.
    If $|N_G[C] \cap S| = 2$, then $S$ intersects $C_6$, which is disjoint from $V(P'_x) \cup V(P'_y)$.
    We therefore have $|(V(P'_x) \cup V(P'_y)) \cap S| \le 1$ and by disjointedness of $P_x,P_y$ we infer that one of the paths $P'_x, P'_y$ is disjoint from $S$; assume w.l.o.g. that it is $P'_x$.
    Let $P''_x$ be $P'_x$ if $z_1 \in V(P'_x)$ or $P'_x$ concatenated with $e_1$ otherwise.
    Then $P''_x$ connects $V(e_3)$ to $z_1$, it is internally vertex disjoint from $C_{3,5}$ and, 
    since $z_1 \not\in S$, it is vertex-disjoint from $S$.
    Using a symmetric argument we can construct a path, disjoint from $S$, which connects $V(e_5)$ to $z_7$.
    By concatenating these paths with $P_z$ we obtain that there is a connected component of $(G \setminus e_4 - S) - V(C_{3,5})$ which is adjacent to at least two vertices on the cycle $C_{3,5}$: one from $V(e_3)$ and the other one from $V(e_5)$.
    Since $V(e_4)$ separates $V(e_3)$ from $V(e_5)$ in $G\brb C$, these vertices are non-consecutive on 
    the cycle $C_{3,5}$.
    By \cref{lem:attaching:to:cycle}, this contradicts the assumption that $G \setminus e_4 - S$ is \op.
    
    It remains to consider the case $V(C_{1,3}) \cap S = \emptyset$, as the case $V(C_{5,7}) \cap S = \emptyset$ is symmetric.
    We have $V(C_1) \subseteq V(C_{1,3})$ so $V(C_1) \cap S =\emptyset$ and by the assumption $|N_G[C] \cap S| \leq 1$.
    Let $P'_x, P'_y$ be the subpaths of respectively $P_x, P_y$ from $x_3 \in V(e_3)$ to $x_7 \in V(e_7)$.
    Recall that none of $P_x, P_y$ goes through $e_4$ as its endpoints lie on each of $P_x, P_y$.
    By disjointedness of $P_x,P_y$ we infer that one of the paths $P'_x, P'_y$ is disjoint from $S$; assume w.l.o.g. that it is $P'_x$.
    Similarly as before, we define~$P''_x$ to be $P'_x$ if $z_7 \in V(P'_x)$ or $P'_x$ concatenated with $e_7$ otherwise.
    Then $P''_x$ connects $V(e_3)$ to $z_7$ in $G\brb C \setminus e_4 - S$ and it is internally vertex disjoint with $C_{1,3}$ 
    By concatenating $P''_x$ with $P_z$ we obtain that there is a connected component of $(G \setminus e_4 - S) - V(C_{1,3})$ which is adjacent to two non-consecutive vertices on the cycle $C_{1,3}$ ($z_1 \in V(e_1)$ and the other one from $V(e_3)$).
    By \cref{lem:attaching:to:cycle}, this contradicts the assumption that $G \setminus e_4 - S$ is \op.
  \end{claimproof}
  
  If~$S$ intersects both~$C_1$ and~$C_6$ then take~$e_a = e_3$ and~$e_b = e_5$, otherwise we have~$|N_G[C] \cap S| \leq 1$ so at least one of~$e_a \in \{e_2,e_3\}$ is not hit by~$S$. Similarly there is an edge~$e_b \in \{e_5,e_6\}$ that does not intersect~$S$. Now note that~$\{e_1,e_a,e_4,e_b,e_7\}$ is a order-respecting matching in~$G\brb C$,~$N_G(V(e_a)) \subseteq N_G[C]$, and~$N_G(V(e_b)) \subseteq N_G[C]$.
  
  To show~$G-S$ is \op we use \cref{criterion:edge-removal} on the edge~$e_a$, i.e., we show for any connected component~$H$ of~$G-S-V(e_a)$ that~$(G-S)\brb H$ is \op and there are at most two induced internally vertex-disjoint paths in~$(G-S) \setminus e_a$ connecting both endpoints of~$e_a$.
  
  To prove the latter, observe that any induced path in~$(G-S) \setminus e_a$ is also an induced path in~$G\setminus e_a$ and by \cref{ladder:claim:separate} they cannot contain a subpath from an endpoint of~$e_1$ to an endpoint of~$e_7$. Now by \cref{ladder:claim:paths} we have any such path is an induced path in~$G\brb C \setminus e_a$. Since~$G \brb C \setminus e_a$ is \op we have by~\cref{criterion:edge-removal} that there are at most two such paths that are internally vertex-disjoint.
  
  We now show for any connected component~$H$ of~$G-S-V(e_a)$ that~$(G-S)\brb H$ is \op. Consider the case that~$H$ does not intersect~$C_{a,b}$. Then we know that~$(G-S)\brb H$ does not contain the edge~$e_4$, hence~$(G-S)\brb H$ is a subgraph of~$G\setminus e_4 - S$, which is \op. For the other case, let~$D_1$ be the graph consisting of all connected components intersecting~$C_{a,b}$. It suffices to show that~$(G-S)\brb {D_1}$ is \op.
  
  We use \cref{criterion:edge-removal} on the edge~$e_b$ in the graph~$D'_1 := (G-S)\brb {D_1}$, i.e., we show for any connected component~$H$ of~$D'_1-V(e_b)$ that~$D'_1 \brb H$ is \op and there are at most two induced internally vertex-disjoint paths in~$D'_1 \setminus e_b$ connecting the endpoints of~$e_b$. Let~$D_2$ be the graph consisting of all connected components of~$D'_1 - V(e_b)$ that intersect~$C_{a,b}$. Observe that~$D'_1 \brb {D_2}$ is a subgraph of~$G \brb C$ \hui{(in fact~$V(D_2) = V(C_{a,b})\setminus V(e_b)$)}, hence~$D'_1 \brb {D_2}$ is \op.
  Any other connected component~$H$ of~$D'_1 - V(e_b)$ that does not intersect~$C_{a,b}$ does not contain the edge~$e_4$, hence~$D'_1\brb H$ is a subgraph of~$G\setminus e_4 - S$, which \op.
  
  It remains to show that there are at most two induced internally vertex-disjoint paths in~$D'_1 \setminus e_b$ connecting both endpoints of~$e_b$. Clearly such paths cannot contain an endpoint of~$e_1$. Since such a path is also an induced path in~$G \setminus e_b$ so then by \cref{ladder:claim:paths} we have that such a path is an induced path in~$G\brb C \setminus e_b$. Since~$G\brb C \setminus e_b$ is \op it follows from \cref{criterion:edge-removal} that there are at most two such internally vertex-disjoint paths. Hence~$D'_1 = (G-S)\brb {D_1}$ is \op completing the argument that~$G-S$ is \op.
 \end{itemize}
 
 We have shown in all cases that there exists a set~$S'$ of size at most~$|S|$ such that~$G - S'$ is \op.
\end{proof}

This concludes our final reduction rule. What remains is how they can be applied in polynomial time to reduce the protrusions to a constant size. Since the number of biconnected components can be bounded by a constant (see \cref{lem:replace-component:apply}), we proceed to show how to reduce the size of these biconnected components to a constant.

\subsection{Reducible structures in biconnected outerplanar graphs}

In this section we use the weak dual of a biconnected outerplanar graph to argue that a large biconnected outerplanar graph contains a structure to which a reduction rule is applicable. We therefore need some terminology to relate objects in a biconnected outerplanar graph~$G$ with those in its weak dual~$\widehat{G}$. The following properties are well-known. 

\begin{observation}[{\cite[Corollary 6]{Syslo79}}] \label{obs:biconnected:hamiltonian} 
Any biconnected outerplanar graph on at least three vertices has a unique Hamiltonian cycle.
\end{observation}

\begin{observation}[\cite{FleischnerGH74}] \label{obs:weakdual:tree}
The weak dual of a biconnected outerplanar graph is a tree.
\end{observation}

It is justified to speak of \emph{the} dual of a biconnected outerplanar graph, since all embeddings in which all vertices lie on the outer face have exactly the same set of faces. This can easily be seen by noting that the unique Hamiltonian cycle has a unique outerplanar embedding up to reversing the ordering, and that all remaining edges are chords of this cycle drawn in the interior. For an outerplanar graph~$G$ we can therefore uniquely classify its edges into exterior edges which lie on the outer face of an outerplanar embedding, and interior edges which bound two interior faces and are chords of the Hamiltonian cycle formed by the outer face.

For a biconnected outerplanar graph~$G$ we use~$\widehat{G}$ to denote its weak dual. For an interior edge~$e \in E(G)$ we use~$\widehat{e}$ to denote its dual in~$\widehat{G}$. Note that exterior edges, which lie on the outer face, do not have a dual in~$\widehat{G}$ since we work with the \emph{weak} dual. Each bounded face~$f$ of an outerplanar embedding of~$G$ corresponds to a vertex~$\widehat{f}$ of~$\widehat{G}$. For a vertex~$\widehat{f}$ of~$\widehat{G}$, we use~$V_G(\widehat{f})$ to denote the vertices of~$G$ incident on face~$f$. We extend this notation to vertex sets and subgraphs of~$\widehat{G}$, so that~$V_G(\widehat{G}') = \bigcup _{\widehat{f} \in \widehat{G}'} V_G(\widehat{f})$. Similarly, for an edge~$\widehat{e}$ of~$\widehat{G}$ we use~$V_G(\widehat{e})$ or simply~$V(e)$ to denote the endpoints of the edge in~$G$ for which~$\widehat{e}$ is the dual. For~$\widehat{Y} \subseteq E(\widehat{G})$ we define~$\mic{V_G(\widehat{Y})} := \bigcup _{\widehat{e} \in Y} V_G(\widehat{e})$.

It is insightful to think about the weak dual~$\widehat{G}$ of a biconnected outerplanar graph~$G$ as prescribing a way in which to build graph~$G$ as follows: starting from the disjoint union of the induced cycles forming the interior faces corresponding to~$V(\widehat{G})$, for each edge~$\widehat{f'}\widehat{f''} = \widehat{e} \in E(\widehat{G})$ glue together the induced cycles for~$f'$ and~$f''$ at the edge~$e$ that is common to both cycles. Since each interior face is an induced cycle and therefore biconnected, while gluing biconnected subgraphs along edges preserves biconnectivity, we observe the following.

\begin{observation} \label{obs:weakdual:subtree:biconnected}
If~$G$ is a biconnected outerplanar graph and~$\widehat{R}$ is a nonempty connected subtree of its weak dual~$\widehat{G}$, then~$G[V_G(\widehat{R})]$ is biconnected.
\end{observation}


Suppose we have a connected subtree~$\widehat{R}$ of the weak dual~$\widehat{G}$ of a biconnected outerplanar graph~$G$. In the process of constructing~$G$ from the induced cycles formed by its faces, the subgraph~$G[\widehat{R}]$ only becomes adjacent to vertices outside the subgraph by gluing faces outside~$\widehat{R}$ onto faces of~$\widehat{R}$. Since these are glued along edges whose dual has one endpoint inside and one endpoint outside~$\widehat{R}$, we observe the following.

\begin{observation} \label{obs:weakdual:separate:correct}
Let~$G$ be a biconnected outerplanar graph and let~$\widehat{R}$ be a connected subtree of its weak dual~$\widehat{G}$. Let~$\widehat{Y} \subseteq E(\widehat{G})$ denote those edges which have exactly one endpoint in~$\widehat{R}$ and let~$S = \bigcup _{\widehat{e} \in \widehat{Y}} V_G(\widehat{e})$. \micr{use the introduced notation $S = V_G(\widehat{Y})$?}The only vertices of~$V_G(\widehat{R})$ which have a neighbor outside~$V_G(\widehat{R})$ are those in~$S$.
\end{observation}

From these observations, we deduce the following lemma stating how a subtree of~$\widehat{G}$ that is attached to the rest of~$\widehat{G}$ at a single edge, represents a connected subgraph.

\begin{lemma} \label{obs:weakdual:pendanttree}
Let~$G$ be a biconnected outerplanar graph, let~$\widehat{e} \in E(\widehat{G})$, and let~$\widehat{R}$ be one of the two trees in~$\widehat{G} \setminus \widehat{e}$. For~$C = V_G(\widehat{R}) \setminus V_G(\widehat{e})$ the graph~$G[C]$ is connected and~$N_G(C) = V_G(\widehat{e})$.
\end{lemma}
\begin{proof}
By \cref{obs:weakdual:subtree:biconnected}, the graph~$G[V_G(\widehat{R})]$ is biconnected, so by \cref{obs:biconnected:hamiltonian} it has a Hamiltonian cycle. It is easy to see that~$\widehat{R}$ is the weak dual of~$G[V_G(\widehat{R})]$. The latter graph contains edge~$e$, since it is incident on both faces represented by the endpoints of~$\widehat{e}$, one of which has its dual vertex in~$\widehat{R}$. Since~$\widehat{e} \notin E(\widehat{R})$, the edge~$e$ is an exterior edge of~$G[\widehat{R}]$. So~$e$ is an edge of the unique Hamiltonian cycle of~$G[\widehat{R}]$, which implies that the removal of~$V(e)$ leaves that subgraph connected. Hence~$V_G(\widehat{R}) \setminus V_G(\widehat{e})$ is a connected subgraph of~$G$. By \cref{obs:weakdual:separate:correct} the only vertices of~$V_G(\widehat{R})$ which have neighbors in~$G$ outside~$V_G(\widehat{R})$ are those in~$V_G(\widehat{e})$, which proves that~$G[\widehat{R}] - V_G(\widehat{e})$ is a connected component of~$G - V_G(\widehat{e})$. Its neighborhood in~$G$ is equal to~$V_G(\widehat{e})$ since the predecessor and successor of the vertices of~$V_G(\widehat{e})$ on the Hamiltonian cycle of~$G[\widehat{R}]$ are contained in~$G[\widehat{R}] - V_G(\widehat{e})$.
\end{proof}

The following lemma gives a condition under which removing the endpoints of a matching of two edges preserves connectivity. Recall the definition of order-respecting matching (\cref{def:order-respecting}).

\begin{lemma} \label{lemma:remove:ordered:matching:connected}
If~$G$ is a biconnected outerplanar graph and~$M = \{e_1, e_2, e_3\}$ is an order-respecting matching in~$G$ such that~$e_1$ and~$e_3$ are exterior edges while~$e_2$ is an interior edge, then the subgraph~$G' := G - V(\{e_1, e_3\})$ is connected and~$N_G(G') = V(\{e_1, e_3\})$.
\end{lemma}
\begin{proof}
Consider the unique Hamiltonian cycle~$C$ of~$G$. The exterior edges~$e_1$ and~$e_3$ lie on~$C$ while the interior edge~$e_2$ is a chord of~$C$. Since~$M$ is an order-respecting matching, the vertex sets~$V(e_1)$ and~$V(e_3)$ lie in different connected components of~$G - V(e_2)$. Partition the cycle~$C$ into two vertex-disjoint~$(V(e_1), V(e_2))$-paths~$C_1, C_2$, and let~$C'_1, C'_2$ be the paths formed by the interior vertices of~$C_1$ and~$C_2$. Since~$V(e_2)$ separates~$V(e_1)$ and~$V(e_3)$, the paths~$C'_1,C'_2$ are nonempty and both contain a vertex of~$V(e_2)$. This implies that the graph~$G[C'_1 \cup C'_2]$ is connected, since the edge~$e_2$ connects the two paths. Since~$C'_1 \cup C'_2$ span all vertices of~$G$ except~$V(\{e_1, e_3\})$ we have~$G' = G[V(C'_1) \cup V(C'_2)]$, and since each vertex of~$V(\{e_1, e_3\})$ is adjacent to an endpoint of~$C'_1$ or~$C'_2$, the lemma follows.
\end{proof}

The following shows that the order-respecting property of a matching can be deduced from its path-like structure in the dual.

\begin{lemma} \label{lemma:matching:from:path}
Let~$G$ be a biconnected outerplanar graph. If~$\widehat{P}$ is a path in~$\widehat{G}$ and~$M = \{e_1, \ldots, e_\ell\}$ is a matching in~$G$ such that the dual edges~$\widehat{e}_1, \ldots, \widehat{e}_\ell$ appear on~$\widehat{P}$ in the order as given by the indices, but not necessarily consecutively, then~$M$ is an order-respecting matching in~$G$. 
\end{lemma}
\begin{proof}
For~$1 < i < \ell$ let~$\widehat{G}_i$ denote the tree in~$\widehat{G} \setminus \widehat{e_i}$ that contains~$\widehat{e_1}$. Since the edges~$\widehat{e}_i$ lie on~$\widehat{P}$ in order of increasing index, tree~$\widehat{G}_i$ contains~$\{\widehat{e_j} \mid j < i\}$ but no edge of~$\{\widehat{e_j} \mid j > i\}$. Since the edges in the matching~$M$ are vertex-disjoint, this implies that~$V_G(\widehat{G}_i)$ contains~$V_G(\widehat{e_j})$ for~$j < i$ but contains no vertex of~$V_G(\widehat{e_j})$ for~$j>i$.\bmpr{This may warrant a small justification.} 

By \cref{obs:weakdual:pendanttree}, we have~$N_G(V_G(\widehat{G_i}) \setminus V_G(\widehat{e_i})) = V_G(\widehat{e_i})$ for all~$1 < i < \ell$, which implies that the only vertices of~$V_G(\widehat{G_i})$ which have neighbors outside that set are those in~$V_G(\widehat{e_i})$. Hence the removal of~$V_G(\widehat{e_i})$ breaks all paths from endpoints of~$V_G(\widehat{e_j})$ to endpoints of~$V_G(\widehat{e_r})$ for~$j < i < r$. This establishes that the matching~$M$ is order-respecting.
\end{proof}

The previous observation leads to the following lemma. Intuitively, it gives a property similar to that of a tree decomposition, saying that the vertices of~$\widehat{G}$ representing faces containing a fixed vertex~$t \in V(G)$ form a connected subtree of~$\widehat{G}$.

\begin{lemma} \label{lem:faces:path}
Let~$G$ be a biconnected outerplanar graph and let~$t \in V(G)$. If interior faces~$f$ and~$f'$ both contain~$t$, then~$t \in V(e)$ for each edge~$\widehat{e}$ on the unique~$(\widehat{f}, \widehat{f'})$-path~$\widehat{P}$ in~$\widehat{G}$.
\end{lemma}
\begin{proof}
Let~$\widehat{e}$ be an edge on~$\widehat{P}$ and consider the two trees~$\widehat{G}_1,\widehat{G}_2$ of~$\widehat{G} \setminus \widehat{e}$. Since~$\widehat{e}$ lies on the path between~$\widehat{f}$ and~$\widehat{f}'$, with~$t \in V_G(\widehat{f}) \cap V_G(\widehat{f}')$, we have~$t \in V_G(\widehat{G}_1) \cap V_G(\widehat{G}_2)$. By applying \cref{obs:weakdual:separate:correct} twice, once for~$\widehat{G}_1$ and once for~$\widehat{G}_2$, the graph~$G - V_G(\widehat{e})$ has one connected component on vertex set~$V_G(\widehat{G}_1) \setminus V_G(\widehat{e})$ and one connected component on vertex set~$V_G(\widehat{G}_2) \setminus V_G(\widehat{e})$, which implies~$V_G(\widehat{G}_1) \cap V_G(\widehat{G}_2) \subseteq V_G(\widehat{e}) = V(e)$. Since~$t$ belongs to~$V_G(\widehat{G}_1) \cap V_G(\widehat{G}_2)$, we have~$t \in V(e)$.
\end{proof}

The next lemma shows analyzes how a star of edges incident on a common vertex~$x$ are represented in the weak dual.

\begin{lemma} \label{lem:fan:weakdual}
Let~$G$ be a biconnected outerplanar graph and let~$\widehat{P}$ be a path in~$\widehat{G}$ consisting of consecutive vertices and edges~$\widehat{f}_0, \widehat{e}_1, \widehat{f}_1, \widehat{e}_2, \widehat{f}_2, \ldots, \widehat{e}_\ell, \widehat{f}_{\ell}$ of~$\widehat{G}$, such that~$x \in V(e_1) \cap V(e_\ell)$. Then~$x \in V(e_i)$ for all~$i \in [\ell]$, and letting~$v_i$ denote the endpoint of~$e_i$ other than~$x$ for each~$i$, the vertices~$\{v_i \mid i \in [\ell]\}$ lie in order of increasing index on an induced~$(v_1,v_\ell)$-path~$P$ in~$G[V_G(\{\widehat{f}_1, \ldots, \widehat{f}_{\ell-1}\})]-x$ that contains no other neighbors of~$x$.
\end{lemma}
\begin{proof}
Since~$\widehat{f}_0$ (respectively~$\widehat{f}_{\ell}$) is incident on~$\widehat{e}_1$ ($\widehat{e}_\ell$) and~$x \in V(e_1) \cap V(e_\ell)$, we have~$x \in V_G(\widehat{f}_0) \cap V_G(\widehat{f}_{\ell})$. By \cref{lem:faces:path}, this implies that~$x \in V(e_i)$ for all~$i \in [\ell]$.

Since~$G$ is a simple graph without parallel edges, the fact that~$x \in V(e_i)$ for all~$i \in [\ell]$ implies that the other endpoints of the edges~$e_i$ are all distinct. Let~$\{v_i\} = V(e_i) \setminus \{x\}$ for each~$i$. 
Since~$\widehat{P}$ is a path in~$\widehat{G}$, we can construct the graph~$G' := G[V_G(\{\widehat{f}_1, \ldots, \widehat{f}_{\ell-1}\})]$ from disjoint induced cycles for the faces~$f_1, \ldots, f_{\ell-1}$ by gluing them back-to-back along the edges~$e_1, \ldots, e_\ell$, all of which are incident on~$x$. Note that we do not use the faces~$f_0$ and~$f_\ell$. Since in this construction we glue disjoint cycles together at distinct edges in a path-like sequence, and all edges along which we glue are incident on~$x$, it follows that~$G' - x$ is an induced path. It contains~$v_1, \ldots, v_\ell$ in \mic{this} order and no other neighbors of~$x$.
\end{proof}

The last property of a weak dual we need states how removing two nonadjacent vertices incident on a common interior face separates the graph.

\begin{lemma} \label{lemma:separate:face}
Let~$G$ be a biconnected outerplanar graph and let~$\widehat{G}$ be its weak dual. Let~$\widehat{f} \in V(\widehat{G})$ and let~$u,v \in V(G)$ be nonadjacent vertices incident to~$f$. Let~$E_1, E_2$ be the edge sets of the two $(u,v)$-subpaths along the boundary of~$f$, respectively. Let~$\widehat{G}_i$ for~$i \in [2]$ denote the union of all trees~$\widehat{R}$ of~$\widehat{G} - \widehat{f}$ for which the unique edge connecting~$\widehat{R}$ to~$\widehat{f}$ is the dual of an edge in~$E_i$. Then the connected components of~$G - \{u,v\}$ are~$G[(V(E_i) \cup V_G(\widehat{G}_i)) \setminus \{u,v\}]$ for~$i \in [2]$.
\end{lemma}
\begin{proof}
Consider the process of constructing~$G$ from the disjoint cycles bounding its interior faces as dictated by the weak dual~$\widehat{G}$. Since any interior face is an induced cycle, removing the nonadjacent vertices~$u,v$ from the induced cycle bounding~$f$ separates the cycle into exactly two paths~$P_1, P_2$. Since weak dual witnesses that~$G$ can be constructed by gluing the subgraphs~$G[\widehat{R}]$ for~$\widehat{R}$ a tree of~$\widehat{G} - \widehat{f}$ onto an edge of the cycle bounding~$f$, the interior vertices of the paths~$P_1,P_2$ belong to different connected components of~$G - \{u,v\}$. Since each subgraph~$G[\widehat{R}]$ is glued onto at least one interior vertices of a path~$P_1,P_2$, the graph~$G - \{u,v\}$ has exactly two connected components and their contents are as claimed; note that~$V(E_i) \setminus \{u,v\}$ are exactly the interior vertices of path~$P_i$.
\end{proof}

Using these properties we can now prove that a large biconnected outerplanar graph contains a reducible structure. Each of the three types of structures below can be reduced by one of our reduction rules.

\begin{lemma} \label{lem:combinatorial:reducible}
Let~$G$ be a biconnected outerplanar graph and let~$T \subseteq V(G)$ be a \bmp{nonempty} subset of vertices with~$|T| \leq 4$. If~$|V(G)| > 6288$, then in polynomial time we can identify one of the following \emph{reducible structures} in~$G$:
\begin{itemize}
    \item two (possibly adjacent) vertices~$u,v$ such that there is a component~$C$ of~$G- \{u,v\}$ that does not contain any vertex of~$T$, for which~$|V(C)| > 2$,
    \item a matching~$e_1,\ldots,e_7$ in~$G$, such that there is a single connected component~$C$ of~$G - (V(e_1) \cup V(e_7))$ which contains~$\{e_2, \ldots, e_6\}$ but no vertex from~$T$, such that~$N_G(C) = V(e_1) \cup V(e_7)$, the graph $G\brb C$ is biconnected, and~$e_1,\ldots,e_7$ is an {order-respecting matching} in~$G \brb C$, or
    \item a vertex~$x \in V(G)$ and five vertices~$v_1, \ldots, v_5 \in N_G(x)$ that lie, in order of increasing index, on an induced path~$P$ in~$G-x$ from~$v_1$ to~$v_5$, such that~$N_G(x) \cap V(P) = \{v_1, \ldots, v_5\}$, and such that the connected component of~$G - \{v_1, v_5, x\}$ which contains~$P - \{v_1, v_5\}$ contains no vertex from~$T$.
\end{itemize}
\end{lemma}
\begin{proof}
\mic{We will refer to vertices from $T$ as terminals.}
For each~$t \in T$, fix an interior face~$f(t)$ of~$G$ incident \mic{to}~$t$ and define~$\widehat{f(T)} := \{ \widehat{f(t)} \mid t \in T\}$. Let~$\widehat{G}_T$ be the minimal subtree of the weak dual~$\widehat{G}$ spanning~$\widehat{f(T)}$. We say a vertex of~$\widehat{G}_T$ is \emph{important} if it belongs to~$\widehat{f(T)}$ or has degree unequal to two in the graph~$\widehat{G}_T$. By minimality of~$\widehat{G}_T$ each leaf of~$\widehat{G}_T$ belongs to~$\widehat{f(T)}$, from which it easily follows that the edges of~$\widehat{G}_T$ can be partitioned into at most five paths between important vertices, such that no interior vertex of such a path is important. The number five corresponds to the fact that any tree on at most four leaves without vertices of degree two has at most two internal vertices, so at most six vertices and hence five edges in total. The following claim shows the relevance of the set of important vertices.

\begin{countclaim} \label{claim:noterminals}
Let~$\widehat{R}$ be a connected subtree of~$\widehat{G}$ and let~$S = \{V(e) \mid \widehat{e} \text{ has exactly one endpoint in }\widehat{R}\}$. If~$\widehat{R}$ contains no important vertex of~$\widehat{G}_T$, then~$V_G(\widehat{R}) \setminus S$ contains no vertex of~$T$.
\end{countclaim}
\begin{claimproof}
Suppose there exists~$t \in T$ with~$t \in V_G(\widehat{R})$, so~$t$ lies on a face~$f$ whose dual~$\widehat{f}$ is in~$\widehat{R}$. Since~$\widehat{R}$ contains no important vertex, some tree~$\widehat{R}'$ of~$\widehat{G} - V(\widehat{R})$ contains the chosen face~$\widehat{f(t)}$ representing~$t$. The edge~$\widehat{e}$ connecting~$\widehat{R}$ to~$\widehat{R}'$ lies on the path between~$\widehat{f}$ and~$\widehat{f(t)}$ in~$\widehat{G}$ and has exactly one endpoint in~$\widehat{R}$. By \cref{lem:faces:path}, we have~$t \in V(e)$ and therefore~$t \in S$.
\end{claimproof}

We derive a number of claims showing how to find a reducible structure under certain conditions. After presenting these claims, we show that at least one of them guarantees the existence of a reducible structure if~$G$ is sufficiently large. The first claim shows that if any subtree of~$\widehat{G}$ attaches to~$\widehat{G}_T$ at a single edge and represents more than two vertices other than the attachments, we find a reducible structure of the first type.

\begin{countclaim} \label{claim:attachedtree:large}
Let~$\widehat{R}$ be a connected component of~$\widehat{G} - V(\widehat{G}_T)$, which is a tree, and let~$\widehat{e}$ be the unique edge of~$\widehat{G}$ connecting~$\widehat{R}$ to~$\widehat{G}_T$. If~$|V_G(\widehat{R}) \setminus V(e)| > 2$, then~$G$ contains a reducible structure.
\end{countclaim}
\begin{claimproof}
By \cref{obs:weakdual:pendanttree} we know that~$V_G(\widehat{R}) \setminus V(e)$ is the vertex set of a connected component of~$G - V(e)$, and by \cref{claim:noterminals} this component contains no terminals since~$\widehat{G}_T$ spans all important vertices. As the number of vertices in the connected component is larger than two, this yields a reducible structure of the first type for~$\{u,v\} = V(e)$.
\end{claimproof}

%
%

Next we show that if~$\widehat{G}_T$ contains a long path of non-important vertices, we find a reducible structure of the second or third type.

\begin{countclaim} \label{claim:longpath}
Let~$\widehat{P}$ be a subpath of~$\widehat{G}_T$ such that no vertex of~$\widehat{P}$ is important. If~$|V(\widehat{P}')| > 6 \cdot 4 + 1$, then~$G$ contains a reducible structure.
\end{countclaim}
\begin{claimproof}
Consider such a subpath~$\widehat{P}$ of~$\widehat{G}_T$. Let~$\widehat{e}_0, \ldots, \widehat{e}_{24}$ be the first~$25$ edges on~$\widehat{P}$. Partition these into sets~$\widehat{E}_1, \ldots, \widehat{E}_6$ of four consecutive edges each, and let~$\widehat{E}_7$ be a singleton set with the next edge. Observe that for any two edges~$\widehat{e_i}, \widehat{e_j}$ on~$\widehat{P}$ with~$i < j-1$, the subtree~$\widehat{R}$ of~$\widehat{G} \setminus \{\widehat{e_i}, \widehat{e_j}\}$ containing~$\widehat{e_{i+1}}$ contains no important vertex since the only vertices of~$\widehat{G}_T$ it contains belong to~$\widehat{P}$; here we exploit the fact that all vertices of degree three or more in~$\widehat{G}_T$ are important. Hence by \cref{claim:noterminals} the set~$V_G(\widehat{R}) \setminus V(\{e_i, e_j\})$ contains no vertex of~$T$. We will use this property below to find a reducible structure of the second or third type via a case distinction.

Suppose first that there exists~$i \in [6]$ such that~$V_G(\widehat{e}_{4i}) \cap V_G(\widehat{e}_{4 (i+1)}) \neq \emptyset$, that is, some vertex~$x \in V(G)$ is a common endpoint of~$e_{4i}$ and~$e_{4 (i+1)}$. Let~$\widehat{E}'_i := \widehat{E}_i \cup \{\widehat{e}_{4(i+1)}\}$ and let~$E'_i := \{e \mid \widehat{e} \in \widehat{E}'_i\}$. Let~$\widehat{P}'$ be the four vertices of~$\widehat{G}$ which are incident to two edges of~$\widehat{E}_i$, that is, the four internal vertices of the path formed by the five edges~$\widehat{E}_i$. By applying \cref{lem:fan:weakdual} to the path formed by~$\widehat{E}_i$ we find~$x \in V(e)$ for each~$e \in E'_i$ while the other endpoints~$v_1, \ldots, v_5$ of the edges in~$E'_i$ are all distinct and lie on an induced~$(v_1,v_5)$-path~$P$ in~$G[V_G(\widehat{P}')] - x$ containing no other neighbors of~$x$. Let~$\widehat{R}$ be the tree of~$\widehat{G} \setminus \{\widehat{e}_{4i}, \widehat{e}_{4(i+1)}\}$ containing~$\widehat{e}_{4i+1}$ and note that~$V_G(\{\widehat{e}_{4i}, \widehat{e}_{4(i+1)}\}) = \{v_1, v_5, x\}$ and~$V_G(\widehat{R}) \supseteq V_G(\widehat{P}')$. By \cref{obs:weakdual:separate:correct}, the only vertices of~$V_G(\widehat{R})$ which have neighbors outside~$V_G(\widehat{R})$ are those in~$V_G(\{\widehat{e}_{4i}, \widehat{e}_{4(i+1)}\}) = \{v_1, v_5, x\}$. Consider the connected component~$C$ of~$G - \{v_1, v_5, x\}$ that contains~$P - \{v_1, v_5\}$. By the previous argument,~$V(C) \subseteq V_G(\widehat{R})$ while the construction ensures~$V(C) \cap \{v_1, v_5, x\} = \emptyset$. By the argument in the first paragraph of this claim, $V_G(\widehat{R}) \setminus V(\{e_i, e_j\}) = V_G(\widehat{R}) \setminus \{v_1, v_5, x\}$ contains no vertex of~$T$, implying that no vertex of~$T$ belongs to~$C$. Hence~$G$ contains a reducible structure of the third type.

Now suppose the previous case does not apply; then the set~$M = \{e_{4i} \mid 0 \leq i \leq 6\}$ is a matching in~$G$. Since~$\widehat{M}$ lies on the path~$\widehat{P}$ in~$\widehat{G}$ in the relative order given by the indices, by \cref{lemma:matching:from:path} the set~$M$ is an order-respecting matching in~$G$, which implies it is order-respecting in all subgraphs containing this matching. Let~$\widehat{R}$ be the tree of~$\widehat{G} \setminus \{\widehat{e}_{0}, \widehat{e}_{24}\}$ containing the rest of~$\widehat{M}$. By \cref{obs:weakdual:subtree:biconnected}, the graph~$G[V_G(\widehat{R})]$ is biconnected. Since for each~$0 \leq i \leq 6$ the tree~$\widehat{R}$ contains a vertex incident on~$\widehat{e}_{4i}$, which corresponds to a face in~$G$ incident on~$e_{4i}$, it follows that~$G[V_G(\widehat{R})]$ contains all edges of~$M$. Since~$\widehat{R}$ is the weak dual of~$G[V_G(\widehat{R})]$, while~$\widehat{R}$ does not contain the edges~$\widehat{e}_0$ and~$\widehat{e}_{24}$, the edges~$e_0, e_{24}$ are exterior edges of~$G[V_G(\widehat{R})]$; since~$\widehat{e}_4 \in E(\widehat{R})$ the edge~$e_4$ is an interior edge of~$G[V_G(\widehat{R})]$. By applying \cref{lemma:remove:ordered:matching:connected} to the order-respecting matching~$\{e_0, e_4, e_{24}\}$ in the biconnected graph~$G[V_G(\widehat{R})]$, we find that~$C = G[V_G(\widehat{R})] - V(\{e_0, e_{24}\})$ is connected and contains a vertex adjacent to each member of~$V(\{e_0, e_{24}\})$, so that~$G \brb C = G[V_G(\widehat{R})]$ is biconnected and contains the order-respecting matching~$M$. By the argument in the first paragraph, the set~$V_G(\widehat{R}) \setminus V(\{e_0, e_{24}\})$ contains no vertex of~$T$, so that~$C$ contains no vertex of~$T$. Since \cref{obs:weakdual:separate:correct} shows that no vertex of~$C$ has a neighbor outside~$V(\{e_0, e_{24}\})$ it follows that~$C$ is a connected component of~$G - V(\{e_0, e_{24}\})$. Hence we find a reducible structure of the second type.
\end{claimproof}

As the last ingredient, we show that if any interior face of~$G$ contains more than~$16$ vertices, we find a reducible structure.

\begin{countclaim} \label{claim:largeface}
If~$G$ contains an interior face~$f$ that is incident to more than~$16$ vertices, then~$G$ contains a reducible structure of the first type.
\end{countclaim}
\begin{claimproof}
Let~$f$ be a bounded face in~$G$ and consider the cycle bounding~$f$ on the edge set~$E(f)$. We call an edge~$e$ that lies on~$f$ a portal if the tree in~$\widehat{G} \setminus \{\widehat{e}\}$ that does not contain~$\widehat{f}$ contains a vertex of~$\widehat{f(T)}$. Since~$|f(T)| \leq 4$ at most four edges on~$f$ are portals. By \cref{lem:faces:path}, if a terminal~$t$ lies on~$f$ but~$f(t) \neq \widehat{f}$, then the edge~$e$ for which~$\widehat{e}$ lies on the path from~$\widehat{f}$ to~$\widehat{f(t)}$ is a portal.

Let~$T' := (T \cap V(f)) \cup \{ V(e) \mid e \in E(f) \text{ is a portal} \}$. Since each terminal~$t$ for which~$\widehat{f(t)} \neq \widehat{f}$ contributes two adjacent vertices to~$T'$, while each terminal~$t$ with~$\widehat{f(t)} = \widehat{f}$ contributes one vertex, it follows that~$T'$ can be partitioned into at most four sets of two vertices each, such that the vertices in each subset are consecutive along the face. Since~$f$ is incident to more than~16 vertices, this implies that there is a subpath~$P$ of the boundary of~$f$ consisting of five vertices, whose three interior vertices do not belong to~$T'$. Let~$\{u,v\}$ be the endpoints of~$P$. Let~$\widehat{G}_P$ denote the union of all trees~$\widehat{R}$ of~$\widehat{G} - \widehat{f}$ for which the unique edge~$\widehat{e}$ connecting~$\widehat{R}$ to~$\widehat{f}$ satisfies~$e \in E(P)$. By \cref{lemma:separate:face} there is a connected component~$C$ of~$G - \{u,v\}$ on vertex set~$(V(P) \cup V_G(\widehat{G}_P)) \setminus \{u,v\}$. This implies that all three interior vertices of~$P$ belong to~$C$. 

We conclude the proof by showing that~$C$ contains no vertex of~$T$. To see this, note first that~$V(P) \setminus \{u,v\}$ contains no vertex of~$T$ by choice of~$P$. Since all endpoints of portals belong to~$T'$, while no interior vertex of~$P$ belongs to~$T'$, it follows that no edge of~$P$ is a portal. Consequently, for each edge~$e \in E(P)$ the tree of~$\widehat{G} \setminus \widehat{e}$ that does not contain~$\widehat{f}$ does not contain any vertex of~$\widehat{f(T)}$. We exploit this in the following argument. Suppose there exists~$t \in T$ with~$t \in V_G(\widehat{G}_P)$, and let~$e$ be the edge of~$P$ such that~$\widehat{G} \setminus \widehat{e}$ contains a tree~$\widehat{R}$ with~$t \in V_G(\widehat{R})$. Since~$e$ is not a portal,~$\widehat{f(t)}$ does not belong to~$\widehat{R}$. Hence~$\widehat{f(t)}$ is either equal to~$\widehat{f}$, or belongs to some tree~$\widehat{R}'$ of~$\widehat{G} - \widehat{f}$ for which the edge~$\widehat{e}'$ connecting~$\widehat{R'}$ to~$\widehat{f}$ satisfies~$e' \notin E(P)$. We consider these cases separately.

\begin{itemize}
	\item If~$\widehat{f(t)} = \widehat{f}$, then~$t \in V(f)$ and therefore~$t \in T'$. By construction, no interior vertex of~$P$ belongs to~$T'$. Since all vertices that belong both to~$V(f)$ and to~$V_G(\widehat{G}_P)$ belong to~$P$, it follows that~$t$ is an endpoint of~$P$, so~$t \in \{u,v\}$. This shows that~$t$ does not belong to the component~$C$ on vertex set~$(V(P) \cup V_G(\widehat{G}_P)) \setminus \{u,v\}$, as required.
	\item If~$\widehat{f(t)} \neq \widehat{f}$, then~$\widehat{f(t)}$ lies in some tree~$\widehat{R}'$ of~$\widehat{G} - \widehat{f}$ for which the edge~$\widehat{e}'$ connecting~$\widehat{R}'$ to~$\widehat{f}$ satisfies~$e' \in E(f) \setminus E(P)$, where~$E(f)$ are the edges of~$G$ bounding face~$f$. Since both~$\widehat{e}$ and ~$\widehat{e}'$ lie on the path in~$\widehat{G}$ between~$\widehat{f(t)}$ and a vertex of~$\widehat{R}$ representing a face containing~$t$, by \cref{lem:faces:path} we have~$t \in V(e') \cap V(e)$. So~$t$ is simultaneously an endpoint of the edge~$e$ that lies on~$P$ and the edge~$e'$ that lies on face~$f$ but not in~$P$. Consequently,~$t$ is an endpoint of~$P$ and therefore~$t \in \{u,v\}$, showing that~$t$ does not belong to the component~$C$ on vertex set~$(V(P) \cup V_G(\widehat{G}_P)) \setminus \{u,v\}$.
\end{itemize}

The preceding argument established that there is a component of~$G - \{u,v\}$ containing at least three vertices and no terminals, which forms a reducible structure of the first kind and completes the proof.
\end{claimproof}

\bmpr{If we wanted to, we could give somewhat better bounds at the cost of a more involved analysis. The argument here says: the tree~$\mathcal{G}_T$ is small, each vertex in there has bounded degree, and each of the remaining parts of the graph attaches onto an edge and is therefore small. The size of the attaching parts is bounded by the number of edges onto which we can attach, which is bounded by the length of a single face. But we can be slightly smarter: for faces which lie on a degree-2 path of~$\mathcal{G}_T$ we can actually say that that the total volume of boring subtrees attaching there is small, because there are two separators of 2 vertices each such that each boring attached tree is in one of the two separated terminal-free subgraphs. So a more refined bound would say that for the six important vertices of~$\mathcal{G}_T$ there can be attachments onto 16 edges of each face, but for the remaining degree-2 faces there are at most 4 boring vertices attached.}

We can now complete the proof of \cref{lem:combinatorial:reducible} by combining the claims above. Recall from the beginning of the proof that~$\widehat{G}_T$ is a tree containing at most six important vertices, such that each vertex of~$\widehat{f(T)}$ and each vertex whose degree in~$\widehat{G}_T$ is unequal to two is important. Removing \mic{the important vertices} from~$\widehat{G}_T$ splits it into at most five paths~$\widehat{P}_1, \ldots, \widehat{P}_\ell$ of non-important vertices. If any of these paths has more than 25 vertices, we find a reducible structure via \cref{claim:longpath}. Assume that this is not the case; then~$\widehat{G}_T$ consists of at most~$5 \cdot 25$ non-important and~$6$ important vertices. If there exists~$\widehat{f} \in \widehat{G}_T$ such that face~$f$ contains more than~$16$ vertices, we find a reducible structure via \cref{claim:largeface}. If not, then the faces represented by~$\widehat{G}_T$ span at most~$(5\cdot25 + 6) \cdot 16$ edges. All remaining vertices of~$G$ lie on faces whose \mic{duals are} not contained in~$\widehat{G}_T$, and therefore \mic{each such vertex lies} on a face~$f$ for which~$\widehat{f}$ belongs to some tree~$\widehat{R}$ of~$\widehat{G} - V(\widehat{G}_T)$. \mic{When~$\widehat{e}$ is the edge connecting~$\widehat{R}$ to~$\widehat{G}_T$ then~$e$ lies on a face represented by~$\widehat{G}_T$.} If for any such tree~$\widehat{R}$ the number of vertices~$|\widehat{R} \setminus V(e)|$ which are not already accounted for is larger than two, then \cref{claim:attachedtree:large} yields a reducible structure. If not, then since the number of attachment edges is bounded by~$(5 \cdot 25+6) \cdot 16$, while each attached tree contributes at most two additional vertices, the total number of vertices in~$G$ is bounded by~$(5 \cdot 25+6) \cdot 16 + 2 \cdot (5 \cdot 25+6) \cdot 16 = 3 (5 \cdot 25 + 6) \cdot 16 = 6288$. Hence any biconnected outerplanar graph with more than this number of vertices contains a reducible structure. The proof above easily turns into a polynomial-time algorithm to find such a structure.
\end{proof}

\section{Wrapping up} \label{sec:wrapup}

Finally, we combine the decomposition from \cref{lem:final-decomposition} with the rules that reduce protrusions.
If $A$ is sufficiently large, $G\brb A$ is \op, and $|N_G(A)| \le 4$, we explicitly detect a reducible structure within $G\brb A$.
First, we use \cref{reduction:replace-component} to reduce the number of biconnected components in $G \brb A$.
Next, we apply \cref{lem:combinatorial:reducible} to a biconnected component $B$ of $G\brb A$ and the set $T = \partial_G(B)$.
It provides us with one of several reducible structures which match our reduction rules.
It is crucial that when $C$ is a subgraph of $B$ and $C \cap \partial_G(B) = \emptyset$, then the neighborhood of $C$ in both $B$ and $G$ is the same.

We remark that the following lemma can be turned into an iterative procedure that maintains the decomposition from \cref{lem:final-decomposition} without the need to recompute the set $L$.
However, we state it in the simplest form as our analysis does not keep track of the polynomial in the running time.

\begin{lemma}\label{lem:protrusion-replacement}
Consider a graph $G$ and a vertex set $A \subseteq {V(G)}$, such that $|A| > 25 \cdot \bcc$, $|N_G(A)| \le 4$, $G[A]$ is connected, and $G\brb A$ is \op.
There is a polynomial-time algorithm that, given $G$ and $A$ satisfying the conditions above, 
outputs a proper minor $G'$ of $G$, so that $\opd(G') = \opd(G)$.
\end{lemma}
\begin{proof}
Within this proof, we say that replacing graph $G$ with $G'$ is safe when $\opd(G') = \opd(G)$.
With $G$ and $A$ as specified, we can execute the algorithm from \cref{lem:replace-component:apply}.
Suppose that it
outputs a vertex set $C \subseteq S$ to
which either \cref{reduction:articulation-point} or \cref{reduction:replace-component} apply and shrinks the graph.
They are safe due to \cref{criterion:articulation-point} and \cref{lem:replace-component:safeness}.
In this case we can terminate the algorithm.

Otherwise \cref{lem:replace-component:apply} provides us with
 a block-cut tree of $G\brb A$ with at most 25 biconnected components, where each such biconnected component $B$ satisfies $|\partial_G(B)| \le 4$.
By a counting argument we can choose one biconnected component $B$ with more than $\bcc$ vertices.
We execute the algorithm from \cref{lem:combinatorial:reducible} for the graph $B$ and the set $T = \partial_G(B)$.
Note that $B$ is a subgraph of $G \brb A$, which is \op by the assumption.
Depending on the type of returned structure, we select an appropriate reduction rule.
\begin{itemize}
    \item Case 1. We find two (possibly adjacent) vertices~$u,v$ such that there is a component~$C$ of~$B - \{u,v\}$ that does not contain any vertex of~$\partial_G(B)$, for which~$|V(C)| > 2$.
    Clearly $N_G(C) \subseteq \{u,v\}$ and $G \brb C$ is \op as a subgraph of $B$.
    If $uv \in E(G)$, then we apply \cref{reduction:contract-bump} to contract $C$ into a single vertex.
    This reduction is safe due to \cref{lem:contract-bump}.
    
    If $uv \not\in E(G)$, we apply \cref{reduction:replace-component}.
    The criterion in the statement of the rule can be easily checked in polynomial time.
    By \cref{lem:replace-component:minor}, the replacement operation is equivalent to a series of edge contractions.
    The safeness follows from \cref{lem:replace-component:safeness}.
    Since $|V(C)| > 2$, we always perform at least one contraction.
    
    \item Case 2. We obtain a matching~$e_1,\ldots,e_7$ in~$B$, such that there is a single connected component~$C$ of~$B - (V(e_1) \cup V(e_7))$ which contains~$\{e_2, \ldots, e_6\}$ but no vertex from~$\partial_G(B)$, such that~$N_B(C) = V(e_1) \cup V(e_7)$, the graph $B\brb C$ is biconnected, and~$e_1,\ldots,e_7$ is an order-respecting matching in~$B \brb C$.
    Since $C \cap \partial_G(B) = \emptyset$, we get that $N_G(C) = N_B(C) = V(e_1) \cup V(e_7)$.
    This allows us to apply \cref{reduction:ladder} and remove the edge $e_4$.
    This is safe thanks to \cref{lem:reduction:ladder}.
    
    \item Case 3. We find a vertex~$x \in V(B)$ and five vertices~$v_1, \ldots, v_5 \in N_B(x)$ that lie, in order of increasing index, on an induced path~$P$ in~$B-x$ from~$v_1$ to~$v_5$, such that~$N_B(x) \cap V(P) = \{v_1, \ldots, v_5\}$, and such that the connected component $C$ of~$B - \{v_1, v_5, x\}$ which contains~$P - \{v_1, v_5\}$ contains no vertex from~$\partial_G(B)$.
    This means that $C$ is also the connected component of $G - \{v_1, v_5, x\}$ which contains~$P - \{v_1, v_5\}$.
    Furthermore, the path $P$ is also induced in the graph $G-x$ because $B$ is an induces subgraph of $G$.
    The graph $G\brb C$ is \op as a subgraph of $B$ and thus \cref{reduction:fan} applies so we can remove the edge $xv_3$.
    This operation is safe due to \cref{reduction:generalized-fan-rule}.
\end{itemize}

In each case we are able to perform a contraction or removal operation while preserving the \op deletion number of the graph.
The claim follows.
\end{proof}

It is important that the graph is guaranteed to shrink at each step, so after polynomially many invocations of \cref{lem:protrusion-replacement}
we must arrive at an irreducible instance.
We are now ready to prove the main theorem with the final bound on the size of compressed graph $2 \cdot (25 \cdot \bcc + 5) \cdot f_5(c,f_3(c))\cdot(k_2+3)^4$
(see \cref{lem:modulator:summary,lem:final-decomposition}), where $c=40$.
Recall that instances $(G,k)$ and $(G',k')$ are equivalent if $\opd(G) \le k \Leftrightarrow \opd(G') \le k'$.

\begin{thmstar}[\ref{thm:kernel}, restated]
\thmKernelCommand
\end{thmstar}
\begin{proof}[Proof of Theorem~\ref{thm:kernel}]
We use \cref{lem:modulator:summary} to either conclude that $\opd(G) > k$ or to find an equivalent instance $(G_1,k_1)$, where $G_1$ is a subgraph of $G$ and $k_1 \le k$.
In the first case, either $K_4$ or $K_{2,3}$ must be a minor of $G$.
We check it in polynomial time and depending on the result
we output an instance $(H, 0)$, where $H = K_4$ or $H = K_{2,3}$, which is then equivalent to $(G,k)$.
Otherwise \mic{we are guaranteed that $\opd(G) \le k$ implies $\opd(G_1) = \opd(G) - (k - k_1)$} and we also obtain a
$(k_1,c,d)$-\op decomposition of $G_1$, where $d = f_3(c)$, which we supply to the algorithm from \cref{lem:final-decomposition}.
In the first scenario, it applies \cref{reduction:irrelevant-edge} or \cref{reduction:articulation-point} to the instance $(G_1,k_1)$.
These rules may remove vertices or edges, but do not change the value of the parameter and transform the instance into an equivalent one due to \cref{lem:irrelevant-edge-safeness} and \cref{criterion:articulation-point}.
\mic{Moreover, when $\opd(G_1) \le k_1$, then the \op deletion number also stays intact.}
If this is the case, we shrink the graph and rerun the algorithm from scratch on the smaller graph, starting from recomputing the \op decomposition.

Since the graph shrinks at each step, at some point the routine from \cref{lem:final-decomposition} terminates with the set $L$ as the outcome.
Let $(G_2, k_2)$ be the instance equivalent to $(G,k)$ obtained so far.
Then for the returned set $L \subseteq V(G_2)$ we have that $f_5(c,d)\cdot(k_2+3)^4$ upper bounds each of: $|L|$, $|E_G(L,L)|$, and the number of connected components in $G_2 - L$.
Furthermore, 
for each connected component $A$ of $G_2-L$ it holds that $G_2 \brb A$ is \op and $|N_{G_2}(A)| \le 4$.
If any of these components has more than $25 \cdot \bcc$ vertices, then \cref{lem:protrusion-replacement} applies and produces an equivalent instance $(G_3, k_2)$, where $G_3$ is a proper minor of $G_2$.
\mic{This reduction does not affect the parameter nor the \op deletion number.}
We can thus again rerun the algorithm from scratch on the smaller graph.

Otherwise, each connected component of $G_2-L$ has bounded size and $(G_2,k_2)$ is the outcome of the kernelization algorithm.
We check that $G_2$ has at most $(25 \cdot \bcc + 1) \cdot f_5(c,d)\cdot(k_2+3)^4$ vertices.
The number of edges in $G_2$ can be upper bounded by $|E_G(L,L)|$ plus the sum of edges in each $G\brb A$, where $A$ is connected component of $G_2-L$.
Note that this also takes into account the edges with just one endpoint in $L$.
We estimate $|V(G \brb A)| \le 25 \cdot \bcc + 4$
and by \cref{lem:prelim:op-edges} we have that $|E(G \brb A)| \le 2 \cdot |V(G \brb A)|$.
Therefore, $|E(G_2)|$ is at most $2 \cdot (25 \cdot \bcc + 5) \cdot f_5(c,d)\cdot(k_2+3)^4$.
\end{proof}

As a consequence of the theorem above, we obtain the first concrete bounds on the sizes of minor-minimal obstructions to having an outerplanar vertex deletion set of size~$k$.

\begin{corstar}[\ref{cor:obstruction:bounds}, restated]
\corollaryBoundsCommand
\end{corstar}
\begin{proof}[Proof of Corollary \ref{cor:obstruction:bounds}]
Let~$p \colon \mathbb{N} \to \mathbb{N}$ be a function such that for each instance~$(G,k)$ of \opvdfull{} there is an equivalent instance~$(G',k')$ where~$G'$ is a minor of~$G$ on at most~$p(k)$ vertices and at most~$p(k)$ edges. Theorem~\ref{thm:kernel} provides such a function with~$p(k) \in \Oh(k^4)$. In the remainder of the proof, we refer to a vertex set~$S$ whose removal makes a graph outerplanar as a \emph{solution}, regardless of its size. 

Let~$(G,k)$ be a pair satisfying the preconditions to the corollary. Note that~$\opd(G) \leq k+1$ since the graph~$G'$ obtained after removing an arbitrary vertex~$v$ satisfies~$\opd(G') \leq k$ and therefore has a solution~$S$ of size at most~$k$, so that~$S \cup \{v\}$ is a solution in~$G$ of size at most~$k+1$.

Next, we argue that for each vertex~$v \in V(G)$ there is a solution of size~$k+1$ in~$G$ without~$v$. If~$v$ is an isolated vertex in~$G$ then this is trivial. Otherwise, let~$uv \in E(G)$ be an arbitrary edge incident on~$v$. Since~$G' := G \setminus uv$ is a proper minor of~$G$, it has a solution~$S'$ of size~$k$ by assumption. We have~$v \notin S'$: if~$v \in S'$, then~$G' - S = G-S$ which would imply that~$G$ has a solution of size~$k$, contradicting~$\opd(G) > k$. Now observe that~$S' \cup \{u\}$ is a solution in~$G$, because the graph~$G - (S' \cup \{u\})$ is a minor of~$G' - S$, as removing vertex~$u$ also removes the edge~$uv$. Hence~$S' \cup \{u\}$ is a solution of size~$k+1$ in~$G$ that does not contain~$v$.

\mic{Consider the effect of applying the kernelization algorithm of Theorem~\ref{thm:kernel} to the instance~$(G,k+1)$, resulting in an instance~$(G',k')$.
Since~$\opd(G) = k+1$, we have that $\opd(G') = \opd(G) - (k + 1 - k')$.
We also know that $G'$ is a minor of $G$ and the number of vertices and edges in $G'$ is at most $p(k+1)$.
We are going to show that $G' = G$.
Suppose otherwise and consider the series of graph modifying reductions applied by the kernelization algorithm, resulting in successive instances $(G,k+1) = (G_1,k_1), (G_2, k_2), \dots, (G_\ell, k_\ell) = (G', k')$.
We consider two cases: $k_2 < k_1$ and $k_2 = k_1$.
In the first case, the only reduction rule that may decrease the value of the parameter is the one from \cref{lem:modulator:compute-augmented}.
However, as $k_1 = k+1$ and for each vertex~$v \in V(G)$ there is a solution of size~$k+1$ in~$G$ without~$v$, this reduction must produce a~new instance with $k_2 = k_1$; a contradiction.
In the second case, the graph $G_2$ is a proper minor of $G_1 = G$ and $\opd(G_2) = \opd(G) = k+1$ because no reduction can change the \op deletion number unless
it is greater than $k_1$ or $k_2 < k_1$.
But each proper minor of~$G$ has a solution of size at most~$k$ by assumption, which again leads to a contradiction.
This implies that $G' = G$ and hence the number of vertices and edges in $G$ is at most $p(k+1) \in \Oh(k^4)$.
}
%
\end{proof}

\section{Conclusion}

We presented a number of elementary reduction rules for \opvdfull that can be applied in polynomial time to obtain a kernel of~$\Oh(k^4)$ vertices and edges. This kernel does not use protrusion replacement and the constants hidden by the $\Oh$-notation can be derived easily. This is the first concrete kernel for \opvdfull, and a step towards more concrete kernelization bounds for {\sc Planar-\F Deletion}. We hope it inspires new kernelization bounds for {\sc Planar Deletion}.

In earlier work Dell and Van Melkebeek~\cite[Theorem 3]{DellM2014} have shown that there is no kernel for \opvdfull of bitsize~$\Oh(k^{2-\varepsilon})$ unless \containment. This naturally leads to the question, can these two bounds be brought closer together?

Another interesting direction for further research is to obtain concrete kernelization bounds for other {\sc Planar-\F Deletion} problems. Our work exploits the fact that $K_{2,3}$-minor-free graphs cannot have many disjoint paths between two vertices. Previous work~\cite{FominLMPS16} used a similar observation to derive a kernel for {\sc $\theta_c$-Minor-Free Deletion}. An interesting next case would be a {\sc Planar-\F Deletion} problem where~\F does not contain~$K_{2,c}$ or~$\theta_c$ for some~$c$, for example {\sc 2-Transversal} which asks whether a graph of treewidth at most~2 can be obtained by deleting~$k$ vertices.

\bibstyle{plainurl}
\bibliography{bib}

\begin{thebibliography}{10}

\bibitem{Biedl11}
Therese~C. Biedl.
\newblock Small drawings of outerplanar graphs, series-parallel graphs, and
  other planar graphs.
\newblock {\em Discret. Comput. Geom.}, 45(1):141--160, 2011.
\newblock \href {http://dx.doi.org/10.1007/s00454-010-9310-z}
  {\path{doi:10.1007/s00454-010-9310-z}}.

\bibitem{Bodlaender96}
Hans~L. Bodlaender.
\newblock A linear-time algorithm for finding tree-decompositions of small
  treewidth.
\newblock {\em {SIAM} J. Comput.}, 25(6):1305--1317, 1996.
\newblock \href {http://dx.doi.org/10.1137/S0097539793251219}
  {\path{doi:10.1137/S0097539793251219}}.

\bibitem{Bodlaender98}
Hans~L. Bodlaender.
\newblock A partial \emph{k}-arboretum of graphs with bounded treewidth.
\newblock {\em Theor. Comput. Sci.}, 209(1-2):1--45, 1998.
\newblock \href {http://dx.doi.org/10.1016/S0304-3975(97)00228-4}
  {\path{doi:10.1016/S0304-3975(97)00228-4}}.

\bibitem{BodlaenderFLPST16}
Hans~L. Bodlaender, Fedor~V. Fomin, Daniel Lokshtanov, Eelko Penninkx, Saket
  Saurabh, and Dimitrios~M. Thilikos.
\newblock ({M}eta) {K}ernelization.
\newblock {\em J. {ACM}}, 63(5):44:1--44:69, 2016.
\newblock \href {http://dx.doi.org/10.1145/2973749}
  {\path{doi:10.1145/2973749}}.

\bibitem{CattellDDFL00}
Kevin Cattell, Michael~J. Dinneen, Rodney~G. Downey, Michael~R. Fellows, and
  Michael~A. Langston.
\newblock On computing graph minor obstruction sets.
\newblock {\em Theor. Comput. Sci.}, 233(1-2):107--127, 2000.
\newblock \href {http://dx.doi.org/10.1016/S0304-3975(97)00300-9}
  {\path{doi:10.1016/S0304-3975(97)00300-9}}.

\bibitem{ChartrandH67}
Gary Chartrand and Frank Harary.
\newblock Planar permutation graphs.
\newblock {\em Annales de l'I.H.P. Probabilit\'es et statistiques},
  3(4):433--438, 1967.

\bibitem{ChenKJ01}
Jianer Chen, Iyad~A. Kanj, and Weijia Jia.
\newblock Vertex cover: Further observations and further improvements.
\newblock {\em Journal of Algorithms}, 41(2):280--301, 2001.
\newblock \href {http://dx.doi.org/https://doi.org/10.1006/jagm.2001.1186}
  {\path{doi:https://doi.org/10.1006/jagm.2001.1186}}.

\bibitem{CoudertHS07}
David Coudert, Florian Huc, and Jean{-}S{\'{e}}bastien Sereni.
\newblock Pathwidth of outerplanar graphs.
\newblock {\em J. Graph Theory}, 55(1):27--41, 2007.
\newblock \href {http://dx.doi.org/10.1002/jgt.20218}
  {\path{doi:10.1002/jgt.20218}}.

\bibitem{cygan2015parameterized}
Marek Cygan, Fedor~V Fomin, {\L}ukasz Kowalik, Daniel Lokshtanov, D{\'a}niel
  Marx, Marcin Pilipczuk, Micha{\l} Pilipczuk, and Saket Saurabh.
\newblock {\em Parameterized algorithms}.
\newblock Springer, 2015.

\bibitem{CyganPPW10}
Marek Cygan, Marcin Pilipczuk, Micha\l{} Pilipczuk, and Jakub~Onufry
  Wojtaszczyk.
\newblock An improved {FPT} algorithm and a quadratic kernel for pathwidth one
  vertex deletion.
\newblock {\em Algorithmica}, 64(1):170–188, September 2012.

\bibitem{DellM2014}
Holger Dell and Dieter Van~Melkebeek.
\newblock Satisfiability allows no nontrivial sparsification unless the
  polynomial-time hierarchy collapses.
\newblock {\em J. ACM}, 61(4), July 2014.
\newblock URL: \url{https://doi.org/10.1145/2629620}, \href
  {http://dx.doi.org/10.1145/2629620} {\path{doi:10.1145/2629620}}.

\bibitem{DingD16}
Guoli Ding and Stan Dziobiak.
\newblock Excluded-minor characterization of apex-outerplanar graphs.
\newblock {\em Graphs Comb.}, 32(2):583--627, 2016.
\newblock \href {http://dx.doi.org/10.1007/s00373-015-1611-9}
  {\path{doi:10.1007/s00373-015-1611-9}}.

\bibitem{Dinneen97}
Michael~J. Dinneen.
\newblock Too many minor order obstructions.
\newblock {\em J. Univers. Comput. Sci.}, 3(11):1199--1206, 1997.
\newblock \href {http://dx.doi.org/10.3217/jucs-003-11-1199}
  {\path{doi:10.3217/jucs-003-11-1199}}.

\bibitem{DinneenCF01}
Michael~J. Dinneen, Kevin Cattell, and Michael~R. Fellows.
\newblock Forbidden minors to graphs with small feedback sets.
\newblock {\em Discret. Math.}, 230(1-3):215--252, 2001.
\newblock \href {http://dx.doi.org/10.1016/S0012-365X(00)00083-2}
  {\path{doi:10.1016/S0012-365X(00)00083-2}}.

\bibitem{DinneenX02}
Michael~J. Dinneen and Liu Xiong.
\newblock Minor-order obstructions for the graphs of vertex cover 6.
\newblock {\em J. Graph Theory}, 41(3):163--178, 2002.
\newblock \href {http://dx.doi.org/10.1002/jgt.10059}
  {\path{doi:10.1002/jgt.10059}}.

\bibitem{DowneyF13}
Rodney~G. Downey and Michael~R. Fellows.
\newblock {\em Fundamentals of Parameterized Complexity}.
\newblock Texts in Computer Science. Springer, 2013.
\newblock \href {http://dx.doi.org/10.1007/978-1-4471-5559-1}
  {\path{doi:10.1007/978-1-4471-5559-1}}.

\bibitem{FleischnerGH74}
Herbert~J. Fleischner, Dennis~P. Geller, and Frank Harary.
\newblock Outerplanar graphs and weak duals.
\newblock {\em Journal of the Indian Mathematical Society}, 38, 1974.
\newblock URL:
  \url{http://www.informaticsjournals.com/index.php/jims/article/view/16694}.

\bibitem{FominLMPS16}
Fedor~V. Fomin, Daniel Lokshtanov, Neeldhara Misra, Geevarghese Philip, and
  Saket Saurabh.
\newblock Hitting forbidden minors: Approximation and kernelization.
\newblock {\em {SIAM} J. Discret. Math.}, 30(1):383--410, 2016.
\newblock \href {http://dx.doi.org/10.1137/140997889}
  {\path{doi:10.1137/140997889}}.

\bibitem{FominLMS12}
Fedor~V. Fomin, Daniel Lokshtanov, Neeldhara Misra, and Saket Saurabh.
\newblock Planar {$\mathcal{F}$}-deletion: {A}pproximation, kernelization and
  optimal {FPT} algorithms.
\newblock In {\em 53rd Annual {IEEE} Symposium on Foundations of Computer
  Science, {FOCS} 2012, New Brunswick, NJ, USA, October 20-23, 2012}, pages
  470--479. {IEEE} Computer Society, 2012.
\newblock \href {http://dx.doi.org/10.1109/FOCS.2012.62}
  {\path{doi:10.1109/FOCS.2012.62}}.

\bibitem{FominLSZ19}
Fedor~V. Fomin, Daniel Lokshtanov, Saket Saurabh, and Meirav Zehavi.
\newblock {\em Kernelization: Theory of Parameterized Preprocessing}.
\newblock Cambridge University Press, 2019.
\newblock \href {http://dx.doi.org/10.1017/9781107415157}
  {\path{doi:10.1017/9781107415157}}.

\bibitem{fomin2019kernelization}
Fedor~V Fomin, Daniel Lokshtanov, Saket Saurabh, and Meirav Zehavi.
\newblock {\em Kernelization: theory of parameterized preprocessing}.
\newblock Cambridge University Press, 2019.

\bibitem{Frati12}
Fabrizio Frati.
\newblock Straight-line drawings of outerplanar graphs in {$O(dn \log n)$}
  area.
\newblock {\em Computational Geometry}, 45(9):524--533, 2012.
\newblock \href {http://dx.doi.org/10.1016/j.comgeo.2010.03.007}
  {\path{doi:10.1016/j.comgeo.2010.03.007}}.

\bibitem{GiacomoLM16}
Emilio~Di Giacomo, Giuseppe Liotta, and Tamara Mchedlidze.
\newblock Lower and upper bounds for long induced paths in 3-connected planar
  graphs.
\newblock {\em Theor. Comput. Sci.}, 636:47--55, 2016.
\newblock URL: \url{https://doi.org/10.1016/j.tcs.2016.04.034}, \href
  {http://dx.doi.org/10.1016/j.tcs.2016.04.034}
  {\path{doi:10.1016/j.tcs.2016.04.034}}.

\bibitem{Giannopoulou2017}
Archontia~C. Giannopoulou, Bart M.~P. Jansen, Daniel Lokshtanov, and Saket
  Saurabh.
\newblock Uniform kernelization complexity of hitting forbidden minors.
\newblock {\em ACM Trans. Algorithms}, 13(3), March 2017.
\newblock URL: \url{https://doi.org/10.1145/3029051}, \href
  {http://dx.doi.org/10.1145/3029051} {\path{doi:10.1145/3029051}}.

\bibitem{gupta2019losing}
Anupam Gupta, Euiwoong Lee, Jason Li, Pasin Manurangsi, and Micha{\l}
  W{\l}odarczyk.
\newblock Losing treewidth by separating subsets.
\newblock In {\em Proceedings of the Thirtieth Annual ACM-SIAM Symposium on
  Discrete Algorithms}, pages 1731--1749. SIAM, 2019.
\newblock \href {http://dx.doi.org/10.1137/1.9781611975482.104}
  {\path{doi:10.1137/1.9781611975482.104}}.

\bibitem{Iwata17}
Yoichi Iwata.
\newblock {Linear-Time Kernelization for Feedback Vertex Set}.
\newblock In Ioannis Chatzigiannakis, Piotr Indyk, Fabian Kuhn, and Anca
  Muscholl, editors, {\em 44th International Colloquium on Automata, Languages,
  and Programming (ICALP 2017)}, volume~80 of {\em Leibniz International
  Proceedings in Informatics (LIPIcs)}, pages 68:1--68:14, Dagstuhl, Germany,
  2017. Schloss Dagstuhl--Leibniz-Zentrum fuer Informatik.
\newblock URL: \url{http://drops.dagstuhl.de/opus/volltexte/2017/7430}, \href
  {http://dx.doi.org/10.4230/LIPIcs.ICALP.2017.68}
  {\path{doi:10.4230/LIPIcs.ICALP.2017.68}}.

\bibitem{JansenP20}
Bart M.~P. Jansen and Astrid Pieterse.
\newblock Polynomial kernels for hitting forbidden minors under structural
  parameterizations.
\newblock {\em Theor. Comput. Sci.}, 841:124--166, 2020.
\newblock \href {http://dx.doi.org/10.1016/j.tcs.2020.07.009}
  {\path{doi:10.1016/j.tcs.2020.07.009}}.

\bibitem{JansenP18}
Bart M.~P. Jansen and Marcin Pilipczuk.
\newblock Approximation and kernelization for chordal vertex deletion.
\newblock {\em {SIAM} J. Discret. Math.}, 32(3):2258--2301, 2018.
\newblock \href {http://dx.doi.org/10.1137/17M112035X}
  {\path{doi:10.1137/17M112035X}}.

\bibitem{JoretPSST14}
Gwena{\"{e}}l Joret, Christophe Paul, Ignasi Sau, Saket Saurabh, and
  St{\'{e}}phan Thomass{\'{e}}.
\newblock Hitting and harvesting pumpkins.
\newblock {\em {SIAM} J. Discret. Math.}, 28(3):1363--1390, 2014.
\newblock \href {http://dx.doi.org/10.1137/120883736}
  {\path{doi:10.1137/120883736}}.

\bibitem{Lagergren98}
Jens Lagergren.
\newblock Upper bounds on the size of obstructions and intertwines.
\newblock {\em J. Comb. Theory, Ser. {B}}, 73(1):7--40, 1998.
\newblock \href {http://dx.doi.org/10.1006/jctb.1997.1788}
  {\path{doi:10.1006/jctb.1997.1788}}.

\bibitem{lee2019partitioning}
Euiwoong Lee.
\newblock Partitioning a graph into small pieces with applications to path
  transversal.
\newblock {\em Math. Program.}, 177(1–2):1–19, September 2019.
\newblock URL: \url{https://doi.org/10.1007/s10107-018-1255-7}, \href
  {http://dx.doi.org/10.1007/s10107-018-1255-7}
  {\path{doi:10.1007/s10107-018-1255-7}}.

\bibitem{LewisY80}
John~M. Lewis and Mihalis Yannakakis.
\newblock The node-deletion problem for hereditary properties is {NP}-complete.
\newblock {\em J. Comput. Syst. Sci.}, 20(2):219--230, 1980.
\newblock \href {http://dx.doi.org/10.1016/0022-0000(80)90060-4}
  {\path{doi:10.1016/0022-0000(80)90060-4}}.

\bibitem{Leydold1998MinimalCB}
J.~Leydold and P.~Stadler.
\newblock Minimal cycle bases of outerplanar graphs.
\newblock {\em Electron. J. Comb.}, 5, 1998.

\bibitem{MchedlidzeS11}
Tamara Mchedlidze and Antonios Symvonis.
\newblock Crossing-optimal acyclic hp-completion for outerplanar st-digraphs.
\newblock {\em J. Graph Algorithms Appl.}, 15(3):373--415, 2011.
\newblock \href {http://dx.doi.org/10.7155/jgaa.00231}
  {\path{doi:10.7155/jgaa.00231}}.

\bibitem{MorganF07}
Kerri Morgan and Graham Farr.
\newblock Approximation algorithms for the maximum induced planar and
  outerplanar subgraph problems.
\newblock {\em J. Graph Algorithms Appl.}, 11(1):165--193, 2007.
\newblock \href {http://dx.doi.org/10.7155/jgaa.00141}
  {\path{doi:10.7155/jgaa.00141}}.

\bibitem{PhilipRV10}
Geevarghese Philip, Venkatesh Raman, and Yngve Villanger.
\newblock A quartic kernel for pathwidth-one vertex deletion.
\newblock In Dimitrios~M. Thilikos, editor, {\em Graph Theoretic Concepts in
  Computer Science - 36th International Workshop, {WG} 2010, Zar{\'{o}}s,
  Crete, Greece, June 28-30, 2010 Revised Papers}, volume 6410 of {\em Lecture
  Notes in Computer Science}, pages 196--207, 2010.
\newblock \href {http://dx.doi.org/10.1007/978-3-642-16926-7_19}
  {\path{doi:10.1007/978-3-642-16926-7_19}}.

\bibitem{Poranen05}
Timo Poranen.
\newblock Heuristics for the maximum outerplanar subgraph problem.
\newblock {\em J. Heuristics}, 11(1):59--88, 2005.
\newblock \href {http://dx.doi.org/10.1007/s10732-005-6999-6}
  {\path{doi:10.1007/s10732-005-6999-6}}.

\bibitem{RobertsonS86}
Neil Robertson and Paul~D. Seymour.
\newblock Graph minors. {V}. {E}xcluding a planar graph.
\newblock {\em J. Comb. Theory, Ser. {B}}, 41(1):92--114, 1986.
\newblock \href {http://dx.doi.org/10.1016/0095-8956(86)90030-4}
  {\path{doi:10.1016/0095-8956(86)90030-4}}.

\bibitem{RueST12}
Juanjo Ru{\'{e}}, Konstantinos~S. Stavropoulos, and Dimitrios~M. Thilikos.
\newblock Outerplanar obstructions for a feedback vertex set.
\newblock {\em Eur. J. Comb.}, 33(5):948--968, 2012.
\newblock \href {http://dx.doi.org/10.1016/j.ejc.2011.09.018}
  {\path{doi:10.1016/j.ejc.2011.09.018}}.

\bibitem{SauST20}
Ignasi Sau, Giannos Stamoulis, and Dimitrios~M. Thilikos.
\newblock An {FPT}-algorithm for recognizing k-apices of minor-closed graph
  classes.
\newblock In Artur Czumaj, Anuj Dawar, and Emanuela Merelli, editors, {\em 47th
  International Colloquium on Automata, Languages, and Programming, {ICALP}
  2020, July 8-11, 2020, Saarbr{\"{u}}cken, Germany (Virtual Conference)},
  volume 168 of {\em LIPIcs}, pages 95:1--95:20. Schloss Dagstuhl -
  Leibniz-Zentrum f{\"{u}}r Informatik, 2020.
\newblock \href {http://dx.doi.org/10.4230/LIPIcs.ICALP.2020.95}
  {\path{doi:10.4230/LIPIcs.ICALP.2020.95}}.

\bibitem{SauST21}
Ignasi Sau, Giannos Stamoulis, and Dimitrios~M. Thilikos.
\newblock $k$-apices of minor-closed graph classes. {I. B}ounding the
  obstructions.
\newblock {\em CoRR}, abs/2103.00882, 2021.
\newblock \href {http://arxiv.org/abs/2103.00882} {\path{arXiv:2103.00882}}.

\bibitem{OpenProblems}
Saket Saurabh.
\newblock Open problems from the workshop on kernelization ({WorKer} 2019),
  2019.
\newblock URL: \url{https://www.youtube.com/watch?v=vCjG5zGjQr4}.

\bibitem{Syslo79}
Maciej~M. Syslo.
\newblock Characterizations of outerplanar graphs.
\newblock {\em Discret. Math.}, 26(1):47--53, 1979.
\newblock \href {http://dx.doi.org/10.1016/0012-365X(79)90060-8}
  {\path{doi:10.1016/0012-365X(79)90060-8}}.

\bibitem{Thomasse10}
St{\'{e}}phan Thomass{\'{e}}.
\newblock A {$4k^2$} kernel for feedback vertex set.
\newblock {\em {ACM} Trans. Algorithms}, 6(2):32:1--32:8, 2010.
\newblock \href {http://dx.doi.org/10.1145/1721837.1721848}
  {\path{doi:10.1145/1721837.1721848}}.

\bibitem{BevernMN12}
Ren{\'e} Van~Bevern, Hannes Moser, and Rolf Niedermeier.
\newblock Approximation and tidying—a problem kernel for s-plex cluster
  vertex deletion.
\newblock {\em Algorithmica}, 62(3):930--950, 2012.

\end{thebibliography}

\end{document}